%% file: BLP_algorithm_paper_DP_202504.tex
\documentclass[11pt]{article}
\usepackage[latin9]{inputenc}
\usepackage{geometry}
\usepackage[american,english]{babel}
\usepackage{array}
\usepackage{float}
\usepackage{url}
\usepackage{multirow}
\usepackage{amsmath}
\usepackage{amsthm}
\usepackage{amssymb}
\usepackage{graphicx}
\usepackage{setspace}
\usepackage[authoryear]{natbib}
\usepackage[unicode=true,pdfusetitle,
 bookmarks=true,bookmarksnumbered=false,bookmarksopen=false,
 breaklinks=false,pdfborder={0 0 0},pdfborderstyle={},backref=false,colorlinks=false]
 {hyperref}

\makeatletter

\providecommand{\tabularnewline}{\\}
\floatstyle{ruled}
\newfloat{algorithm}{tbp}{loa}
\providecommand{\algorithmname}{Algorithm}
\floatname{algorithm}{\protect\algorithmname}

\theoremstyle{plain}
\newtheorem{prop}{\protect\propositionname}
\theoremstyle{remark}
\newtheorem{rem}{\protect\remarkname}
\theoremstyle{plain}
\newtheorem{lem}{\protect\lemmaname}

\bibpunct{(}{)}{,}{a}{,}{,}
\usepackage{url}
\usepackage{csvsimple}
\theoremstyle{plain}

\usepackage{lscape}

\usepackage[bottom]{footmisc}

\makeatother

\addto\captionsamerican{\renewcommand{\algorithmname}{Algorithm}}
\addto\captionsamerican{\renewcommand{\lemmaname}{Lemma}}
\addto\captionsamerican{\renewcommand{\propositionname}{Proposition}}
\addto\captionsamerican{\renewcommand{\remarkname}{Remark}}
\addto\captionsenglish{\renewcommand{\lemmaname}{Lemma}}
\addto\captionsenglish{\renewcommand{\propositionname}{Proposition}}
\addto\captionsenglish{\renewcommand{\remarkname}{Remark}}
\providecommand{\lemmaname}{Lemma}
\providecommand{\propositionname}{Proposition}
\providecommand{\remarkname}{Remark}

\begin{document}
\title{Fast and simple inner-loop algorithms of static / dynamic BLP estimations}
\author{Takeshi Fukasawa\thanks{Waseda Institute for Advanced Study, Waseda University, 1-21-1, Nishiwaseda, Shinjuku, Tokyo, Japan; fukasawa3431@gmail.com\protect \\
I thank Naoshi Doi, Hiroshi Ohashi, Suguru Otani, Katsumi Shimotsu, and participants at JEMIOW Fall 2023, APIOC 2023, and IIOC 2024 for their helpful comments. Replication code of the numerical experiments in this article is available at \protect\url{https://github.com/takeshi-fukasawa/BLP_algorithm}.\protect \\
This study is supported \foreignlanguage{american}{JSPS KAKENHI Grant Number JP24K22629.}}}
\maketitle
\begin{abstract}
This study investigates computationally efficient inner-loop algorithms for estimating static/dynamic BLP models. It provides the following ideas for reducing the number of inner-loop iterations: (1). Add a term relating to the outside option share in the BLP contraction mapping; (2). Analytically represent the mean product utilities as a function of value functions and solve for value functions (for dynamic BLP); (3). Combine an acceleration method of fixed-point iterations, especially the Anderson acceleration. They are independent and easy to implement. This study shows the good performance of these methods using numerical experiments. 

\end{abstract}
\pagebreak{}

\section{Introduction}

\input{introduction_QE.tex}

\section{Literature\label{sec:Literature}}

\input{literature_review_QE.tex}

\section{Static BLP model\label{sec:Static-BLP-model}}

This section focuses on the standard static BLP models without a nest structure (Random Coefficient Logit (RCL) model). We discuss dynamic BLP models without nest structure in Section \ref{sec:Dynamic-BLP-model}, and static BLP models with nest structure (Random Coefficient Nested Logit (RCNL) model) in Appendix \ref{subsec:Extension:-RCNL-model}. Let $\left\Vert \cdot\right\Vert _{\infty}$ be the sup norm, and let $\left\Vert \cdot\right\Vert _{2}$ be the L2 norm. $\left\Vert \cdot\right\Vert $ also denotes a norm. $\epsilon_{\delta},\epsilon_{V},\epsilon$ denote values of the inner-loop tolerance levels.\footnote{\citet{conlon2020best} recommended setting the values between 1E-12 and 1E-14 when using 64-bit computers and applying the standard BLP contraction mapping.} 

\subsection{Model}

This section considers a static BLP model with random coefficients. This study proposes new algorithms whose convergence properties outperform those of the traditional BLP contraction mapping, and also generalizable to dynamic BLP models. Let consumer $i$'s utility when buying product $j$ be $v_{ij}=\delta_{j}+\mu_{ij}+\epsilon_{ij}$, and utility when buying nothing be $v_{i0}=\epsilon_{i0}$. $\delta_{j}$ denotes product $j$'s mean utility, and $\mu_{ij}$ denotes consumer $i$-specific utility of product $j$. $\epsilon$ denotes the idiosyncratic utility shocks. Let $\mathcal{J}$ be the set of products. Then, under the assumption that $\epsilon$ follows Gumbel distribution and that the observed market share data match the market share predicted by the model, the following equations hold:

\begin{eqnarray}
S_{j}^{(data)} & = & \int\frac{\exp\left(\delta_{j}+\mu_{ij}\right)}{1+\sum_{k\in\mathcal{J}}\exp\left(\delta_{k}+\mu_{ik}\right)}dP(i)\approx\sum_{i\in\mathcal{I}}w_{i}\frac{\exp\left(\delta_{j}+\mu_{ij}\right)}{1+\sum_{k\in\mathcal{J}}\exp\left(\delta_{k}+\mu_{ik}\right)},\label{eq:static_BLP_Sj}\\
S_{0}^{(data)} & = & \int\frac{1}{1+\sum_{k\in\mathcal{J}}\exp\left(\delta_{k}+\mu_{ik}\right)}dP(i)\approx\sum_{i\in\mathcal{I}}w_{i}\frac{1}{1+\sum_{k\in\mathcal{J}}\exp\left(\delta_{k}+\mu_{ik}\right)}.\label{eq:static_BLP_S0}
\end{eqnarray}
Here, $dP(i)$ denotes the density of consumer $i$. We use appropriate discretizations of the consumer types to solve the model numerically, when consumer types are continuously distributed.\footnote{See \citet{conlon2020best}for a review of the efficient discretization methods.} Let $\mathcal{I}$ be the set of discretized consumer types. $w_{i}$ denotes the fraction of type $i$ consumers, and $\sum_{i\in\mathcal{I}}w_{i}=1$ hold. If we use $I$ Monte Carlo draws to approximate the integral, $w_{i}=\frac{1}{I}$ holds. Note that equation (\ref{eq:static_BLP_S0}) can be derived from (\ref{eq:static_BLP_Sj}).

In empirical applications, the distribution of $\mu_{ij}$ is parameterized by $\theta_{n}$. For instance, $\mu_{ij}=\theta_{n}X_{j}\nu_{i}$ where $\nu_{i}\sim N(0,\Sigma)$. The mean product utility $\delta$ is represented as $\delta_{j}=X_{j}\theta_{l}+\xi_{j}$, where $\xi_{j}$ denotes the unobserved product characteristics of product $j$. Econometricians seek the values of $(\theta_{l},\theta_{n})$ minimizing the GMM objective $G(\theta_{l},\theta_{n})WG(\theta_{l},\theta_{n})$ where $G(\theta_{l},\theta_{n})\equiv Z\xi(\theta_{l},\theta_{n})=Z\left(\delta(\theta_{n})-X\theta_{l}\right)$, $W$ denotes the weight matrix, and $Z$ denotes the appropriate instrumental variables. $\delta(\theta_{n})$ denotes the values of $\delta$ satisfying equations (\ref{eq:static_BLP_Sj}) and (\ref{eq:static_BLP_S0}). Algorithm \ref{alg:Full-estimation-algorithm-static-BLP} shows the formal steps.

\begin{algorithm}[H]
\begin{enumerate}
\item Inner loop: Given $\theta_{n}$,
\begin{enumerate}
\item Compute $\mu$ using $\theta_{n}$
\item Solve for $\delta$ satisfying equations (\ref{eq:static_BLP_Sj}) and (\ref{eq:static_BLP_S0}) given $\mu$. Let the solution be $\delta(\theta_{n})$.
\item Compute linear parameters $\theta_{l}^{D}$ that minimize the GMM objective given $\theta_{n}^{D}$. Let $\theta_{l}^{*}(\theta_{n})$ be the solution to the minimization problem, and let $m(\theta_{n})\equiv G^{\prime}\left(\theta_{l}^{*}(\theta_{n}),\theta_{n}\right)WG\left(\theta_{l}^{*}(\theta_{n}),\theta_{n}\right)$.
\end{enumerate}
\item Outer loop: Search for the value of $\theta_{n}$ minimizing the GMM objective $m(\theta_{n})$.
\end{enumerate}
\caption{Full estimation algorithm of the static BLP model\label{alg:Full-estimation-algorithm-static-BLP}}
\end{algorithm}

Below, we focus on Step 1(b) of Algorithm \ref{alg:Full-estimation-algorithm-static-BLP}, namely, how to efficiently solve for $\delta$ satisfying the constraints (\ref{eq:static_BLP_Sj}) and (\ref{eq:static_BLP_S0}) given $\mu$ and observed market shares $S^{(data)}$.\footnote{\citet{conlon2020best} discuss the details of overall estimation procedures of static BLP models. Note that when no consumer heterogeneity exists, we can recover the value of $\delta$ by $\delta_{j}=\log\left(S_{j}^{(data)}\right)-\log\left(S_{0}^{(data)}\right)$, and we can estimate the linear parameters $\theta_{l}$ by a linear GMM without solving the fixed-point problem, as discussed in \citet{berry1994estimating}.}

\subsection{Mappings of mean product utility $\delta$\label{subsec:Mappings-on-delta-static-BLP}}

To solve for $\delta$ given $\mu$ and observed market shares $S^{(data)}\equiv\left(S_{j}^{(data)}\right)_{j\in\mathcal{J}}$, \citet{berry1995automobile} proposed iterating the update $\delta_{j}^{(n+1)}=\delta_{j}^{(n)}+\log\left(S_{j}^{(data)}\right)-\log\left(s_{j}\left(\delta^{(n)}\right)\right)$ with appropriate initial values $\delta^{(0)}$ until convergence. Here, $s_{j}(\delta)\equiv\sum_{i}w_{i}\frac{\exp\left(\delta_{j}+\mu_{ij}\right)}{1+\sum_{k\in\mathcal{J}}\exp\left(\delta_{k}+\mu_{ik}\right)}$ denotes the market share of product $j$ predicted by the structural model. The mapping is known as the BLP contraction mapping.

Motivated by the BLP contraction mapping, we define a mapping $\Phi^{\delta,\gamma}:B_{\delta}\rightarrow B_{\delta}$ for $\gamma\in\mathbb{R}$ such that\footnote{Although the mapping $\Phi_{j}^{\delta,\gamma}$ implicitly depends on the values of $S^{(data)}$ and $\mu$, we omit expressing them as arguments of the mapping to simplify the exposition.}:
\begin{eqnarray*}
\Phi_{j}^{\delta,\gamma}\left(\delta\right) & \equiv & \delta_{j}+\left[\log\left(S_{j}^{(data)}\right)-\log\left(s_{j}(\delta)\right)\right]-\gamma\left[\log\left(S_{0}^{(data)}\right)-\log\left(s_{0}(\delta)\right)\right],
\end{eqnarray*}
where $s_{0}(\delta)=1-\sum_{k\in\mathcal{J}}s_{k}(\delta)$. $B_{\delta}$ denotes the space of $\delta\equiv\left\{ \delta_{j}\right\} _{j\in\mathcal{J}}$. Obviously, $\Phi^{\delta,\gamma=0}$ matches the traditional BLP contraction mapping.

Proposition \ref{prop:sol_delta_mapping_static_BLP} ensures that we can find the solution $\delta$ satisfying $S_{j}^{(data)}=s_{j}(\delta)\ (j\in\mathcal{J})$ by alternatively solving the fixed-point constraint $\delta=\Phi^{\delta,\gamma\geq0}\left(\delta\right)$.
\begin{prop}
\label{prop:sol_delta_mapping_static_BLP}Solution of $\delta=\Phi^{\delta,\gamma\geq0}\left(\delta\right)$ satisfies $S_{j}^{(data)}=s_{j}(\delta)\ \forall j\in\mathcal{J}$.
\end{prop}
Clearly, the solution of $S_{j}^{(data)}=s_{j}(\delta)\ (j\in\mathcal{J})$ satisfies $\delta_{j}=\Phi_{j}^{\delta,\gamma}\left(\delta\right)=\delta_{j}+\left[\log\left(S_{j}^{(data)}\right)-\log\left(s_{j}(\delta)\right)\right]-\gamma\left[\log\left(S_{0}^{(data)}\right)-\log\left(s_{0}(\delta)\right)\right]\ j\in\mathcal{J}$. However, its converse is not trivial, even without formal proof. Proposition \ref{prop:sol_delta_mapping_static_BLP} formally shows that the converse actually holds under $\gamma\geq0$.

Algorithm \ref{alg:static_BLP_delta} shows the algorithm to solve for $\delta$ using the mapping $\Phi^{\delta,\gamma}$.

\begin{algorithm}[H]
Set initial values of $\delta^{(0)}$. Iterate the following $(n=0,1,2,\cdots)$:
\begin{enumerate}
\item Compute $\delta^{(n+1)}=\Phi^{\delta,\gamma}(\delta^{(n)})$
\item Exit the iteration if $\left\Vert \delta^{(n+1)}-\delta^{(n)}\right\Vert <\epsilon_{\delta}$
\end{enumerate}
\caption{Inner loop algorithm of static BLP using $\Phi^{\delta,\gamma}$\label{alg:static_BLP_delta}}
\end{algorithm}

Although, in principle, we can choose any real value of $\gamma$, we mainly consider $\gamma=1$ because of its good convergence properties. As discussed in detail in Appendix \ref{subsec:Convergence-properties-mappings}, the consumer heterogeneity size largely affects the convergence speed of $\Phi^{\delta,\gamma=1}$. As consumer heterogeneity decreases, the convergence speed increases. In the absence of consumer heterogeneity, the iteration immediately converges after one iteration, regardless of the choice of initial values $\delta^{(0)}$. When no consumer heterogeneity exists, $\delta_{j}=\log(S_{j}^{(data)})-\log(S_{0}^{(data)})$ holds, as shown in \citet{berry1994estimating}. Regarding the mapping $\Phi^{\delta,\gamma=1}$, $\Phi_{j}^{\delta,\gamma}\left(\delta\right)=\log\left(S_{j}^{(data)}\right)-\log\left(S_{0}^{(data)}\right)$ holds, and the output of $\Phi^{\delta,\gamma}$ is equal to the true $\delta$ for any input $\delta$. 

In general, there is no guarantee that $\Phi^{\delta,\gamma=1}$ is a contraction. It implies there is no guarantee that the iteration with $\delta_{j}^{(n+1)}=\delta_{j}^{(n)}+\left[\log\left(S_{j}^{(data)}\right)-\log\left(s_{j}(\delta^{(n)})\right)\right]-\left[\log\left(S_{0}^{(data)}\right)-\log\left(s_{0}(\delta^{(n)})\right)\right]$ converges. In fact, we can construct a simple numerical example where $\Phi_{j}^{\delta,\gamma=1}$ is not a contraction in an extreme setting where consumer heterogeneity is too large, as shown in the Supplemental Appendix. However, the iteration converges when combining the spectral/SQUAREM algorithm, which we discuss in detail in Section \ref{sec:Acceleration-methods}. In addition, as shown in Section \ref{sec:Numerical-Experiments}, the iteration always converges and sometimes leads to drastic improvement of convergence speed in the standard settings of Monte Carlo simulations experimented in the previous studies (\citealp{dube2012improving}; \citealp{lee2015computationally}) and the estimations using datasets used by \citet{nevo2001measuring} and Berry et al. \citeyearpar{berry1995automobile,berry1999voluntary}. As mapping $\Phi^{\delta,\gamma=1}$ is straightforward to implement once coding the standard BLP contraction mapping, applying the mapping $\Phi^{\delta,\gamma=1}$ is worth considering in the static BLP estimation. 

Please note that the global convergence of iterations using $\Phi^{\delta,\gamma=1}$can be ensured by adding a few lines in the programming code, by using the fact that $\Phi^{\delta,\gamma=0}$, traditional BLP contraction mapping, is a contraction. For details, see Appendix \ref{subsec:Global-convergence}. Conservative practitioners can consider this process for convergence. 

Concerning the convergence speed of the BLP contraction mapping ($\Phi^{\delta,\gamma=0}$), it is slow when the outside option share is small (\citealp{dube2012improving}). In contrast, regarding the new mapping $\Phi^{\delta,\gamma=1}$, the convergence speed is less sensitive to the outside option share, as formally discussed in Appendix \ref{subsec:Convergence-properties-mappings} and numerically shown in Section \ref{sec:Numerical-Experiments}. 
\begin{rem}
\label{rem:additional_cost_delta_map}Applying the mapping $\Phi^{\delta,\gamma=1}$ requires additionally subtracting the term $\left[\log\left(S_{0}^{(data)}\right)-\log\left(s_{0}(\delta)\right)\right]$, compared to the case with traditional BLP contraction mapping $\Phi^{\delta,\gamma=0}$. It implies the computational cost of applying a mapping per iteration increases when we apply $\Phi^{\delta,\gamma=1}$ rather than $\Phi^{\delta,\gamma=0}$. Nevertheless, the additional computational cost is generally much smaller than the one for applying $\Phi^{\delta,\gamma=0}$. Regarding $\Phi^{\delta,\gamma=0}$, computationally costly part is the computation of $s_{j}(\delta)=\sum_{i}w_{i}\frac{\exp\left(\delta_{j}+\mu_{ij}\right)}{1+\sum_{k\in\mathcal{J}}\exp\left(\delta_{k}+\mu_{ik}\right)}$, which requires the following operations: (1). Summation of $\delta_{j}$ and $\mu_{ij}$ for all $i\in\mathcal{I}$ and $j\in\mathcal{J}$; (2). Computation of the exponential of $\delta_{j}+\mu_{ij}$; (3). Summation of a $|\mathcal{J}|$-dimensional vector $\left\{ \exp\left(\delta_{j}+\mu_{ij}\right)\right\} _{j\in\mathcal{J}}$ for all $i\in\mathcal{I}$; (4). Computation of $\frac{1}{1+\sum_{k\in\mathcal{J}}\exp\left(\delta_{k}+\mu_{ik}\right)}$ for all $i\in\mathcal{I}$; (5). Computation of $\frac{\exp\left(\delta_{j}+\mu_{ij}\right)}{1+\sum_{k\in\mathcal{J}}\exp\left(\delta_{k}+\mu_{ik}\right)}=\frac{1}{1+\sum_{k\in\mathcal{J}}\exp\left(\delta_{k}+\mu_{ik}\right)}\cdot\exp\left(\delta_{j}+\mu_{ij}\right)$ for all $i\in\mathcal{I}$ and $j\in\mathcal{J}$; (6). Summation of $\left\{ w_{i}\frac{\exp\left(\delta_{j}+\mu_{ij}\right)}{1+\sum_{k\in\mathcal{J}}\exp\left(\delta_{k}+\mu_{ik}\right)}\right\} _{i\in\mathcal{I}}$ for all $j\in\mathcal{J}$.\footnote{See also \citet{brunner2017reliable} for efficient computation.} In contrast, computing a scalar $s_{0}(\delta)=1-\sum_{k\in\mathcal{J}}s_{k}(\delta)$ requires the summation of $\sum_{k\in\mathcal{J}}s_{k}(\delta)$ using already computed $s_{k}(\delta)$, and it is much less expensive than computing $s_{j}(\delta)$. Consequently, applying the mapping $\Phi^{\delta,\gamma=1}$ is less expensive to apply compared to the use of the BLP contraction mapping $\Phi^{\delta,\gamma=0}$.
\end{rem}

\subsection{Mappings of value functions $V$\label{subsec:Mappings-on-V-static-BLP}}

Although it may seem intuitive to consider mappings of $\delta$ to solve for $\delta$, we can alternatively solve for consumers' value functions $V$ using a mapping of $V$, and then recover $\delta$ using an analytical formula. Here, consumer $i$'s value function, or the inclusive value, is defined as follows:

\begin{eqnarray}
V_{i} & \equiv & \log\left(1+\sum_{k\in\mathcal{J}}\exp\left(\delta_{k}+\mu_{ik}\right)\right).\label{eq:static_BLP_V}
\end{eqnarray}

Below, we discuss the approach solving for $V$. It seems that applying the approach is not necessary, when only considering the static BLP models. The approach directly solving for $\delta$ is more intuitive and standard. However, when examining dynamic BLP models, introducing the approach solving for $V$ is attractive, because it closely relates with the value function iterations typically used in solving dynamic models, as discussed in Section \ref{sec:Dynamic-BLP-model}.

\citet{kalouptsidi2012market} first proposed solving for consumer-type specific variables in the inner loop of static BLP estimations. However, the original algorithms have convergence problems. The following algorithms overcome the issues by choosing appropriate mappings that have dualistic relations with the mappings of $\delta$. For further discussion on the \citet{kalouptsidi2012market}'s algorithms, see the Supplemental Appendix. 

First, because $S_{j}^{(data)}=\exp\left(\delta_{j}\right)\cdot\sum_{i}w_{i}\frac{\exp\left(\mu_{ij}\right)}{1+\sum_{k\in\mathcal{J}}\exp\left(\delta_{k}+\mu_{ik}\right)}=\exp\left(\delta_{j}\right)\cdot\sum_{i}w_{i}\frac{\exp\left(\mu_{ij}\right)}{\exp\left(V_{i}\right)}$ holds by (\ref{eq:static_BLP_Sj}) and (\ref{eq:static_BLP_V}), the following equation should hold:

\begin{eqnarray*}
\delta_{j} & = & \log\left(S_{j}^{(data)}\right)-\log\left(\sum_{i}w_{i}\exp\left(\mu_{ij}-V_{i}\right)\right).
\end{eqnarray*}

Based on this formula, we define a function $\Phi^{V,\gamma}:B_{V}\rightarrow B_{V}$ for $\gamma\in\mathbb{R}$ such that:

\begin{eqnarray*}
\Phi^{V,\gamma}(V) & \equiv & \iota_{\delta\rightarrow V}\left(\iota_{V\rightarrow\delta}^{\gamma}\left(V\right)\right),
\end{eqnarray*}
where $\iota_{\delta\rightarrow V}:B_{\delta}\rightarrow B_{V}$ is a mapping such that:

\begin{eqnarray*}
\iota_{\delta\rightarrow V,i}\left(\delta\right) & \equiv & \log\left(1+\sum_{j\in\mathcal{J}}\exp\left(\delta_{j}+\mu_{ij}\right)\right),
\end{eqnarray*}
and $\iota_{V\rightarrow\delta}^{\gamma}:B_{V}\rightarrow B_{\delta}$ is a mapping such that: 

\begin{eqnarray*}
\iota_{V\rightarrow\delta,j}^{\gamma}\left(V\right) & \equiv & \log\left(S_{j}^{(data)}\right)-\log\left(\sum_{i}w_{i}\exp\left(\mu_{ij}-V_{i}\right)\right)-\gamma\log\left(\frac{S_{0}^{(data)}}{\sum_{i}w_{i}\exp(-V_{i})}\right).
\end{eqnarray*}

$B_{V}$ denotes the space of $V\equiv\left\{ V_{i}\right\} _{i\in\mathcal{I}}$. Note that $\Phi_{i}^{V,\gamma}\left(V\right)=\log\left(1+\left(\sum_{j\in\mathcal{J}}S_{j}^{(data)}\frac{\exp(\mu_{ij})}{\sum_{i}w_{i}\exp(\mu_{ij})\exp(-V_{i})}\right)\cdot\left(\frac{\sum_{i}w_{i}\exp(-V_{i})}{S_{0}^{(data)}}\right)^{\gamma}\right)$ holds.

The following proposition shows the validity of using the mapping $\Phi^{V,\gamma\geq0}$:
\begin{prop}
\label{prop:sol_V_mapping_static_BLP}$\delta$ such that $V=\Phi^{V,\gamma\geq0}(V),\ \delta=\iota_{V\rightarrow\delta}^{\gamma\geq0}(V)$ satisfies $S_{j}^{(data)}=s_{j}(\delta)\ \forall j\in\mathcal{J}$. 
\end{prop}
Hence, to solve for $\delta$, two algorithms can be considered. The first is to apply $\Phi^{\delta,\gamma}$ iteratively and solve for $\delta$ with appropriate initial values of $\delta$, as discussed in Section \ref{subsec:Mappings-on-delta-static-BLP}. The second one is to apply $\Phi^{V,\gamma}$ iteratively with appropriate initial values of $V$, and obtain $\delta$, which is analytically computed using $V$ in the iteration. Algorithm \ref{alg:static_BLP_V} outlines this second approach.
\begin{algorithm}[H]
Set initial values of $V^{(0)}$. Iterate the following $(n=0,1,2,\cdots)$:
\begin{enumerate}
\item Compute $\delta^{(n)}=\iota_{V\rightarrow\delta}^{\gamma}\left(V^{(n)}\right)$
\item Update $V$ by $V^{(n+1)}=\iota_{\delta\rightarrow V}\left(\delta^{(n)}\right)$
\item Exit the iteration if $\left\Vert V^{(n+1)}-V^{(n)}\right\Vert <\epsilon_{V}$
\end{enumerate}
\caption{Inner loop algorithm of static BLP using $\Phi^{V,\gamma}$\label{alg:static_BLP_V}}
\end{algorithm}

\begin{rem}
As in the case of $\Phi^{\delta,\gamma=1}$ and $\Phi^{\delta,\gamma=0}$(cf. Remark \ref{rem:additional_cost_delta_map}), the computational cost of applying $\Phi^{V,\gamma=1}$ is similar to $\Phi^{V,\gamma=0}$, because computing a scalar $s_{0}(V)=\sum_{i}w_{i}\exp(-V_{i})$ is much less costly than computing $\sum_{j\in\mathcal{J}}\exp\left(\delta_{j}+\mu_{ij}\right)$.
\end{rem}
\begin{rem}
If we define a mapping $\Psi_{V\delta\rightarrow V,i}^{\gamma}\left(V,\delta\right)=\log\left(1+\left(\sum_{j\in\mathcal{J}}\exp(\delta_{j}+\mu_{ij})\right)\cdot\left(\frac{\sum_{i}w_{i}\exp(-V_{i})}{S_{0}^{(data)}}\right)^{\gamma}\right)$, $\Phi^{V,\gamma}(V)=\Psi_{V\delta\rightarrow V}^{\gamma}\left(V,\iota_{V\rightarrow\delta}\left(V\right)\right)$ holds. Hence, we can alternatively compute $\delta^{(n)}=\iota_{V\rightarrow\delta}\left(V^{(n)}\right)$, and update $V$ by $V^{(n+1)}=\Psi_{V\delta\rightarrow V}^{\gamma}\left(V^{(n)},\delta^{(n)}\right)$. Such algorithm is used in the dynamic BLP model discussed in the next section.
\end{rem}
We can easily prove a dualistic relationship between $\Phi^{\delta,\gamma}$ and $\Phi^{V,\gamma}$. As shown in the Monte Carlo simulation (Section \ref{sec:Numerical-Experiments}), their convergence speed is generally the same, provided that they take a common value $\gamma$. The following proposition is a formal statement:
\begin{prop}
\label{prop:duality_mapping}(Duality of mappings of $\delta$ and $V$) The following equations hold for all $\gamma\in\mathbb{R}$:

\begin{eqnarray*}
\Phi^{V,\gamma} & = & \iota_{\delta\rightarrow V}\circ\iota_{V\rightarrow\delta}^{\gamma},\\
\Phi^{\delta,\gamma} & = & \iota_{V\rightarrow\delta}^{\gamma}\circ\iota_{\delta\rightarrow V}.
\end{eqnarray*}
\end{prop}
Figure \ref{fig:Duality-mapping} graphically shows the relationships between the mappings of $\delta$ and $V$. Here, $\delta$-($\gamma$) denotes the algorithm using the mapping $\Phi^{\delta,\gamma}$, and $V$-($\gamma$) denotes the algorithm using the mapping $\Phi^{V,\gamma}$.

\begin{figure}[H]
\caption{Duality between the mappings of $\delta$ and $V$ \label{fig:Duality-mapping}}

\centering{}\includegraphics[scale=0.7]{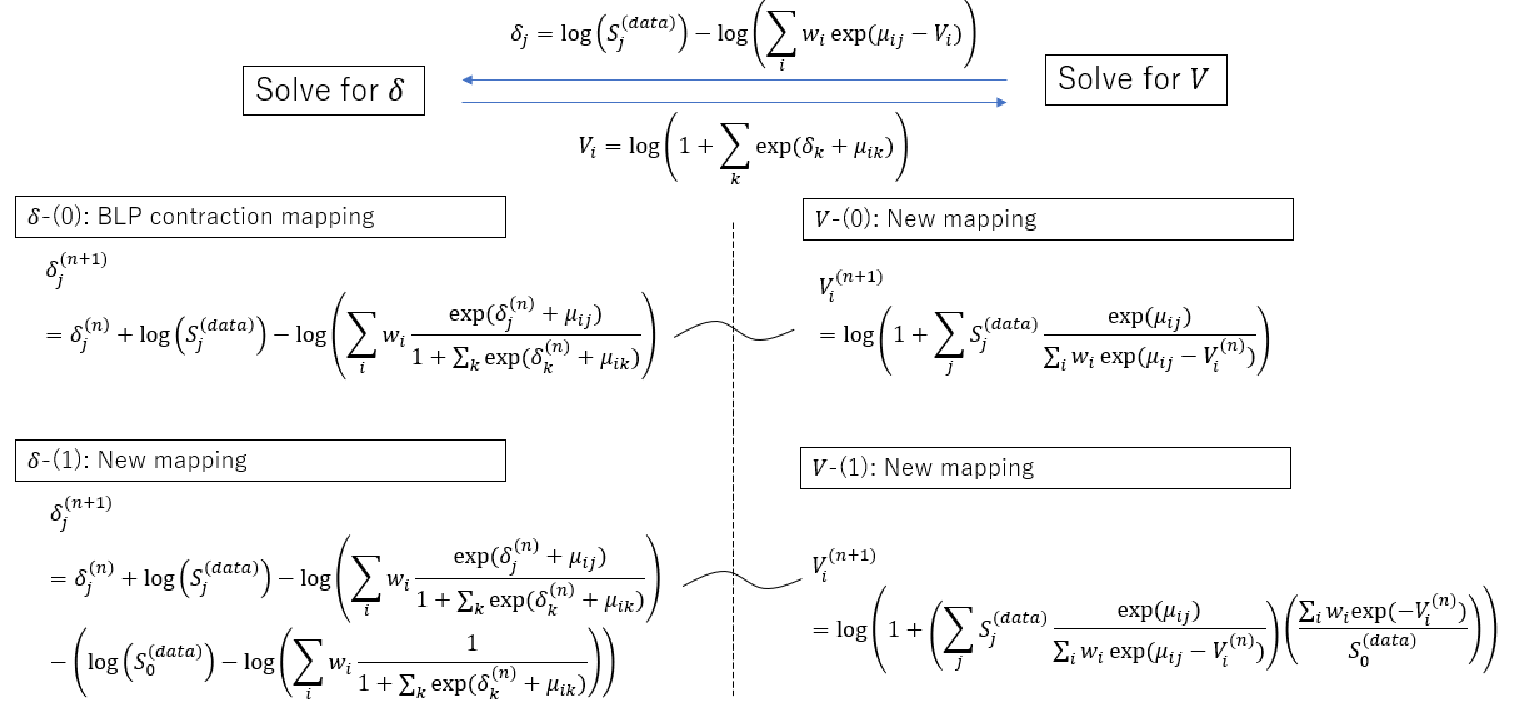}
\end{figure}

In Appendix \ref{subsec:Convergence-properties-mappings}, we discuss the convergence properties of the mapping $\Phi^{V,\gamma}$. As in the case of $\Phi^{\delta,\gamma}$, the size of the consumer heterogeneity affects the convergence speed. 

\section{Dynamic BLP model\label{sec:Dynamic-BLP-model}}

\subsection{Model}

Next, we consider a dynamic BLP model, where consumers' forward-looking decisions are integrated into the static BLP models. Similar to the static BLP models, we assume that consumer types are appropriately discretized to solve the model numerically. Let $\mathcal{I}$ be the set of consumer types, and let $w_{i}$ be the fraction of the type $i$ consumers.

The discounted sum of utility of consumer $i$ at state $(x_{it},\Omega_{t})$ when buying product $j$ at time $t$ is:

\begin{eqnarray*}
v_{ijt}(x_{it},\Omega_{t}) & = & \delta_{jt}+\mu_{ijt}(x_{it},\Omega_{t})+\beta E_{t}\left[V_{it+1}(x_{it+1},\Omega_{t+1})|x_{it},\Omega_{t},a_{it}=j\right]+\epsilon_{ijt}.
\end{eqnarray*}

The discounted sum of utility of consumer $i$ at state $(x_{it},\Omega_{t})$ when not buying any product at time $t$ is:

\begin{eqnarray*}
v_{i0t}(x_{it},\Omega_{t}) & = & \mu_{i0t}(x_{it},\Omega_{t})+\beta E_{t}\left[V_{it+1}(x_{it+1},\Omega_{t+1})|x_{it},\Omega_{t},a_{it}=0\right]+\epsilon_{i0t}.
\end{eqnarray*}
Here, $\beta$ denotes consumers' discount factor. $x_{it}$ denotes individual state variables, such as durable product holdings (for durable goods) or the product purchased in the previous period (for goods with switching costs). $\Omega_{t}$ denotes market-level state variables, including characteristics of products sold in the market. $\Omega_{t}$ may include the mean product utilities $\delta$. $\epsilon$ denotes i.i.d. idiosyncratic utility shocks, and $V$ denotes the (integrated) value function, defined by $V_{it}(x_{it},\Omega_{t})\equiv E_{\epsilon}\left[\max_{j\in\mathcal{A}_{t}(x_{it})}v_{ijt}(x_{it},\Omega_{t})\right]$, where $E_{\epsilon}$ denotes the expectation operator concerning $\epsilon$. $E_{t}$ denotes the expectation operator concerning the realizations of the future states $(x_{it+1},\Omega_{t+1})$.

Assuming that $\epsilon$ follows i.i.d. mean zero type-I extreme value distribution and that the observed market shares match the predicted market shares, $\delta$ and $V$ satisfy the following equations:

\begin{eqnarray}
S_{jt}^{(data)} & = & s_{jt}(V,\delta),\label{eq:S_j_dynamic_BLP}\\
S_{0t}^{(data)} & = & s_{0t}(V,\delta),\label{eq:S_0_dynamic_BLP}\\
V_{it}(x_{it},\Omega_{t}) & = & \log\left(\exp\left(\mu_{i0t}(x_{it},\Omega_{t})+\beta E_{t}\left[V_{it+1}(x_{it+1},\Omega_{t+1})|x_{it},\Omega_{t},a_{it}=0\right]\right)+\right.\label{eq:V_dynamic_BLP}\\
 &  & \ \ \ \ \left.\sum_{j\in\mathcal{A}_{t}(x_{it})-\{0\}}\exp\left(\delta_{jt}+\mu_{ijt}(x_{it},\Omega_{t})+\beta E_{t}\left[V_{it+1}(x_{it+1},\Omega_{t+1})|x_{it},\Omega_{t},a_{it}=j\right]\right)\right),\nonumber 
\end{eqnarray}
where 
\begin{eqnarray*}
s_{jt}(V,\delta) & \equiv & \sum_{i\in\mathcal{I}}w_{i}\sum_{x_{it}\in\chi}Pr_{it}(x_{it})\cdot\frac{\exp\left(\delta_{jt}+\mu_{ijt}(x_{it},\Omega_{t})+\beta E_{t}\left[V_{it+1}(x_{it+1},\Omega_{t+1})|x_{it},\Omega_{t},a_{it}=j\right]\right)}{\exp\left(V_{it}(x_{it},\Omega_{t})\right)},\\
s_{0t}(V,\delta) & \equiv & \sum_{i\in\mathcal{I}}w_{i}\sum_{x_{it}\in\chi}Pr_{it}(x_{it})\cdot\frac{\exp\left(\mu_{i0t}(x_{it},\Omega_{t})+\beta E_{t}\left[V_{it+1}(x_{it+1},\Omega_{t+1})|x_{it},\Omega_{t},a_{it}=0\right]\right)}{\exp\left(V_{it}(x_{it},\Omega_{t})\right)}.
\end{eqnarray*}

Here, $S_{jt}^{(data)}$ denotes product $j$'s market share at time $t$, and $S_{0t}^{(data)}$ denotes the fraction of consumers not buying anything at time $t$. $Pr_{it}(x_{it})$ denotes the fraction of type $i$ consumers at state $x_{it}$. For example, for durables, $Pr_{it}(x_{it}=\emptyset)$ denotes the fraction of type $i$ consumers not holding any durable product at time $t$, if we let $x_{it}=\emptyset$ be the state where consumers do not hold any durable products. $\mathcal{A}_{t}(x_{it})\subset\mathcal{J}_{t}\cup\{0\}$ denotes the consideration set of consumers at state $x_{it}$ at time $t$. For instance, in the durable goods example above, $\mathcal{A}_{t}(x_{it}\neq\emptyset)=\{0\}$ holds if consumers holding any durable product do not buy additional products. Here, $\mathcal{J}_{t}$ denotes the set of products sold at time $t$. The term $\frac{\exp\left(\delta_{jt}+\mu_{ijt}(x_{it},\Omega_{t})+\beta E_{t}\left[V_{it+1}(x_{it+1},\Omega_{t+1})|x_{it},\Omega_{t},a_{it}=j\right]\right)}{\exp\left(V_{it}(x_{it},\Omega_{t})\right)}\equiv s_{ijt}^{(ccp)}(x_{it},\Omega_{t})$ denotes the conditional choice probability (CCP) of type $i$ consumers at state $(x_{it},\Omega_{t})$ choosing product $j$ at time $t$. $\frac{\exp\left(\mu_{i0t}(x_{it},\Omega_{t})+\beta E_{t}\left[V_{it+1}(x_{it+1},\Omega_{t+1})|x_{it},\Omega_{t},a_{it}=0\right]\right)}{\exp\left(V_{it}(x_{it},\Omega_{t})\right)}\equiv s_{i0t}^{(ccp)}(x_{it},\Omega_{t})$ denotes the CCP of type $i$ consumers at state $(x_{it},\Omega_{t})$ choosing the outside option at time $t$.

\subsection{Algorithm}

The overall estimation algorithm is similar to the static BLP: As for dynamic BLP, we should solve for $\delta$ satisfying equations (\ref{eq:S_j_dynamic_BLP})-(\ref{eq:V_dynamic_BLP}) given $\mu$.\footnote{Unlike the static BLP models, there is no guarantee that $\delta$ satisfying (\ref{eq:S_j_dynamic_BLP})-(\ref{eq:V_dynamic_BLP}) is unique. If multiple $\delta$ satisfying these equations exist, we should adopt $\delta$ with the lowest GMM objective.} As in the static BLP, we focus on efficient algorithms for solving for $\delta$ satisfying (\ref{eq:S_j_dynamic_BLP})-(\ref{eq:V_dynamic_BLP}) given $\mu$.

Here, we examine a nonstationary environment, where consumers have perfect foresight on the state transition of $\Omega_{t}$, and $\Omega_{t}$ remains constant after the terminal period $T$.\footnote{Similar specification was considered in \citealp{gowrisankaran2012dynamics} and \citet{conlon2012dynamic}. Analogous specification was used in \citet{igami2017estimating} in the context of dynamic discrete games. Besides, regarding consumer expectations except for the terminal period, we can relax the assumption of perfect foresight to rational expectations where consumer expectations are on average correct. For details, see the discussion in \citet{fukasawa2024lightbulb}. The assumption of rational expectation has been applied in \citet{kalouptsidi2020linear} and empirical studies cited therein. } Although this specification differs from the one under inclusive value sufficiency, widely applied in earlier studies (e.g., \citealp{gowrisankaran2012dynamics}), the algorithm under the current specification is relatively simple and is the basis of algorithms under alternative specifications. Besides, we consider a setting where the values of $\left\{ Pr_{it}(x_{it})\right\} _{x_{it}\in\chi,t=1,\cdots,T}$ are known, to clarify the algorithm's essence. In Section \ref{subsec:Numerical-Dynamic-BLP-model}, we address cases where the values of $\left\{ Pr_{it}(x_{it})\right\} _{x_{it}\in\chi,t=1,\cdots,T}$ are unknown. 

Equations (\ref{eq:S_j_dynamic_BLP}), (\ref{eq:S_0_dynamic_BLP}), and (\ref{eq:V_dynamic_BLP}) imply $V$ is the solution of $V=\Phi^{V,\gamma}(V)\ \ \forall\gamma\in\mathbb{R}$, where 

\begin{eqnarray}
\Phi^{V,\gamma}(V) & \equiv & \Psi_{V\delta\rightarrow V}^{\gamma}\left(V,\iota_{V\rightarrow\delta}\left(V\right)\right).\label{eq:Phi_V_dynamic_BLP}
\end{eqnarray}

Here, we define a function $\Psi_{V\delta\rightarrow V}^{\gamma}:B_{V}\times B_{\delta}\rightarrow B_{V}$ such that:

\begin{eqnarray*}
 &  & \Psi_{V\delta\rightarrow V,it}^{\gamma}(V,\delta)\\
 & \equiv & \log\left(\exp\left(\mu_{i0t}(x_{it},\Omega_{t}^{(data)})+\beta E_{t}\left[V_{it+1}(x_{it+1},\Omega_{t+1}^{(data)})|x_{it},a_{it}=0\right]\right)+\right.\\
 &  & \left.\sum_{j\in\mathcal{A}_{t}(x_{it})-\{0\}}\exp\left(\delta_{jt}+\mu_{ijt}(x_{it},\Omega_{t}^{(data)})+\beta E_{t}\left[V_{it+1}(x_{it+1},\Omega_{t+1}^{(data)})|x_{it},a_{it}=j\right]\right)\cdot\left(\frac{s_{0t}(V,\delta)}{S_{0t}^{(data)}}\right)^{\gamma}\right).
\end{eqnarray*}

We also define a function $\iota_{V\rightarrow\delta}:B_{V}\rightarrow B_{\delta}$ such that:

{\small{}
\begin{eqnarray*}
 &  & \iota_{V\rightarrow\delta,jt}(V)\\
 & \equiv & \log\left(S_{jt}^{(data)}\right)-\log\left(\sum_{i\in\mathcal{I}}w_{i}\sum_{x_{it}\in\chi}Pr_{it}(x_{it})\cdot\frac{\exp\left(\mu_{ijt}(x_{it},\Omega_{t}^{(data)})+\beta E_{t}\left[V_{it+1}(x_{it+1},\Omega_{t+1}^{(data)})|x_{it},a_{it}=j\right]\right)}{\exp\left(V_{it}(x_{it},\Omega_{t}^{(data)})\right)}\right).
\end{eqnarray*}
}{\small\par}

$B_{V}$ is a space of $V\equiv\left\{ V_{it}(x_{it},\Omega_{t})\right\} _{x_{it}\in\mathcal{\chi},i\in\mathcal{I},t=1,\cdots,T}$, and $B_{\delta}$ is a space of $\delta\equiv\left\{ \delta_{jt}\right\} _{j\in\mathcal{J}_{t},t=1,\cdots,T}$. Note that $V_{iT+1}(x_{iT+1},\Omega_{T+1}^{(data)})=V_{iT}(x_{iT},\Omega_{T}^{(data)})$ holds by the assumption that consumers have perfect foresight on the state transition of $\Omega_{t}$, and $\Omega_{t}$ remains constant after the terminal period $T$.

Therefore, when the values of $Pr_{it}(x_{it})$ are known, we can solve for $\delta$ by iteratively applying $\Phi^{V,\gamma}$ and obtain $\delta$ in the process. Algorithm \ref{alg:dynamic_BLP_algorithm} shows the steps:

\begin{algorithm}[H]
Set initial values of $V^{(0)}$. Iterate the following $(n=0,1,2,\cdots)$:
\begin{enumerate}
\item Compute $\delta^{(n)}=\iota_{V\rightarrow\delta}\left(V^{(n)}\right)$
\item Update $V$ by $V^{(n+1)}=\Psi_{V\delta\rightarrow V}^{\gamma}\left(V^{(n)},\delta^{(n)}\right)$
\item Exit the iteration if $\left\Vert V^{(n+1)}-V^{(n)}\right\Vert <\epsilon_{V}$
\end{enumerate}
\caption{Proposed dynamic BLP inner-loop algorithm\label{alg:dynamic_BLP_algorithm}}

{\scriptsize{}Note. After the convergence, we should verify that $\left\Vert \log(S^{(data)})-\log(s(\delta,V))\right\Vert $ is sufficiently small. See also Remark \ref{rem:dynamic_BLP_verify}.}{\scriptsize\par}
\end{algorithm}

\begin{rem}
\label{rem:dynamic_BLP_verify}Although $\delta$ satisfying equations (\ref{eq:S_j_dynamic_BLP})-(\ref{eq:V_dynamic_BLP}) satisfies $V=\Phi^{V,\gamma}(V)$, it has not been shown that the solution of $V=\Phi^{V,\gamma}(V)$ actually satisfies equations (\ref{eq:S_j_dynamic_BLP})-(\ref{eq:V_dynamic_BLP}), unlike the static BLP case. If $V=\Phi^{V,\gamma}(V)$, has multiple solutions the solution may not satisfy $S_{jt}^{(data)}=s_{jt}(\delta,V)$. Hence, we should verify that $\left\Vert \log(S^{(data)})-\log(s(\delta,V))\right\Vert $ is sufficiently small after the convergence of the algorithm, using the solution $(\delta,V)$. Note that the Bellman equation (\ref{eq:V_dynamic_BLP}) holds for any $\gamma\in\mathbb{R}$ under $S_{jt}^{(data)}=s_{jt}(\delta,V)\ \forall j,t$.
\end{rem}

\subsubsection*{Algorithms used in the previous studies}

The upper part of Algorithm \ref{alg:dynamic_BLP_traditional-algorithm} was applied in earlier studies. The algorithm is simple: Researchers first solve for the value function $V$ given the values of mean product utilities $\delta$, then calculate the difference between the observed and predicted market shares based on the structural model. If the norm of the difference is larger than a tolerance level, the BLP contraction mapping $\delta_{jt}^{(n+1)}=\delta_{jt}^{(n)}+\phi\left[\log(S_{jt}^{(data)})-\log(s_{jt}(\delta))\right]$ is applied\footnote{In dynamic BLP models, there is no guarantee that the iteration $\delta_{jt}^{(n+1)}=\delta_{jt}^{(n)}+\left[\log(S_{jt}^{(data)})-\log(s_{jt}(\delta^{(n)}))\right]$ converges. Hence, introducing a dampening parameter $\phi\in(0,1]$ is sometimes necessary to stabilize the convergence. Proposition 2 of \citet{Sun2019} shows that the mapping $\delta_{jt}^{(n+1)}=\delta_{jt}^{(n)}+\phi\left[\log(S_{jt}^{(data)})-\log(s_{jt}(\delta^{(n)}))\right]$ is a contraction if $\phi$ is ``sufficiently small'' and additional conditions hold.} and repeat the process again. However, the nested loops increase the computational burden.

We can think of an alternative algorithm shown in the lower part of Algorithm \ref{alg:dynamic_BLP_traditional-algorithm}, which avoids the nested loops by jointly updating the values of $\delta$ and $V$ based on the fixed-point constraint.\footnote{The idea of jointly updating two types of variables can also be found in \citet{Pakes_McGuire1994} algorithm to solve dynamic games. In the algorithm, essentially, all the firms' value functions and strategic variables are jointly updated.} Note that we should solve for two different types of variables $\delta$ and $V$ in the algorithm. In contrast, in the proposed algorithm (Algorithm \ref{alg:dynamic_BLP_algorithm}), $\delta$ is analytically represented as a function of $V$, and we must only solve for $V$, which speeds up convergence. The superior performance of the proposed algorithm is shown in Section \ref{subsec:Numerical-Dynamic-BLP-model}. 

\begin{algorithm}[H]
\begin{itemize}
\item The case with nested loops ($\delta V$-(0) (nested)):

Set initial values of $\delta^{(0)}$. Iterate the following $(n=0,1,2,\cdots)$:
\begin{enumerate}
\item Set initial values of $V^{(0,n)}$. Given $\delta^{(n)}$, iterate the following $(m=0,1,2,\cdots)$:
\begin{enumerate}
\item Update $V$ by $V^{(m+1,n)}=\Psi_{V\delta\rightarrow V}^{\gamma=0}\left(V^{(m,n)},\delta^{(n)}\right)$
\item If $\left\Vert V^{(m+1,n)}-V^{(m,n)}\right\Vert <\epsilon$, let $V^{(n)*}$ be $V^{(m+1,n)}$ and exit
\end{enumerate}
\item Update $\delta$ by $\delta_{jt}^{(n+1)}=\Phi_{\delta V\rightarrow\delta,jt}^{\gamma=0,\phi}(\text{\ensuremath{\delta^{(n)},V^{(n)*}}})\equiv\delta_{jt}^{(n)}+\phi\left[\log\left(S_{jt}^{(data)}\right)-\log\left(s_{jt}(\delta^{(n)},V^{(n)*})\right)\right]$, where $\phi\in(0,1]$ is a dampening parameter. 
\item Exit the iteration if $\left\Vert \delta^{(n+1)}-\delta^{(n)}\right\Vert <\epsilon_{\delta}$
\end{enumerate}
\end{itemize}
\medskip{}

\begin{itemize}
\item The case with one loop ($\delta V$-(0) (joint)):

Set initial values of $\delta^{(0)},V^{(0)}$. Iterate the following $(n=0,1,2,\cdots)$:
\begin{enumerate}
\item Update $\delta$ by $\delta^{(n+1)}=\delta_{jt}^{(n)}+\phi\left[\log\left(S_{jt}^{(data)}\right)-\log\left(s_{jt}(\delta^{(n)},V^{(n)})\right)\right]$
\item Update $V$ by $V^{(n+1)}=\Psi_{V\delta\rightarrow V}^{\gamma=0}\left(V^{(n)},\delta^{(n)}\right)$
\item Exit the iteration if $\left\Vert \delta^{(n+1)}-\delta^{(n)}\right\Vert <\epsilon_{\delta}$ and $\left\Vert V^{(n+1)}-V^{(n)}\right\Vert <\epsilon_{V}$
\end{enumerate}
\end{itemize}
\caption{Traditional dynamic BLP inner-loop Algorithm\label{alg:dynamic_BLP_traditional-algorithm}}
\end{algorithm}

\begin{rem}
\citet{schiraldi2011automobile} and \citet{gowrisankaran2012dynamics} examined durable goods models with replacement demand, where $\mu_{i0t}(x_{it},\Omega_{t})$ depends on $\left\{ \delta_{kt}\right\} _{k\in\mathcal{J}_{t}}$, based on the product holding represented by $x_{it}$. Equation (\ref{eq:S_j_dynamic_BLP}) is valid even in this case, and we can derive an equation on $V$ without using $\delta$.
\end{rem}
\begin{rem}
Typically, the values of $\left\{ Pr_{it}(x_{it})\right\} _{i\in\mathcal{I},t=1,\cdots,T,x_{it}\in\chi}$ are unknown. We must therefore impose assumptions on the form of $\left\{ Pr_{it}(x_{it})\right\} _{i\in\mathcal{I},t=1,\cdots,T,x_{it}\in\chi}$ at the initial period and also solve for the variables. Here, suppose that the values of $\left\{ Pr_{it=1}(x_{it=1})\right\} _{i\in\mathcal{I},x_{it}\in\chi}$ are known.\footnote{We can alternatively assume that $t=1$ is at the stationary state, as in \citet{fukasawa2024lightbulb}. \citet{kasahara2009nonparametric} discuss the issue in the context of standard dynamic discrete choice model estimation with unobserved heterogeneity.} Generally, $\left\{ Pr_{it+1}(x_{it+1})\right\} _{i\in\mathcal{I},t=1,\cdots,T,x_{it+1}\in\chi}$ satisfies the following:

\begin{eqnarray}
Pr_{it+1}(x_{it+1}) & = & \sum_{x_{it}\in\chi}Pr{}_{it}(x_{it})\cdot\sum_{j\in\mathcal{A}_{t}(x_{it})}F(x_{it+1}|x_{it},a_{it}=j)\cdot s_{ijt}^{(ccp)}(x_{it},\Omega_{t}),\label{eq:Pr0_transition}
\end{eqnarray}
where $F(x_{it+1}|x_{it},a_{it}=j)$ denotes the state transition probability of state $x_{it}$. Then, as discussed in more specific model settings in Section \ref{subsec:Numerical-Dynamic-BLP-model}, $Pr_{it}(x_{it})$ can be represented as a function of $V$. Hence, even when the values of $Pr_{it}(x_{it})$ are unknown, we can slightly alter Algorithm \ref{alg:dynamic_BLP_algorithm} accordingly.\footnote{Intuitively, $s_{ijt}^{(ccp)}(x_{it},\Omega_{t})$ is a function of $V$ and $\delta$. Because $\delta$ can be analytically represented as a function of $V$, $Pr_{it}(x_{it})\ (t=1\cdots,T)$ can be sequentially represented as a function of $V$.} 
\end{rem}

\section{Acceleration methods of fixed point iterations\label{sec:Acceleration-methods}}

This section briefly describes three acceleration methods of fixed point iterations: Anderson acceleration, spectral algorithm, and SQUAREM algorithm.

\subsection{Anderson acceleration\label{subsec:Anderson}}

Anderson acceleration, originally proposed by \citet{anderson1965iterative}, is a method to accelerate the convergence of fixed-point iterations.

Suppose that we want to solve $x=\Phi(x)$ for $x$. The Anderson acceleration method solves the problem by the steps shown in Algorithm \ref{alg:Anderson-algorithm}.

\begin{algorithm}[H]
\begin{enumerate}
\item Set initial values of $x^{(0)}$. Choose $m\in\mathbb{N}$ and tolerance level $\epsilon$.
\item Compute $x^{(1)}\equiv\Phi\left(x^{(0)}\right)$ and $f_{0}\equiv\Phi\left(x^{(0)}\right)-x^{(0)}$.
\item Iterate the following $(n=1,2,\cdots)$:
\begin{enumerate}
\item Set $m_{n}=\min\left\{ m,n\right\} $.
\item Set $f_{n}\equiv\Phi\left(x^{(n)}\right)-x^{(n)}$. If $\left\Vert \Phi(x^{(n)})-x^{(n)}\right\Vert <\epsilon$, exit the iteration.
\item Compute 
\begin{eqnarray*}
\theta^{(n)}\equiv\arg\min_{\theta=\left(\theta_{0}^{(n)},\theta_{1}^{(n)},\cdots,\theta_{m_{n}}^{(n)}\right)^{T}} & = & \left\Vert \sum_{l=0}^{m_{n}}\theta_{l}^{(n)}f_{n-l}\right\Vert _{2}^{2}\ s.t.\ \sum_{l=0}^{m_{n}}\theta_{l}^{(n)}=1
\end{eqnarray*}
\item Compute $x^{(n+1)}=\sum_{l=0}^{m_{n}}\theta_{l}^{(n)}x^{(n-l)}$
\end{enumerate}
\end{enumerate}
\caption{Anderson acceleration method for fixed point iterations\label{alg:Anderson-algorithm}}

{\footnotesize{}Note. The constrained optimization problem in Step 2(c) can be solved as follows:}{\footnotesize\par}
\begin{enumerate}
\item {\footnotesize{}Compute $\gamma^{(n)}=\arg\min_{\gamma^{(n)}=\left(\gamma_{0}^{(n)},\gamma_{1}^{(n)},\cdots,\gamma_{m_{n}-1}^{(n)}\right)}\left\Vert f_{n}-\mathcal{F}_{n}\gamma^{(n)}\right\Vert _{2}^{2}=\left(\mathcal{F}_{n}^{T}\mathcal{F}_{n}\right)^{-1}\mathcal{F}_{n}f_{n}$, where $\mathcal{F}_{n}=\left(\Delta f_{n-m_{n}},\cdots,\Delta f_{n-1}\right)$ with $\Delta f_{i}=f_{i+1}-f_{i}$.}{\footnotesize\par}
\item {\footnotesize{}Compute $\theta^{(n)}$ such that $\theta_{0}^{(n)}=\gamma_{0}^{(n)}$, $\theta_{i}^{(n)}=\gamma_{i}^{(n)}-\gamma_{i-1}^{(n)}\ \left(i=1,\cdots,m_{n}-1\right)$, $\theta_{m_{n}}^{(n)}=1-\gamma_{m_{n}-1}^{(n)}$.}{\footnotesize\par}
\end{enumerate}
\end{algorithm}

Intuitively, $x^{(n+1)}=\sum_{l=0}^{m_{n}}\theta_{l}^{(n)}x^{(n-l)}$ is chosen to assign a larger weight on $x^{(n-l)}$ with smaller absolute values of $f_{n-l}\equiv\Phi\left(x^{(n-l)}\right)-x^{(n-l)}$. As noted in \citet{fang2009two}, this method closely relates to the quasi-Newton method. $m$ or $m_{n}$ is sometimes called memory size, because it represents how long we retain the information of past iterations $\left(f_{n-m_{n}},\cdots,f_{n}\right)$. In the numerical experiments shown in Section \ref{sec:Numerical-Experiments}, $m=5$ works well.

In Step 3(c) of the algorithm, we must solve a linear constrained optimization problem. The problem can be easily solved by alternatively solving the linear least squares problem\footnote{Strictly speaking, we cannot rule out the possibility that the values of $\left(f_{n-m_{n}},\cdots,f_{n}\right)$ are close to collinear, and the matrix $\left(\mathcal{F}_{n}^{T}\mathcal{F}_{n}\right)$ is close to singular. In that case, we choose one of $\theta$ minimizing $\left\Vert \sum_{l=0}^{m_{n}}\theta_{l}^{(n)}f_{n-l}\right\Vert _{2}^{2}$. In the numerical experiments in Section \ref{sec:Numerical-Experiments}, using the mldivide function suppressing any warnings on singular matrices in MATLAB and the lstsq function in numpy package in Python works well.} outlined in the note of Algorithm \ref{alg:Anderson-algorithm}.

\subsection{Spectral algorithm\label{subsec:Spectral}}

The spectral algorithm is designed to solve nonlinear equations and nonlinear continuous optimization problems. To solve a nonlinear equation $F(x)=0$, $x^{(n)}$ is iteratively updated as follows, given the initial values of $x^{(0)}$:\footnote{Newton's method, which uses the updating equation $x^{(n+1)}=x^{(n)}-\left(\nabla F(x^{(n)})\right)^{-1}F\left(x^{(n)}\right)\ (n=0,1,2,\cdots)$, attains fast convergence around the solution. Although Newton's method can be applied to solve the equation, computing $\nabla F(x^{(n)})$ requires coding analytical first derivatives of $F$, which is not an easy task, especially when the function $F$ is complicated. Moreover, especially when $n_{x}$, the dimension of $x$, is large, computing the inverse of $n_{x}\times n_{x}$ matrix $\nabla F(x^{(n)})$ is computationally costly. Hence, the use of the spectral algorithm is attractive from the perspective of simplicity and computational cost.}

\[
x^{(n+1)}=x^{(n)}+\alpha^{(n)}F\left(x^{(n)}\right)\ (n=0,1,2,\cdots)
\]

Here, $\alpha^{(n)}$ is typically based on the values of $s^{(n)}\equiv x^{(n)}-x^{(n-1)}$ and $y^{(n)}\equiv F(x^{(n)})-F(x^{(n-1)})$, which can vary across $n=0,1,2,\cdots$. 

The idea of the spectral algorithm can be utilized in the context of fixed point iterations. Suppose we want to solve a fixed point constraint $x=\Phi(x)$. Then, by letting $F(x)=\Phi(x)-x$, we can solve $x=\Phi(x)$ by iteratively updating the values of $x$ as follows:\footnote{Spectral algorithm can be thought of as a generalization of extrapolation method (cf. \citealp{judd1998numerical}) with varying values of $\alpha$.}

\begin{eqnarray}
x^{(n+1)} & = & x^{(n)}+\alpha^{(n)}\left(\Phi\left(x^{(n)}\right)-x^{(n)}\right)\ (n=0,1,2,\cdots)\label{eq:spectral_update}\\
 & = & \alpha^{(n)}\Phi(x^{(n)})+(1-\alpha^{(n)})x^{(n)}\ (n=0,1,2,\cdots).\nonumber 
\end{eqnarray}

As discussed in detail in \citet{varadhan2008simple},\footnote{Step size $\alpha_{S3}\equiv\frac{\left\Vert s^{(n)}\right\Vert _{2}}{\left\Vert y^{(n)}\right\Vert _{2}}$ can be derived from a simple optimization problem $\min_{\alpha^{(n)}}\frac{\left\Vert x^{(n+1)}-x^{(n)}\right\Vert _{2}^{2}}{\left|\alpha^{(n)}\right|}=\frac{\left\Vert s^{(n)}+\alpha^{(n)}y^{(n)}\right\Vert _{2}^{2}}{\left|\alpha^{(n)}\right|}$, as discussed in \citet{varadhan2008simple}.} the step size $\alpha^{(n)}=\alpha_{S3}\equiv\frac{\left\Vert s^{(n)}\right\Vert _{2}}{\left\Vert y^{(n)}\right\Vert _{2}}>0$ works well. The iteration choosing negative values of $\alpha^{(n)}$ leads to farther point from the solution when the mapping $\Phi$ is a contraction, as discussed in the Supplemental Appendix. Although \citet{conlon2020best} and \citet{pal2023comparing} mentioned $\alpha_{S1}\equiv-\frac{s^{(n)\prime}y^{(n)}}{y^{(n)\prime}y^{(n)}}$, they may take negative values, which can destabilize the convergence. In contrast, $\alpha_{S3}\equiv\frac{\left\Vert s^{(n)}\right\Vert _{2}}{\left\Vert y^{(n)}\right\Vert _{2}}$ always takes positive values, and it is preferable to apply $\alpha_{S3}$ as the step size. Algorithm \ref{alg:Spectral-algorithm} shows the detailed steps in the spectral algorithm under the step size $\alpha_{S3}$.

\begin{algorithm}[H]
\begin{enumerate}
\item Set initial values of $x^{(0)}$. Choose $\alpha^{(0)}$ and tolerance level $\epsilon$.
\item Iterate the following $(n=0,1,2,\cdots)$:
\begin{enumerate}
\item Compute $\Phi\left(x^{(n)}\right)$ and $F\left(x^{(n)}\right)=\Phi\left(x^{(n)}\right)-x^{(n)}$.
\item Compute $s^{(n)}\equiv x^{(n)}-x^{(n-1)},y^{(n)}\equiv F(x^{(n)})-F(x^{(n-1)}),\text{and }\alpha^{(n)}=\frac{\left\Vert s^{(n)}\right\Vert _{2}}{\left\Vert y^{(n)}\right\Vert _{2}}\text{\ ir\ }n\geq1.$
\item Compute $x^{(n+1)}=x^{(n)}+\alpha^{(n)}F(x^{(n)})$
\item If $\left\Vert \Phi(x^{(n)})-x^{(n)}\right\Vert <\epsilon$, exit the iteration. Otherwise, go back to Step 2(a).
\end{enumerate}
\end{enumerate}
\caption{Spectral algorithm with fixed-point mapping\label{alg:Spectral-algorithm}}
\end{algorithm}

\subsubsection*{Variable-type-specific step size}

Although the standard spectral algorithm uses a scalar $\alpha^{(n)}$, I newly introduce the idea of variable-type-specific step sizes to the spectral algorithm.

To introduce the idea, suppose we would like to solve a nonlinear equation $F(x)=\left(\begin{array}{c}
F_{1}(x)\\
F_{2}(x)
\end{array}\right)=0\in\mathbb{R}^{n_{1}+n_{2}}$, where $x=\left(\begin{array}{c}
x_{1}\\
x_{2}
\end{array}\right)\in\mathbb{R}^{n_{1}+n_{2}}$. When using the standard step size $\alpha_{S3}$, we set $\alpha^{(n)}=\frac{\left\Vert s^{(n)}\right\Vert _{2}}{\left\Vert y^{(n)}\right\Vert _{2}}$, where $s^{(n)}\equiv\left(\begin{array}{c}
x_{1}^{(n)}\\
x_{2}^{(n)}
\end{array}\right)-\left(\begin{array}{c}
x_{1}^{(n-1)}\\
x_{2}^{(n-1)}
\end{array}\right)$, $y^{(n)}\equiv\left(\begin{array}{c}
F_{1}(x^{(n)})\\
F_{2}(x^{(n)})
\end{array}\right)-\left(\begin{array}{c}
F_{1}(x^{(n-1)})\\
F_{2}(x^{(n-1)})
\end{array}\right)$, and update $x$ by $x^{(n+1)}=x^{(n)}+\alpha^{(n)}F\left(x^{(n)}\right)$. 

Here, we can alternatively update the values of $x=\left(\begin{array}{c}
x_{1}\\
x_{2}
\end{array}\right)$ by $x^{(n+1)}=\left(\begin{array}{c}
x_{1}^{(n+1)}\\
x_{2}^{(n+1)}
\end{array}\right)=\left(\begin{array}{c}
x_{1}^{(n)}\\
x_{2}^{(n)}
\end{array}\right)+\left(\begin{array}{c}
\alpha_{1}^{(n)}F_{1}(x^{(n)})\\
\alpha_{2}^{(n)}F_{2}(x^{(n)})
\end{array}\right),$ where $\alpha_{m}^{(n)}=\frac{\left\Vert s_{m}^{(n)}\right\Vert _{2}}{\left\Vert y_{m}^{(n)}\right\Vert _{2}}$, $s_{m}^{(n)}\equiv x_{m}^{(n)}-x_{m}^{(n-1)}$, $y_{m}^{(n)}\equiv F_{m}(x^{(n)})-F_{m}(x^{(n-1)})\ (m=1,2)$. The idea can be generalized to the case where $N\in\mathbb{N}$ types of variables exist. 

As discussed earlier, $\alpha^{(n)}$ is chosen to accelerate the convergence. In principle, any choices of $\alpha^{(n)}$ are allowed in the spectral algorithm provided it is effective, and its numerical performance is more critical.\footnote{We can justify the strategy with a simple thought experiment. Suppose $x_{1}^{(n)}$ is fixed at the true value $x_{1}^{*}$, and $x_{2}^{(n)}$ is not. Then, we should solve the equation $F_{2}(x_{2};x_{1}^{*})=0$ as an equation of $x_{2}$, and setting $\alpha_{2}^{(n)}=\frac{\left\Vert s_{2}^{(n)}\right\Vert _{2}}{\left\Vert y_{2}^{(n)}\right\Vert _{2}}$ is desirable. If we assume $\alpha_{1}^{(n)}=\alpha_{2}^{(n)}=\alpha^{(n)}$, we should set $\alpha^{(n)}=\frac{\left\Vert s^{(n)}\right\Vert _{2}}{\left\Vert y^{(n)}\right\Vert _{2}}=\frac{\sqrt{\left\Vert s_{1}^{(n)}\right\Vert _{2}^{2}+\left\Vert s_{2}^{(n)}\right\Vert _{2}^{2}}}{\sqrt{\left\Vert y_{1}^{(n)}\right\Vert _{2}^{2}+\left\Vert y_{2}^{(n)}\right\Vert _{2}^{2}}}$, which may not be equal to $\frac{\left\Vert s_{2}^{(n)}\right\Vert _{2}}{\left\Vert y_{2}^{(n)}\right\Vert _{2}}$.}

\subsection{SQUAREM algorithm\label{subsec:SQUAREM}}

\citet{varadhan2008simple} proposed the SQUAREM (Squared Polynomial Extrapolation Methods) algorithm, which solves the fixed-point problem $x=\Phi(x)$, and the updating equation is as follows:

\begin{eqnarray}
x^{(n+1)} & = & x^{(n)}+2\alpha^{(n)}s^{(n)}+\left(\alpha^{(n)}\right)^{2}y^{(n)}\label{eq:SQUAREM_update}
\end{eqnarray}
where $s^{(n)}\equiv\Phi\left(x^{(n)}\right)-x^{(n)}$, $y^{(n)}\equiv F\left(\Phi\left(x^{(n)}\right)\right)-F\left(x^{(n)}\right)=\Phi^{2}\left(x^{(n)}\right)-2\Phi\left(x^{(n)}\right)+x^{(n)}$, and $F(x)=\Phi(x)-x$. Step size $\alpha^{(n)}$ is chosen based on the values of $s^{(n)}$ and $y^{(n)}$.

The SQUAREM algorithm is closely related to the spectral algorithm. Updating equation (\ref{eq:SQUAREM_update}) can be reformulated as:

\begin{eqnarray*}
x^{(n+1)} & = & \left(1-\alpha^{(n)}\right)\left[(1-\alpha^{(n)})x^{(n)}+\alpha^{(n)}\Phi\left(x^{(n)}\right)\right]+\alpha^{(n)}\left[(1-\alpha^{(n)})\Phi\left(x^{(n)}\right)+\alpha^{(n)}\Phi^{2}\left(x^{(n)}\right)\right]\\
 & = & \left(1-\alpha^{(n)}\right)\Psi_{spectral}\left(x^{(n)};\Phi,\alpha^{(n)}\right)+\alpha^{(n)}\Psi_{spectral}\left(\Phi\left(x^{(n)}\right);\Phi,\alpha^{(n)}\right).
\end{eqnarray*}

Here, we define a spectral update function, defined by $\Psi_{spectral}\left(x;\Phi,\alpha\right)\equiv(1-\alpha)x+\alpha\Phi(x)$. The right-hand side can be regarded as the output of the spectral update $\Psi_{spectral}\left(x^{(n)};\Phi,\alpha^{(n)}\right)$ and $\Psi_{spectral}\left(\Phi\left(x^{(n)}\right);\Phi,\alpha^{(n)}\right)$.

Algorithm \ref{alg:SQUAREM-algorithm} shows the steps when we the SQUAREM algorithm with step size $\alpha_{S3}\equiv\frac{\left\Vert s^{(n)}\right\Vert _{2}}{\left\Vert y^{(n)}\right\Vert _{2}}$ is applied:

\begin{algorithm}[H]
\begin{enumerate}
\item Set initial values of $x^{(0)}$ and tolerance level $\epsilon$.
\item Iterate the following $(n=0,1,2,\cdots)$:
\begin{enumerate}
\item Compute $\Phi(x^{(n)})$ and $\Phi^{2}(x^{(n)})\equiv\Phi\left(\Phi(x^{(n)})\right)$
\item Compute $s^{(n)}\equiv\Phi\left(x^{(n)}\right)-x^{(n)},y^{(n)}\equiv\Phi^{2}\left(x^{(n)}\right)-2\Phi\left(x^{(n)}\right)+x^{(n)}$, and $\alpha^{(n)}=\frac{\left\Vert s^{(n)}\right\Vert _{2}}{\left\Vert y^{(n)}\right\Vert _{2}}$.
\item Compute $x^{(n+1)}=x^{(n)}+2\alpha^{(n)}s^{(n)}+\left(\alpha^{(n)}\right)^{2}y^{(n)}$.
\item If $\left\Vert \Phi(x^{(n)})-x^{(n)}\right\Vert <\epsilon$, exit the iteration. Otherwise, go back to step 2(a).
\end{enumerate}
\end{enumerate}
\caption{SQUAREM algorithm\label{alg:SQUAREM-algorithm}}
\end{algorithm}

\subsection{Choice of acceleration methods}

\input{choice_acceleration_methods.tex}

\section{Numerical Experiments\label{sec:Numerical-Experiments}}

This section shows the results of numerical experiments for static and dynamic BLP models. All the experiments were run on a laptop computer with the CPU AMD Ryzen 5 6600H 3.30 GHz, 16.0 GB of RAM, Windows 11 64-bit, and MATLAB 2022b. In this study, we assume an algorithm solving a fixed-point constraint $x=\Phi(x)$ does not converge, when the value of $\left\Vert \Phi(x)-x\right\Vert _{\infty}$ takes an infinite or NaN value, or the number of function evaluations reaches a pre-determined maximum.

\subsection{Static BLP model\label{subsec:Numerical-Static-BLP-model}}

Here, we define $\delta$-(0) and $\delta$-(1) as the algorithms using mapping $\Phi^{\delta,\gamma=0}$ and $\Phi^{\delta,\gamma=1}$, respectively. Similarly, we denote $V$-(0) and $V$-(1) as the algorithms using mapping $\Phi^{V,\gamma=0}$ and $\Phi^{V,\gamma=1}$, respectively.

\subsubsection{Monte Carlo simulation}

To compare the performance of the algorithms, I conducted Monte Carlo simulation. The settings are the same as those of \citet{lee2015computationally} and \citet{dube2012improving}, and similar to those of \citet{pal2023comparing}. Consumer $i$'s utility for choosing product $j$ is $U_{ij}=X_{j}\theta_{i}+\xi_{j}+\epsilon_{ij}\ (j=1,\cdots,J)$, and utility for not buying is $U_{i0}=\epsilon_{i0}$. Let $\delta_{j}\equiv X_{j}\cdot E[\theta_{i}]+\xi_{j}$. For details on the data-generating process, see the Supplemental Appendix. Number of simulation draws is set to 1000. The tolerance level of the convergence is set to 1E-13. The initial value $\delta_{j}^{(0)}$ is set to $\log(S_{j}^{(data)})-\log(S_{0}^{(data)})$, and the initial value of $V_{i}^{(0)}$ is set to 0. 

We evaluate the performance of the BLP inner-loop algorithms in solving for $\delta$ given parameter values $\theta_{n}$ (standard deviation of random coefficients $\theta_{i}$). First, market share data are generated based on the true nonlinear parameter values $\theta_{n}^{*}$ and true $\delta^{*}$. As analysts lack knowledge of the true nonlinear parameter values during estimation process, we imitate the setting by drawing parameter values $\theta_{n}\sim U[0,2\theta_{n}^{*}]$, and measuring the performance of each algorithm under the parameters. This process is repeated 50 times. 

Here, we mainly focus on the number of function evaluations, rather than CPU time. Although the CPU time indicates the exact duration required for iteration termination, generally it largely depends on the algorithm's syntax and the characteristics of each programming language. In contrast, the number of function evaluations remains unaffected by these factors as we focus on the values.\footnote{In the case of the standard fixed-point iterations, Anderson acceleration, and the spectral algorithm, the number of function evaluations equals the number of iterations. Regarding the SQUAREM algorithm, the number of function evaluations is roughly half of the number of iterations, and we focus on the number of function evaluations as the indicator. }

Tables \ref{tab:static_BLP_Monte_Carlo_J-25} and \ref{tab:static_BLP_Monte_Carlo_J-250} show the results for $J=25$ and $J=250$.\footnote{The current study also conducted numerical experiments in settings with $J=25$ and larger constant term of the utility than the setting in Table \ref{tab:static_BLP_Monte_Carlo_J-25}. The results are also consistent with the discussions in this section. They are available upon request.} The mappings $\delta$-(1) and $V$-(1) require significantly fewer iterations than $\delta$-(0) and $V$-(0), with their superior performance more evident for $J=250$. As discussed in Section \ref{subsec:Mappings-on-delta-static-BLP} and detailed in Appendix \ref{subsec:Convergence-properties-mappings}, the convergence speed of the BLP contraction mapping ($\delta$-(0)) can be slow when the outside option share is small, whereas $\delta$-(1) remains unaffected. As shown in the tables' notes, the mean outside option share is 0.850 for $J=25$, and 0.307 for $J=250$. The observations are in line with the discussion. Moreover, the numerical results highlight the duality between the mappings of $\delta$ and $V$: the numbers of iterations are mostly the same for $\delta$-($\gamma$) and $V$-($\gamma$) ($\gamma=0,1$). 

Combining the acceleration methods further reduces the number of function evaluations. Anderson acceleration is the fastest, and outperforms both the spectral and SQUAREM algorithms, which show similar performance. Note that the spectral and SQUAREM algorithms perform similarly. Here, we use the step size $\alpha_{S3}$. When we use the other step sizes ($\alpha_{S1}$,$\alpha_{S2}$), the iterations sometimes do not converge. Detailed results are provided in the Supplemental Appendix.

Note that the value of $\left\Vert \log(S_{j}^{(data)})-\log(s_{j})\right\Vert _{\infty}$ might not be close to 0, even when the iteration converges. For instance, the convergence of the iteration based on the mapping $\delta$-(1) is determined by $\left\Vert \Phi^{\delta,\gamma=1}(\delta)-\delta\right\Vert =\left\Vert \left(\log(S_{j}^{(data)})-\log(s_{j})\right)-\left(\log(S_{0}^{(data)})-\log(s_{0})\right)\right\Vert <\epsilon$, not $\left\Vert \log(S_{j}^{(data)})-\log(s_{j})\right\Vert <\epsilon$. Hence, to ensure that the iteration yields a solution with sufficiently small $\left\Vert \log(S_{j}^{(data)})-\log(s_{j})\right\Vert $, the tables display the mean values of $DIST\equiv\left\Vert \log(S_{j}^{(data)})-\log(s_{j})\right\Vert _{\infty}$ computed using the iteration results. They also show the percentage of settings where the value of $DIST$ is smaller than 1E-12. The results indicate that all algorithms except for the slow $\delta$-(0), $V$-(0) algorithms satisfy the precision constraint.
\begin{center}
\begin{table}[H]
\caption{Results of Monte Carlo simulation (Static BLP model; Continuous consumer types; $J=25$)\label{tab:static_BLP_Monte_Carlo_J-25}}

\scalebox{0.8}{
\begin{centering}
{\scriptsize{}}%
\begin{tabular}{cccccccccccc}
\hline 
 & \multirow{2}{*}{{\scriptsize{}$J$}} & \multicolumn{6}{c}{{\scriptsize{}Func. Eval.}} & {\scriptsize{}Mean} & {\scriptsize{}Conv.} & {\scriptsize{}Mean} & {\scriptsize{}$DIST<\epsilon_{tol}$}\tabularnewline
\cline{3-8} \cline{4-8} \cline{5-8} \cline{6-8} \cline{7-8} \cline{8-8} 
 &  & {\scriptsize{}Mean} & {\scriptsize{}Min.} & {\scriptsize{}25th} & {\scriptsize{}Median.} & {\scriptsize{}75th} & {\scriptsize{}Max.} & {\scriptsize{}CPU time (s)} & {\scriptsize{}(\%)} & {\scriptsize{}$\log_{10}\left(DIST\right)$} & {\scriptsize{}(\%)}\tabularnewline
\hline 
{\scriptsize{}$\delta$-(0) (BLP)} & {\scriptsize{}25} & {\scriptsize{}42.64} & {\scriptsize{}5} & {\scriptsize{}9} & {\scriptsize{}14} & {\scriptsize{}31} & {\scriptsize{}630} & {\scriptsize{}0.0123} & {\scriptsize{}100} & {\scriptsize{}-15.5} & {\scriptsize{}100}\tabularnewline
{\scriptsize{}$\delta$-(0) (BLP) + Anderson} & {\scriptsize{}25} & {\scriptsize{}8.34} & {\scriptsize{}5} & {\scriptsize{}7} & {\scriptsize{}8} & {\scriptsize{}9} & {\scriptsize{}19} & {\scriptsize{}0.0038} & {\scriptsize{}100} & {\scriptsize{}-17} & {\scriptsize{}100}\tabularnewline
{\scriptsize{}$\delta$-(0) (BLP) + Spectral} & {\scriptsize{}25} & {\scriptsize{}12.3} & {\scriptsize{}5} & {\scriptsize{}7} & {\scriptsize{}10} & {\scriptsize{}13} & {\scriptsize{}45} & {\scriptsize{}0.00408} & {\scriptsize{}100} & {\scriptsize{}-16.3} & {\scriptsize{}100}\tabularnewline
{\scriptsize{}$\delta$-(0) (BLP) + SQUAREM} & {\scriptsize{}25} & {\scriptsize{}12.92} & {\scriptsize{}5} & {\scriptsize{}8} & {\scriptsize{}10} & {\scriptsize{}13} & {\scriptsize{}46} & {\scriptsize{}0.00374} & {\scriptsize{}100} & {\scriptsize{}-16.4} & {\scriptsize{}100}\tabularnewline
{\scriptsize{}$\delta$-(1)} & {\scriptsize{}25} & {\scriptsize{}14.58} & {\scriptsize{}5} & {\scriptsize{}8} & {\scriptsize{}12} & {\scriptsize{}17} & {\scriptsize{}43} & {\scriptsize{}0.00424} & {\scriptsize{}100} & {\scriptsize{}-15.9} & {\scriptsize{}100}\tabularnewline
{\scriptsize{}$\delta$-(1) (BLP) + Anderson} & {\scriptsize{}25} & {\scriptsize{}7.5} & {\scriptsize{}5} & {\scriptsize{}6} & {\scriptsize{}7} & {\scriptsize{}9} & {\scriptsize{}11} & {\scriptsize{}0.00326} & {\scriptsize{}100} & {\scriptsize{}-17.1} & {\scriptsize{}100}\tabularnewline
{\scriptsize{}$\delta$-(1) + Spectral} & {\scriptsize{}25} & {\scriptsize{}9.38} & {\scriptsize{}5} & {\scriptsize{}7} & {\scriptsize{}9} & {\scriptsize{}11} & {\scriptsize{}22} & {\scriptsize{}0.00308} & {\scriptsize{}100} & {\scriptsize{}-16.4} & {\scriptsize{}100}\tabularnewline
{\scriptsize{}$\delta$-(1) + SQUAREM} & {\scriptsize{}25} & {\scriptsize{}9.54} & {\scriptsize{}5} & {\scriptsize{}7} & {\scriptsize{}9} & {\scriptsize{}11} & {\scriptsize{}17} & {\scriptsize{}0.00282} & {\scriptsize{}100} & {\scriptsize{}-16.7} & {\scriptsize{}100}\tabularnewline
\hline 
{\scriptsize{}$V$-(0)} & {\scriptsize{}25} & {\scriptsize{}43.66} & {\scriptsize{}4} & {\scriptsize{}8} & {\scriptsize{}14} & {\scriptsize{}34} & {\scriptsize{}643} & {\scriptsize{}0.01466} & {\scriptsize{}100} & {\scriptsize{}-15.8} & {\scriptsize{}100}\tabularnewline
{\scriptsize{}$V$-(0) + Anderson} & {\scriptsize{}25} & {\scriptsize{}8.22} & {\scriptsize{}4} & {\scriptsize{}5} & {\scriptsize{}7} & {\scriptsize{}9} & {\scriptsize{}30} & {\scriptsize{}0.0038} & {\scriptsize{}100} & {\scriptsize{}-16.7} & {\scriptsize{}100}\tabularnewline
{\scriptsize{}$V$-(0) + Spectral} & {\scriptsize{}25} & {\scriptsize{}10.98} & {\scriptsize{}4} & {\scriptsize{}6} & {\scriptsize{}9} & {\scriptsize{}12} & {\scriptsize{}42} & {\scriptsize{}0.00436} & {\scriptsize{}100} & {\scriptsize{}-16.6} & {\scriptsize{}100}\tabularnewline
{\scriptsize{}$V$-(0) + SQUAREM} & {\scriptsize{}25} & {\scriptsize{}11.5} & {\scriptsize{}4} & {\scriptsize{}6} & {\scriptsize{}9} & {\scriptsize{}13} & {\scriptsize{}37} & {\scriptsize{}0.00408} & {\scriptsize{}100} & {\scriptsize{}-16.6} & {\scriptsize{}100}\tabularnewline
{\scriptsize{}$V$-(1)} & {\scriptsize{}25} & {\scriptsize{}14.52} & {\scriptsize{}4} & {\scriptsize{}7} & {\scriptsize{}12} & {\scriptsize{}17} & {\scriptsize{}44} & {\scriptsize{}0.0054} & {\scriptsize{}100} & {\scriptsize{}-16.1} & {\scriptsize{}100}\tabularnewline
{\scriptsize{}$V$-(1) + Anderson} & {\scriptsize{}25} & {\scriptsize{}7.24} & {\scriptsize{}4} & {\scriptsize{}5} & {\scriptsize{}7} & {\scriptsize{}8} & {\scriptsize{}16} & {\scriptsize{}0.0034} & {\scriptsize{}100} & {\scriptsize{}-16.6} & {\scriptsize{}100}\tabularnewline
{\scriptsize{}$V$-(1) + Spectral} & {\scriptsize{}25} & {\scriptsize{}9.74} & {\scriptsize{}4} & {\scriptsize{}6} & {\scriptsize{}8} & {\scriptsize{}12} & {\scriptsize{}30} & {\scriptsize{}0.00412} & {\scriptsize{}100} & {\scriptsize{}-16.4} & {\scriptsize{}100}\tabularnewline
{\scriptsize{}$V$-(1) + SQUAREM} & {\scriptsize{}25} & {\scriptsize{}9.8} & {\scriptsize{}4} & {\scriptsize{}6} & {\scriptsize{}9} & {\scriptsize{}12} & {\scriptsize{}19} & {\scriptsize{}0.0038} & {\scriptsize{}100} & {\scriptsize{}-16.5} & {\scriptsize{}100}\tabularnewline
\hline 
\end{tabular}{\scriptsize\par}
\par\end{centering}
}

{\footnotesize{}\input{notes/note_static_BLP_results.tex}}{\footnotesize\par}

{\footnotesize{}The mean outside option share is 0.847.}{\footnotesize\par}
\end{table}
\par\end{center}

\begin{center}
\begin{table}[H]
\caption{Results of Monte Carlo simulation (Static BLP model; Continuous consumer types; $J=250$)\label{tab:static_BLP_Monte_Carlo_J-250}}

\scalebox{0.8}{
\begin{centering}
{\scriptsize{}}%
\begin{tabular}{cccccccccccc}
\hline 
 & \multirow{2}{*}{{\scriptsize{}$J$}} & \multicolumn{6}{c}{{\scriptsize{}Func. Eval.}} & {\scriptsize{}Mean} & {\scriptsize{}Conv.} & {\scriptsize{}Mean} & {\scriptsize{}$DIST<\epsilon_{tol}$}\tabularnewline
\cline{3-8} \cline{4-8} \cline{5-8} \cline{6-8} \cline{7-8} \cline{8-8} 
 &  & {\scriptsize{}Mean} & {\scriptsize{}Min.} & {\scriptsize{}25th} & {\scriptsize{}Median.} & {\scriptsize{}75th} & {\scriptsize{}Max.} & {\scriptsize{}CPU time (s)} & {\scriptsize{}(\%)} & {\scriptsize{}$\log_{10}\left(DIST\right)$} & {\scriptsize{}(\%)}\tabularnewline
\hline 
{\scriptsize{}$\delta$-(0) (BLP)} & {\scriptsize{}250} & {\scriptsize{}201.56} & {\scriptsize{}21} & {\scriptsize{}70} & {\scriptsize{}114.5} & {\scriptsize{}196} & {\scriptsize{}1000} & {\scriptsize{}0.732} & {\scriptsize{}94} & {\scriptsize{}-13.3} & {\scriptsize{}98}\tabularnewline
{\scriptsize{}$\delta$-(0) (BLP) + Anderson} & {\scriptsize{}250} & {\scriptsize{}12.06} & {\scriptsize{}9} & {\scriptsize{}11} & {\scriptsize{}12} & {\scriptsize{}13} & {\scriptsize{}22} & {\scriptsize{}0.04486} & {\scriptsize{}100} & {\scriptsize{}-15.5} & {\scriptsize{}100}\tabularnewline
{\scriptsize{}$\delta$-(0) (BLP) + Spectral} & {\scriptsize{}250} & {\scriptsize{}33.16} & {\scriptsize{}13} & {\scriptsize{}23} & {\scriptsize{}27} & {\scriptsize{}36} & {\scriptsize{}91} & {\scriptsize{}0.12858} & {\scriptsize{}100} & {\scriptsize{}-14.3} & {\scriptsize{}100}\tabularnewline
{\scriptsize{}$\delta$-(0) (BLP) + SQUAREM} & {\scriptsize{}250} & {\scriptsize{}37.12} & {\scriptsize{}14} & {\scriptsize{}22} & {\scriptsize{}30} & {\scriptsize{}43} & {\scriptsize{}142} & {\scriptsize{}0.13904} & {\scriptsize{}100} & {\scriptsize{}-14.5} & {\scriptsize{}100}\tabularnewline
{\scriptsize{}$\delta$-(1)} & {\scriptsize{}250} & {\scriptsize{}24.32} & {\scriptsize{}9} & {\scriptsize{}16} & {\scriptsize{}24} & {\scriptsize{}30} & {\scriptsize{}56} & {\scriptsize{}0.08928} & {\scriptsize{}100} & {\scriptsize{}-14.5} & {\scriptsize{}100}\tabularnewline
{\scriptsize{}$\delta$-(1) (BLP) + Anderson} & {\scriptsize{}250} & {\scriptsize{}9.76} & {\scriptsize{}7} & {\scriptsize{}9} & {\scriptsize{}10} & {\scriptsize{}11} & {\scriptsize{}14} & {\scriptsize{}0.03622} & {\scriptsize{}100} & {\scriptsize{}-15.9} & {\scriptsize{}100}\tabularnewline
{\scriptsize{}$\delta$-(1) + Spectral} & {\scriptsize{}250} & {\scriptsize{}13.88} & {\scriptsize{}8} & {\scriptsize{}12} & {\scriptsize{}14} & {\scriptsize{}16} & {\scriptsize{}20} & {\scriptsize{}0.05212} & {\scriptsize{}100} & {\scriptsize{}-15.1} & {\scriptsize{}100}\tabularnewline
{\scriptsize{}$\delta$-(1) + SQUAREM} & {\scriptsize{}250} & {\scriptsize{}14.02} & {\scriptsize{}8} & {\scriptsize{}12} & {\scriptsize{}14} & {\scriptsize{}16} & {\scriptsize{}22} & {\scriptsize{}0.05056} & {\scriptsize{}100} & {\scriptsize{}-15.4} & {\scriptsize{}100}\tabularnewline
\hline 
{\scriptsize{}$V$-(0)} & {\scriptsize{}250} & {\scriptsize{}211.7} & {\scriptsize{}21} & {\scriptsize{}81} & {\scriptsize{}126} & {\scriptsize{}211} & {\scriptsize{}1000} & {\scriptsize{}0.9056} & {\scriptsize{}94} & {\scriptsize{}-13.4} & {\scriptsize{}98}\tabularnewline
{\scriptsize{}$V$-(0) + Anderson} & {\scriptsize{}250} & {\scriptsize{}14.82} & {\scriptsize{}8} & {\scriptsize{}12} & {\scriptsize{}14} & {\scriptsize{}16} & {\scriptsize{}37} & {\scriptsize{}0.06734} & {\scriptsize{}100} & {\scriptsize{}-15.1} & {\scriptsize{}100}\tabularnewline
{\scriptsize{}$V$-(0) + Spectral} & {\scriptsize{}250} & {\scriptsize{}29.74} & {\scriptsize{}12} & {\scriptsize{}22} & {\scriptsize{}27} & {\scriptsize{}34} & {\scriptsize{}67} & {\scriptsize{}0.12864} & {\scriptsize{}100} & {\scriptsize{}-14.8} & {\scriptsize{}100}\tabularnewline
{\scriptsize{}$V$-(0) + SQUAREM} & {\scriptsize{}250} & {\scriptsize{}29.82} & {\scriptsize{}11} & {\scriptsize{}23} & {\scriptsize{}28} & {\scriptsize{}33} & {\scriptsize{}62} & {\scriptsize{}0.13044} & {\scriptsize{}100} & {\scriptsize{}-14.7} & {\scriptsize{}100}\tabularnewline
{\scriptsize{}$V$-(1)} & {\scriptsize{}250} & {\scriptsize{}26.12} & {\scriptsize{}10} & {\scriptsize{}18} & {\scriptsize{}25} & {\scriptsize{}32} & {\scriptsize{}63} & {\scriptsize{}0.11434} & {\scriptsize{}100} & {\scriptsize{}-14.6} & {\scriptsize{}100}\tabularnewline
{\scriptsize{}$V$-(1) + Anderson} & {\scriptsize{}250} & {\scriptsize{}10.8} & {\scriptsize{}7} & {\scriptsize{}9} & {\scriptsize{}10.5} & {\scriptsize{}12} & {\scriptsize{}18} & {\scriptsize{}0.04738} & {\scriptsize{}100} & {\scriptsize{}-15.2} & {\scriptsize{}100}\tabularnewline
{\scriptsize{}$V$-(1) + Spectral} & {\scriptsize{}250} & {\scriptsize{}17.14} & {\scriptsize{}9} & {\scriptsize{}15} & {\scriptsize{}17} & {\scriptsize{}19} & {\scriptsize{}29} & {\scriptsize{}0.0752} & {\scriptsize{}100} & {\scriptsize{}-15} & {\scriptsize{}100}\tabularnewline
{\scriptsize{}$V$-(1) + SQUAREM} & {\scriptsize{}250} & {\scriptsize{}16.88} & {\scriptsize{}11} & {\scriptsize{}14} & {\scriptsize{}16} & {\scriptsize{}19} & {\scriptsize{}24} & {\scriptsize{}0.07268} & {\scriptsize{}100} & {\scriptsize{}-14.9} & {\scriptsize{}100}\tabularnewline
\hline 
\end{tabular}{\scriptsize\par}
\par\end{centering}
}

{\footnotesize{}\input{notes/note_static_BLP_results.tex}}{\footnotesize\par}

{\footnotesize{}The mean outside option share is 0.308.}{\footnotesize\par}
\end{table}
\par\end{center}

\subsubsection{Replications of \citet{nevo2001measuring} and Berry et al. \citeyearpar{berry1995automobile,berry1999voluntary}}

To further examine the performance of inner-loop algorithms, we use the datasets from \citet{nevo2001measuring} and Berry et al. \citeyearpar{berry1995automobile,berry1999voluntary} to estimate demand parameters under different inner-loop algorithms.

In the numerical experiments, the PyBLP package (version 0.13.0) in Python (\citealp{conlon2020best}) was used. Because the package itself does not support the new mapping $\delta$-(1), the original PyBLP source code was modified to introduce the new mapping $\delta$-(1). In addition, the steps of the Anderson acceleration were newly coded. Regarding the spectral algorithm, we cannot choose step sizes other than $\alpha_{S2}$, because the dfsane function in the SciPy package, which the PyBLP relies on, exclusively supports $\alpha_{S2}$. Hence, the dfsane function in the SciPy package was modified to allow different step sizes. Moreover, because the current dfsane function does not support such iterations without globalization strategies, the code was adjusted to implement such iterations. Notably, it rarely affects the results, as shown in the Supplemental Appendix.

\paragraph*{\citet{nevo2001measuring}'s dataset}

Settings largely reflect those specified in Figure 5 of \citet{conlon2020best}, including initial parameter values and tolerance levels, except for the choice of step sizes, the existence of globalization strategies, and the choice of mappings.

Table \ref{tab:replication-Nevo} shows that the algorithm $\delta$-(1) is more than twice as fast as the algorithm $\delta$-(0). Incorporating Anderson acceleration further reduces total function evaluations by approximately 75\%. In particular, the values of total objective evaluations and the objective remain consistent across all inner-loop algorithms, indicating that the choice of the inner-loop algorithm does not affect the outer-loop estimation in this setting. 
\begin{center}
\begin{table}[H]
\caption{Estimation results using the \citet{nevo2001measuring}'s dataset\label{tab:replication-Nevo}}

\begin{centering}
{\footnotesize{}}%
\begin{tabular}{ccccc}
\hline 
 & {\footnotesize{}Mean feval} & {\footnotesize{}Total obj eval} & {\footnotesize{}Total feval} & {\footnotesize{}Objective}\tabularnewline
\hline 
\hline 
{\footnotesize{}$\delta$-(1)} & {\footnotesize{}43.288} & {\footnotesize{}57} & {\footnotesize{}231937} & {\footnotesize{}4.562}\tabularnewline
{\footnotesize{}$\delta$-(1) + Anderson} & {\footnotesize{}11.506} & {\footnotesize{}57} & {\footnotesize{}61649} & {\footnotesize{}4.562}\tabularnewline
{\footnotesize{}$\delta$-(1) + Spectral} & {\footnotesize{}19.209} & {\footnotesize{}57} & {\footnotesize{}102921} & {\footnotesize{}4.562}\tabularnewline
{\footnotesize{}$\delta$-(1) + SQUAREM} & {\footnotesize{}19.911} & {\footnotesize{}57} & {\footnotesize{}106683} & {\footnotesize{}4.562}\tabularnewline
{\footnotesize{}$\delta$-(0)} & {\footnotesize{}94.714} & {\footnotesize{}57} & {\footnotesize{}507475} & {\footnotesize{}4.562}\tabularnewline
{\footnotesize{}$\delta$-(0) + Anderson} & {\footnotesize{}13.629} & {\footnotesize{}57} & {\footnotesize{}73026} & {\footnotesize{}4.562}\tabularnewline
{\footnotesize{}$\delta$-(0) + Spectral} & {\footnotesize{}27.796} & {\footnotesize{}57} & {\footnotesize{}148931} & {\footnotesize{}4.562}\tabularnewline
{\footnotesize{}$\delta$-(0) + SQUAREM} & {\footnotesize{}26.866} & {\footnotesize{}57} & {\footnotesize{}143949} & {\footnotesize{}4.562}\tabularnewline
\hline 
\end{tabular}{\footnotesize\par}
\par\end{centering}
\input{notes/note_pyblp_results.tex}
\end{table}
\par\end{center}

\paragraph*{Berry et al.\citeyearpar{berry1995automobile,berry1999voluntary}'s dataset}

The same settings as the one specified in Figure 6 of \citet{conlon2020best} were applied. \footnote{Besides globalization strategies and the choice of mappings, the key difference lies in the optimization tolerance level, which is set to 1E-4 here, instead of the default setting 1E-8. The results under the optimization tolerance 1E-8 are shown in the Supplemental Appendix. As discussed, the choice of the inner-loop algorithms has a slight impact on the inner-loop numerical errors, and may largely affect the convergence of the optimization, especially with a tight optimization tolerance level. We discuss that the convergence speed of outer-loop iterations can be sensitive to the inner-loop numerical errors, even when estimated parameters are not sensitive to inner-loop numerical errors. For details, see the Supplemental Appendix.} 

Table \ref{tab:replication-Berry-et-al.} shows that the algorithm $\delta$-(1) reduces the number of function evaluations by more than 10\%. Combining the spectral/Anderson algorithm further reduces the number of function evaluations by more than 90\%.
\begin{center}
\begin{table}[H]
\caption{Estimation results using the Berry et al.\citeyearpar{berry1995automobile,berry1999voluntary}'s dataset\label{tab:replication-Berry-et-al.}}

\begin{centering}
{\footnotesize{}}%
\begin{tabular}{ccccc}
\hline 
 & {\footnotesize{}Mean feval} & {\footnotesize{}Total obj eval} & {\footnotesize{}Total feval} & {\footnotesize{}Objective}\tabularnewline
\hline 
\hline 
{\footnotesize{}$\delta$-(1)} & {\footnotesize{}216.938} & {\footnotesize{}70} & {\footnotesize{}303713} & {\footnotesize{}497.336}\tabularnewline
{\footnotesize{}$\delta$-(1) + Anderson} & {\footnotesize{}16.280} & {\footnotesize{}70} & {\footnotesize{}22792} & {\footnotesize{}497.336}\tabularnewline
{\footnotesize{}$\delta$-(1) + Spectral} & {\footnotesize{}44.153} & {\footnotesize{}70} & {\footnotesize{}61814} & {\footnotesize{}497.336}\tabularnewline
{\footnotesize{}$\delta$-(1) + SQUAREM} & {\footnotesize{}48.229} & {\footnotesize{}70} & {\footnotesize{}67520} & {\footnotesize{}497.336}\tabularnewline
{\footnotesize{}$\delta$-(0)} & {\footnotesize{}246.358} & {\footnotesize{}70} & {\footnotesize{}344901} & {\footnotesize{}497.336}\tabularnewline
{\footnotesize{}$\delta$-(0) + Anderson} & {\footnotesize{}16.061} & {\footnotesize{}70} & {\footnotesize{}22486} & {\footnotesize{}497.336}\tabularnewline
{\footnotesize{}$\delta$-(0) + Spectral} & {\footnotesize{}47.518} & {\footnotesize{}70} & {\footnotesize{}66525} & {\footnotesize{}497.336}\tabularnewline
{\footnotesize{}$\delta$-(0) + SQUAREM} & {\footnotesize{}46.700} & {\footnotesize{}70} & {\footnotesize{}65380} & {\footnotesize{}497.336}\tabularnewline
\hline 
\end{tabular}{\footnotesize\par}
\par\end{centering}

\input{notes/note_pyblp_results.tex}
{\footnotesize{}The outer loop tolerance is set to 1E-4.}{\footnotesize\par}
\end{table}
\par\end{center}

\subsubsection*{Convergence speed of $\delta$-(1)}

In the estimation using the Berry et al.\citeyearpar{berry1995automobile,berry1999voluntary}'s dataset, applying $\delta$-(1) rather than $\delta$-(0) had minimal impact on convergence speed, differing from the Monte Carlo experiments and the estimation using the \citet{nevo2001measuring}'s dataset. However, the results align with the convergence properties of $\delta$-(1). As briefly discussed in Section \ref{subsec:Mappings-on-delta-static-BLP} and discussed in detail in Appendix \ref{subsec:Convergence-properties-mappings}, the convergence speed of $\delta$-(1) is significantly influenced by the size of the consumer heterogeneity. As shown in Proposition \ref{prop:property_mapping_delta} in Appendix \ref{subsec:Convergence-properties-mappings}, $\left\Vert \Phi^{\delta,\gamma=1}(\delta_{1})-\Phi^{\delta,\gamma=1}(\delta_{2})\right\Vert _{\infty}\leq\widetilde{c_{\gamma=1}}\left(\left\Vert \max\{\delta_{1},\delta_{2}\}-\delta_{1}\right\Vert _{\infty}+\left\Vert \max\{\delta_{1},\delta_{2}\}-\delta_{2}\right\Vert _{\infty}\right)\leq2\widetilde{c_{\gamma=1}}\cdot\left\Vert \delta_{1}-\delta_{2}\right\Vert _{\infty}$ holds, where $\widetilde{c_{\gamma=1}}\equiv\sup_{i\in\mathcal{I},J^{*}\subset\mathcal{J},\delta}s_{iJ^{*}}(\delta)-\inf_{i\in\mathcal{I},J^{*}\subset\mathcal{J},\delta}s_{iJ^{*}}(\delta)\leq1$. Because the value of $\widetilde{c_{\gamma=1}}$ itself is not easy to compute when the number of products is large, we alternatively consider $\max_{i\in\mathcal{I}}s_{i0t}(\delta)-\min_{i\in\mathcal{I}}s_{i0t}(\delta)$ at the converged $\delta$, as a rough approximation of the term. By the proposition, $\delta$-(1) is expected to converge slowly when consumer heterogeneity, measured by $\max_{i\in\mathcal{I}}s_{i0t}(\delta)-\min_{i\in\mathcal{I}}s_{i0t}(\delta)\leq1$, is close to 1.

Table \ref{tab:Values-hetero} shows the values of $\max_{i\in\mathcal{I}}s_{i0t}(\delta)-\min_{i\in\mathcal{I}}s_{i0t}(\delta)$ from the numerical experiments presented earlier. Using the Berry et al.\citeyearpar{berry1995automobile,berry1999voluntary}'s dataset, the value of $\max_{i}s_{i0t}(\delta)-\min_{i}s_{i0t}(\delta)$ is nearly 1. In contrast, in the Monte Carlo experiment with $J=250$ and the estimation using the \citet{nevo2001measuring}'s dataset, the values are far from 1.

\begin{table}[H]
\begin{centering}
{\footnotesize{}}%
\begin{tabular}{ccc}
\hline 
{\footnotesize{}Settings} & {\footnotesize{}Mean} & {\footnotesize{}Std}\tabularnewline
\hline 
\hline 
{\footnotesize{}Monte Carlo $(J=250)$} & {\footnotesize{}0.839} & {\footnotesize{}0.180}\tabularnewline
{\footnotesize{}Nevo} & {\footnotesize{}0.896} & {\footnotesize{}0.058}\tabularnewline
{\footnotesize{}BLP} & {\footnotesize{}0.999} & {\footnotesize{}0.001}\tabularnewline
\hline 
\end{tabular}{\footnotesize\par}
\par\end{centering}
\caption{Values of $\max_{i}s_{i0t}(\delta)-\min_{i}s_{i0t}(\delta)$\label{tab:Values-hetero}}
\end{table}

\medskip{}

From the numerical experiments, we can obtain the following findings:
\begin{itemize}
\item Algorithm $\delta$-(1) significantly outperforms $\delta$-(0), especially when the outside option share is small and consumer heterogeneity is relatively small
\item Anderson acceleration is the best acceleration method
\item The spectral and SQUAREM algorithms are also fairly effective
\end{itemize}
Considering these observations, it is preferable to first consider applying the new algorithm $\delta$-(1) in the standard fixed-point iteration, and then consider combining the Anderson acceleration. 

\subsection{Dynamic BLP model\label{subsec:Numerical-Dynamic-BLP-model}}

Next, the numerical results of inner-loop algorithms for estimating the dynamic BLP models are shown. Here, we consider the model of perfectly durable goods, following the setting considered in Section 5.2 of \citet{Sun2019}. Consumers are forward-looking, and do not buy a product once they purchase any product.

When consumer $i$ purchases product $j$ at time $t$, they receive utility $U_{ijt}=X_{jt}^{\prime}\theta_{i}+\xi_{jt}+\epsilon_{ijt}$. If no product is purchased given that they do not own any durable product, the utility is $U_{i0t}=\beta E_{t}[V_{it+1}(\Omega_{t+1})|\Omega_{t}]+\epsilon_{i0t}$. Here, $X_{jt},\xi_{jt}$ denotes the observed / unobserved product characteristics of product $j$ at time $t$, and let $\delta_{jt}\equiv X_{jt}^{\prime}E[\theta_{i}]+\xi_{jt}$. $V_{it}(\Omega_{t})$ denotes the (integrated) value function of consumer $i$ not owning any product at time $t$. Besides, let $\mu_{ijt}\equiv X_{jt}^{\prime}\left(\theta_{i}-E[\theta_{i}]\right)$.

Assuming that $\epsilon$ follows i.i.d. type-I extreme value distribution, the CCP of product $j$ for consumer $i$ not holding any product at time $t$ is $s_{ijt}^{(ccp)}=\frac{\exp\left(\delta_{jt}+\mu_{ijt}\right)}{\exp(V_{it}(\Omega_{t}))}$. Moreover, the CCP of not buying any product for consumer $i$ not holding any product at time $t$ is $s_{i0t}^{(ccp)}=\frac{\exp\left(\beta E_{t}[V_{it+1}(\Omega_{t+1})|\Omega_{t}]\right)}{\exp(V_{it}(\Omega_{t}))}.$

Because consumers exit the market after a purchase, $Pr0_{it}$, the fraction of type $i$ consumers not owning any durable product, satisfies $Pr0_{it+1}=Pr0_{it}\cdot s_{i0t}^{(ccp)}$. Note that $Pr0_{it}$ corresponds to $Pr_{it}(x_{it})$ in the general dynamic BLP model discussed in Section \ref{sec:Dynamic-BLP-model}.

As in the Monte Carlo experiments of static BLP models, I draw the values of $\theta_{n}$ (standard deviation of random coefficients) by $\theta_{n}\sim U[0,2\theta_{n}^{*}]$, where $\theta_{n}^{*}$ denotes the true parameters. The number of simulation draws to approximate continuous consumer distribution is set to 50, and the process is repeated 50 times. Further details on the data-generating process are provided in the Supplemental Appendix. Regarding future expectations, we consider two specifications: Perfect foresight, and Inclusive value sufficiency. We discuss them below.

\subsubsection{Model under perfect foresight}

First, consider the setting where consumers have perfect foresight regarding the transition of $\Omega_{t}$, and $\Omega_{t}$ remaining constant after the terminal period $T$, as described in Section \ref{sec:Dynamic-BLP-model}. Algorithm \ref{alg:durable-perfect-foresight} shows the proposed algorithm, which we denote $V$-($\gamma$).

The tolerance level of the convergence is set to 1E-12. Observations from the current data-generating process indicate that the simulated market stabilizes after $t=50$. Consequently, the value of $T$ is set to 50. 

\begin{algorithm}[H]
Set the initial values of $V^{(0)}$. Iterate the following $(n=0,1,2,\cdots)$:
\begin{enumerate}
\item For $t=1:T$,
\begin{enumerate}
\item Using $Pr0_{it}$ and $V_{it}^{(n)}$, compute $\delta_{jt}^{(n)}=\iota_{V\rightarrow\delta,jt}^{\gamma}\left(V^{(n)}\right)=\log\left(S_{jt}^{(data)}\right)-\log\left(\sum_{i}w_{i}Pr0_{it}\cdot\frac{\exp\left(\mu_{ijt}\right)}{\exp\left(V_{it}^{(n)}\right)}\right)$
\item Compute $s_{ijt}^{(ccp)}=\frac{\exp\left(\delta_{jt}^{(n)}+\mu_{ijt}\right)}{\exp(V_{it}^{(n)})}$ for $i\in\mathcal{I},j\in\mathcal{J}$
\item Compute $s_{i0t}^{(ccp)}=1-\sum_{j\in\mathcal{J}_{t}}s_{ijt}^{(ccp)}$
\item Update $Pr0_{it+1}=Pr0_{it}\cdot s_{i0t}^{(ccp)}$
\end{enumerate}
\item Update $V$ by $V_{it}^{(n+1)}=\Psi_{V\delta\rightarrow V,it}^{\gamma}\left(V^{(n)},\delta^{(n)}\right)=\log\left(\exp\left(\beta V_{it+1}^{(n)}\right)+\sum_{j\in\mathcal{J}_{t}}\exp\left(\delta_{jt}^{(n)}+\mu_{ijt}\right)\cdot\left(\frac{s_{0t}(V^{(n)})}{S_{0t}^{(data)}}\right)^{\gamma}\right)$ for $t=1,\cdots,T$. Here, let $V_{iT+1}^{(n)}=V_{iT}^{(n)}$. $s_{0t}(V^{(n)})$ is computed by $s_{0t}(V^{(n)})=\sum_{i}w_{i}\frac{\exp\left(\beta V_{it+1}^{(n)}\right)}{\exp\left(V_{it}^{(n)}\right)}$.
\item Exit the iteration if $\left\Vert V^{(n+1)}-V^{(n)}\right\Vert <\epsilon_{V}$
\end{enumerate}
\caption{Inner-loop Algorithm of dynamic BLP (Perfectly durable goods; Perfect foresight)\label{alg:durable-perfect-foresight}}
\end{algorithm}

To compare the performance of the proposed and the traditional algorithms, both algorithms are run. In the traditional algorithm, the values of $\delta$ and $V$ are jointly updated until convergence. Unlike Algorithm \ref{alg:durable-perfect-foresight} , Step 1(a) is omitted and instead add a step $\delta_{jt}^{(n+1)}=\delta_{jt}^{(n)}+\log\left(S_{jt}^{(data)}\right)-\log\left(s_{jt}(\delta^{(n)},V^{(n)})\right)-\gamma\left[\log\left(S_{0t}^{(data)}\right)-\log\left(s_{0t}\left(\delta^{(n)},V^{(n)}\right)\right)\right]\ (\gamma=0,1)$ between Steps 2 and 3. In addition, the termination of the iteration is determined by $\left\Vert V^{(n+1)}-V^{(n)}\right\Vert <\epsilon_{V},\left\Vert \delta^{(n+1)}-\delta^{(n)}\right\Vert <\epsilon_{\delta}$. We denote the algorithm as $V\delta$-($\gamma$) (joint). Note that the nested version of the algorithm (upper part of Algorithm \ref{alg:dynamic_BLP_traditional-algorithm}) exhibits poor performance, as shown in the Supplemental Appendix. Thus, the results are excluded in this section.

As in the static BLP models, the results of the algorithms combining the acceleration methods are also presented. For the spectral and SQUAREM methods, ``time-dependent step sizes'' are introduced, based on the discussion of the variable-type-specific step sizes in Section \ref{subsec:Spectral}.\footnote{Supplemental Appendix of the current paper shows the performance of the time-dependent step sizes compared to the case of time-independent step sizes. The results imply that introducing time-dependent step sizes leads to faster convergence by several times.} This specification is motivated by the nonstationarity of the dynamic model. In the case of Algorithm \ref{alg:durable-perfect-foresight}, $V$ is updated by $V_{it}^{(n+1)}\leftarrow\alpha_{t}^{(n)}V_{it}^{(n+1)}+(1-\alpha_{t}^{(n)})V_{it}^{(n)}$, where 

$\alpha_{t}^{(n)}\equiv\frac{\left\Vert s_{t}^{(n)}\right\Vert _{2}}{\left\Vert y_{t}^{(n)}\right\Vert _{2}}$, $s_{t}^{(n)}\equiv V_{t}^{(n)}-V_{t}^{(n-1)},y_{t}^{(n)}\equiv F_{t}\left(V^{(n)}\right)-F_{t}\left(V^{(n-1)}\right),F_{t}(V)\equiv\text{\ensuremath{\Phi_{t}(V)-V_{t}}}$. As $\alpha_{t}^{(n)}$ occasionally take too large values (e.g., 100), which might lead to unstable convergence, the maximum value of $\alpha_{t}$ is set to 10.\footnote{Coding time-specific step sizes is not difficult in any programming language for any $T$. In the setting above, $V\equiv\left(V_{t=1}\cdots,V_{t=T}\right)$ is a $|\mathcal{I}|\times T$ dimensional variable. Then, $s^{(n)}\equiv V^{(n)}-V^{(n-1)}$ and $y^{(n)}\equiv F\left(V^{(n)}\right)-F\left(V^{(n-1)}\right)$ are also $|\mathcal{I}|\times T$ dimensional variables. Then, when applying the $\alpha_{S3}$-type step size, we can update the values of $V$ by $V^{(n+1)}\leftarrow\alpha^{(n)}V^{(n+1)}+(1-\alpha^{(n)})V^{(n)}$, where $\alpha^{(n)}\equiv\left(\begin{array}{ccc}
\alpha_{t=1}^{(n)} & \cdots & \alpha_{t=T}^{(n)}\end{array}\right)=\left(\begin{array}{ccc}
\left\Vert s_{t=1}^{(n)}\right\Vert _{2} & \cdots & \left\Vert s_{t=T}^{(n)}\right\Vert _{2}\end{array}\right)./\left(\begin{array}{ccc}
\left\Vert y_{t=1}^{(n)}\right\Vert _{2} & \cdots & \left\Vert y_{t=T}^{(n)}\right\Vert _{2}\end{array}\right)$. Here, ``$.$'' denotes the element-wise operation of matrices. $\left(\begin{array}{ccc}
\left\Vert s_{t=1}^{(n)}\right\Vert _{2} & \cdots & \left\Vert s_{t=T}^{(n)}\right\Vert _{2}\end{array}\right)$ can be computed by $\left(\begin{array}{ccc}
\left\Vert s_{t=1}^{(n)}\right\Vert _{2} & \cdots & \left\Vert s_{t=T}^{(n)}\right\Vert _{2}\end{array}\right)=$sqrt.$\left(\text{colsums}\left(s^{(n)}\text{.\textasciicircum2}\right)\right)$. Here, ``colsums'' denotes the sum of the values in each column of the matrix, and ``sqrt'' denotes the squared root.} 

Table \ref{tab:dynamic_BLP_Monte_Carlo_Perfect-foresight} shows that combining the acceleration methods significantly reduces the number of function evaluations and accelerates convergence. Among these acceleration methods, the Anderson acceleration method outperforms the spectral and SQUAREM. Furthermore, the new algorithms $V$-(0) and $V$-(1) require significantly fewer function evaluations than the traditional algorithms $V\delta$-(0) (joint) and $V\delta$-(1) (joint). In addition, the traditional algorithms $V\delta$-(0) (joint) and $V\delta$-(1) (joint) are not so stable, when the spectral/SQUAREM algorithms are combined.

Note that we cannot clearly see the prominently superior performance of $V$-(1) relative to $V$-(0) in the table. In the current setting, the values of the outside option CCPs are relatively large as shown in the table note, and the advantage of introducing the outside option shares in the updating equations is not necessarily very large. In contrast, the Supplemental Appendix additionally shows numerical results in a setting where the values of the outside option CCPs are smaller. It shows results where $V$-(1) with Anderson acceleration is on average 5$\sim$10\% faster than $V$-(0) with Anderson acceleration. Intuitively, the mapping $V$-(0) shares similarity with both the BLP contraction mapping and the value function iteration mapping in dynamic discrete choice models. The convergence speed of the former can be slow as the discount factor gets closer to 1. Hence, the introduction of the outside option shares in the updating equation may not drastically reduce the number of iterations when the discount factor is large, though it mitigates the convergence speed problem associated with the BLP contraction mapping. However, it can reduce the number of iterations, and it is helpful because dynamic BLP estimations are typically computationally burdensome.
\begin{center}
\begin{table}[H]
\caption{Results of the Dynamic BLP Monte Carlo simulation (Perfectly durable goods; Perfect foresight)\label{tab:dynamic_BLP_Monte_Carlo_Perfect-foresight}}

\scalebox{0.8}{
\begin{centering}
\begin{tabular}{cccccccccccc}
\hline 
 & \multirow{2}{*}{{\scriptsize{}$J$}} & \multicolumn{6}{c}{{\scriptsize{}Func. Evals.}} & {\scriptsize{}Mean} & {\scriptsize{}Conv.} & {\footnotesize{}Mean} & {\scriptsize{}$DIST<\epsilon_{tol}$}\tabularnewline
\cline{3-8} \cline{4-8} \cline{5-8} \cline{6-8} \cline{7-8} \cline{8-8} 
 &  & {\scriptsize{}Mean} & {\scriptsize{}Min.} & {\scriptsize{}25th} & {\scriptsize{}Median.} & {\scriptsize{}75th} & {\scriptsize{}Max.} & {\scriptsize{}CPU time (s)} & {\scriptsize{}(\%)} & {\footnotesize{}$\log_{10}\left(DIST\right)$} & {\scriptsize{}(\%)}\tabularnewline
\hline 
{\scriptsize{}$V$-(0)} & {\scriptsize{}25} & {\scriptsize{}2576.15} & {\scriptsize{}2505} & {\scriptsize{}2549} & {\scriptsize{}2577} & {\scriptsize{}2600} & {\scriptsize{}2632} & {\scriptsize{}4.76295} & {\scriptsize{}100} & {\scriptsize{}-13.9} & {\scriptsize{}100}\tabularnewline
{\scriptsize{}$V$-(0) + Anderson} & {\scriptsize{}25} & {\scriptsize{}442.6} & {\scriptsize{}363} & {\scriptsize{}395} & {\scriptsize{}427} & {\scriptsize{}501.5} & {\scriptsize{}531} & {\scriptsize{}0.8638} & {\scriptsize{}100} & {\scriptsize{}-14.5} & {\scriptsize{}100}\tabularnewline
{\scriptsize{}$V$-(0) + Spectral} & {\scriptsize{}25} & {\scriptsize{}719.1} & {\scriptsize{}430} & {\scriptsize{}447.5} & {\scriptsize{}477.5} & {\scriptsize{}514.5} & {\scriptsize{}2177} & {\scriptsize{}1.31145} & {\scriptsize{}100} & {\scriptsize{}-16.3} & {\scriptsize{}100}\tabularnewline
{\scriptsize{}$V$-(0) + SQUAREM} & {\scriptsize{}25} & {\scriptsize{}641.25} & {\scriptsize{}353} & {\scriptsize{}370} & {\scriptsize{}377.5} & {\scriptsize{}392.5} & {\scriptsize{}2204} & {\scriptsize{}1.18765} & {\scriptsize{}100} & {\scriptsize{}-15.8} & {\scriptsize{}100}\tabularnewline
{\scriptsize{}$V$-(1)} & {\scriptsize{}25} & {\scriptsize{}2535.3} & {\scriptsize{}2486} & {\scriptsize{}2530} & {\scriptsize{}2538.5} & {\scriptsize{}2544.5} & {\scriptsize{}2564} & {\scriptsize{}4.70105} & {\scriptsize{}100} & {\scriptsize{}-14} & {\scriptsize{}100}\tabularnewline
{\scriptsize{}$V$-(1) + Anderson} & {\scriptsize{}25} & {\scriptsize{}450.3} & {\scriptsize{}335} & {\scriptsize{}394} & {\scriptsize{}450.5} & {\scriptsize{}492} & {\scriptsize{}629} & {\scriptsize{}0.91515} & {\scriptsize{}100} & {\scriptsize{}-14.6} & {\scriptsize{}100}\tabularnewline
{\scriptsize{}$V$-(1) + Spectral} & {\scriptsize{}25} & {\scriptsize{}709.15} & {\scriptsize{}398} & {\scriptsize{}449.5} & {\scriptsize{}474.5} & {\scriptsize{}499.5} & {\scriptsize{}2131} & {\scriptsize{}1.3495} & {\scriptsize{}100} & {\scriptsize{}-16.3} & {\scriptsize{}100}\tabularnewline
{\scriptsize{}$V$-(1) + SQUAREM} & {\scriptsize{}25} & {\scriptsize{}637.6} & {\scriptsize{}359} & {\scriptsize{}367.5} & {\scriptsize{}375} & {\scriptsize{}400.5} & {\scriptsize{}2168} & {\scriptsize{}1.20105} & {\scriptsize{}100} & {\scriptsize{}-15.8} & {\scriptsize{}100}\tabularnewline
\hline 
{\scriptsize{}$V\delta$-(0) (joint)} & {\scriptsize{}25} & {\scriptsize{}2886.55} & {\scriptsize{}2739} & {\scriptsize{}2830.5} & {\scriptsize{}2861} & {\scriptsize{}3000} & {\scriptsize{}3000} & {\scriptsize{}4.0452} & {\scriptsize{}70} & {\scriptsize{}-13.5} & {\scriptsize{}100}\tabularnewline
{\scriptsize{}$V\delta$-(0) (joint) + Anderson} & {\scriptsize{}25} & {\scriptsize{}672.35} & {\scriptsize{}467} & {\scriptsize{}522} & {\scriptsize{}621} & {\scriptsize{}712.5} & {\scriptsize{}1303} & {\scriptsize{}1.00665} & {\scriptsize{}100} & {\scriptsize{}-14.1} & {\scriptsize{}100}\tabularnewline
{\scriptsize{}$V\delta$-(0) (joint) + Spectral} & {\scriptsize{}25} & {\scriptsize{}2220.6} & {\scriptsize{}42} & {\scriptsize{}2065} & {\scriptsize{}2393} & {\scriptsize{}2783} & {\scriptsize{}3000} & {\scriptsize{}3.30755} & {\scriptsize{}70} & {\scriptsize{}-12.4} & {\scriptsize{}85}\tabularnewline
{\scriptsize{}$V\delta$-(0) (joint) + SQUAREM} & {\scriptsize{}25} & {\scriptsize{}1212.9} & {\scriptsize{}8} & {\scriptsize{}28} & {\scriptsize{}1545} & {\scriptsize{}1705} & {\scriptsize{}2934} & {\scriptsize{}1.452} & {\scriptsize{}65} & {\scriptsize{}-9.8} & {\scriptsize{}65}\tabularnewline
{\scriptsize{}$V\delta$-(1) (joint)} & {\scriptsize{}25} & {\scriptsize{}2876.7} & {\scriptsize{}2729} & {\scriptsize{}2819} & {\scriptsize{}2845} & {\scriptsize{}3000} & {\scriptsize{}3000} & {\scriptsize{}3.7986} & {\scriptsize{}70} & {\scriptsize{}-13.6} & {\scriptsize{}100}\tabularnewline
{\scriptsize{}$V\delta$-(1) (joint) + Anderson} & {\scriptsize{}25} & {\scriptsize{}659.45} & {\scriptsize{}479} & {\scriptsize{}531.5} & {\scriptsize{}576.5} & {\scriptsize{}708} & {\scriptsize{}1311} & {\scriptsize{}1.01485} & {\scriptsize{}100} & {\scriptsize{}-14.1} & {\scriptsize{}100}\tabularnewline
{\scriptsize{}$V\delta$-(1) (joint) + Spectral} & {\scriptsize{}25} & {\scriptsize{}2364.9} & {\scriptsize{}42} & {\scriptsize{}2168} & {\scriptsize{}2462} & {\scriptsize{}2986.5} & {\scriptsize{}3000} & {\scriptsize{}3.3087} & {\scriptsize{}70} & {\scriptsize{}-13.2} & {\scriptsize{}90}\tabularnewline
{\scriptsize{}$V\delta$-(1) (joint) + SQUAREM} & {\scriptsize{}25} & {\scriptsize{}1406.6} & {\scriptsize{}8} & {\scriptsize{}655} & {\scriptsize{}1571} & {\scriptsize{}2035} & {\scriptsize{}2986} & {\scriptsize{}1.70135} & {\scriptsize{}75} & {\scriptsize{}-11.2} & {\scriptsize{}75}\tabularnewline
\hline 
\end{tabular}
\par\end{centering}
}

{\footnotesize{}\input{notes/note_dynamic_BLP.tex}}{\footnotesize\par}

{\footnotesize{}The minimum and median outside option CCPs are 0.317 and 0.993 respectively.}{\footnotesize\par}
\end{table}
\par\end{center}

\subsubsection{Model under Inclusive value sufficiency}

Next, we examine a setting where consumers form expectations based on their inclusive values. Generally, $\Omega_{t}$,\-\- which includes product characteristics, is a high-dimensional variable, and fully specifying their stochastic state transitions becomes impractical, especially when dealing with tens of products. Hence, many previous studies (e.g., \citealp{hendel2006measuring}; \citealp{gowrisankaran2012dynamics}) have adopted the idea of inclusive value sufficiency. Specifically, we alternatively use inclusive value $\omega_{it}\equiv\log\left(\sum_{k\in\mathcal{J}_{t}}\exp(\delta_{kt}+\mu_{ikt})\right)$ as the state of the dynamic problem. Under this setting, consumers form expectations based on the transitions of the state $\omega_{it}$.

In this study, the value of $T$ is set to 25, as in \citet{Sun2019}. Besides, unlike the perfect foresight scenario, the model under inclusive value sufficiency operates within a mostly stationary framework. A common scalar $\alpha^{(n)}$ is used in the spectral/SQUAREM algorithms. 

Because the algorithm becomes more complex due to the need to take grid points, the complete steps of the algorithm are detailed in the Supplemental Appendix. Nevertheless, the idea of analytically representing $\delta$ as a function of $V$ remains unchanged.

Table \ref{tab:dynamic_BLP_Monte_Carlo_IVS} shows the results. As in the case under perfect foresight, combining the acceleration methods speeds up convergence, and the Anderson acceleration outperforms the others. In addition, the new algorithms $V$-(0) and $V$-(1) are faster than $V\delta$-(0) and $V\delta$-(1), especially when combining the spectral or SQUAREM algorithms.
\begin{center}
\begin{table}[H]
\caption{Results of the Dynamic BLP Monte Carlo simulation (Perfectly durable goods; Inclusive value sufficiency)\label{tab:dynamic_BLP_Monte_Carlo_IVS}}

\scalebox{0.8}{
\begin{centering}
\begin{tabular}{cccccccccccc}
\hline 
 & \multirow{2}{*}{{\scriptsize{}$J$}} & \multicolumn{6}{c}{{\scriptsize{}Func. Evals.}} & {\scriptsize{}Mean} & {\scriptsize{}Conv.} & {\footnotesize{}Mean} & {\scriptsize{}$DIST<\epsilon_{tol}$}\tabularnewline
\cline{3-8} \cline{4-8} \cline{5-8} \cline{6-8} \cline{7-8} \cline{8-8} 
 &  & {\scriptsize{}Mean} & {\scriptsize{}Min.} & {\scriptsize{}25th} & {\scriptsize{}Median.} & {\scriptsize{}75th} & {\scriptsize{}Max.} & {\scriptsize{}CPU time (s)} & {\scriptsize{}(\%)} & {\footnotesize{}$\log_{10}\left(DIST\right)$} & {\scriptsize{}(\%)}\tabularnewline
\hline 
{\scriptsize{}$V$-(0)} & {\scriptsize{}25} & {\scriptsize{}2462.15} & {\scriptsize{}2017} & {\scriptsize{}2242} & {\scriptsize{}2399} & {\scriptsize{}2644.5} & {\scriptsize{}3000} & {\scriptsize{}8.82705} & {\scriptsize{}90} & {\scriptsize{}-14} & {\scriptsize{}100}\tabularnewline
{\scriptsize{}$V$-(0) + Anderson} & {\scriptsize{}25} & {\scriptsize{}247.55} & {\scriptsize{}216} & {\scriptsize{}226.5} & {\scriptsize{}235.5} & {\scriptsize{}269.5} & {\scriptsize{}311} & {\scriptsize{}0.94705} & {\scriptsize{}100} & {\scriptsize{}-15.1} & {\scriptsize{}100}\tabularnewline
{\scriptsize{}$V$-(0) + Spectral} & {\scriptsize{}25} & {\scriptsize{}540.25} & {\scriptsize{}451} & {\scriptsize{}497.5} & {\scriptsize{}525} & {\scriptsize{}579} & {\scriptsize{}687} & {\scriptsize{}2.0182} & {\scriptsize{}100} & {\scriptsize{}-14.4} & {\scriptsize{}100}\tabularnewline
{\scriptsize{}$V$-(0) + SQUAREM} & {\scriptsize{}25} & {\scriptsize{}527.8} & {\scriptsize{}396} & {\scriptsize{}455} & {\scriptsize{}539} & {\scriptsize{}602} & {\scriptsize{}660} & {\scriptsize{}1.9194} & {\scriptsize{}100} & {\scriptsize{}-14.2} & {\scriptsize{}100}\tabularnewline
{\scriptsize{}$V$-(1)} & {\scriptsize{}25} & {\scriptsize{}2458.35} & {\scriptsize{}1996} & {\scriptsize{}2236.5} & {\scriptsize{}2396} & {\scriptsize{}2642} & {\scriptsize{}3000} & {\scriptsize{}9.47235} & {\scriptsize{}90} & {\scriptsize{}-14} & {\scriptsize{}100}\tabularnewline
{\scriptsize{}$V$-(1) + Anderson} & {\scriptsize{}25} & {\scriptsize{}238.1} & {\scriptsize{}192} & {\scriptsize{}214.5} & {\scriptsize{}239.5} & {\scriptsize{}258} & {\scriptsize{}296} & {\scriptsize{}0.93795} & {\scriptsize{}100} & {\scriptsize{}-15.2} & {\scriptsize{}100}\tabularnewline
{\scriptsize{}$V$-(1) + Spectral} & {\scriptsize{}25} & {\scriptsize{}557.35} & {\scriptsize{}450} & {\scriptsize{}524.5} & {\scriptsize{}545} & {\scriptsize{}588} & {\scriptsize{}752} & {\scriptsize{}2.20445} & {\scriptsize{}100} & {\scriptsize{}-14.5} & {\scriptsize{}100}\tabularnewline
{\scriptsize{}$V$-(1) + SQUAREM} & {\scriptsize{}25} & {\scriptsize{}564.75} & {\scriptsize{}430} & {\scriptsize{}505} & {\scriptsize{}572} & {\scriptsize{}599} & {\scriptsize{}766} & {\scriptsize{}2.1251} & {\scriptsize{}100} & {\scriptsize{}-14.3} & {\scriptsize{}100}\tabularnewline
\hline 
{\scriptsize{}$V\delta$-(0) (joint)} & {\scriptsize{}25} & {\scriptsize{}2482.75} & {\scriptsize{}2036} & {\scriptsize{}2274} & {\scriptsize{}2417} & {\scriptsize{}2667.5} & {\scriptsize{}3000} & {\scriptsize{}9.0842} & {\scriptsize{}90} & {\scriptsize{}-13.9} & {\scriptsize{}100}\tabularnewline
{\scriptsize{}$V\delta$-(0) (joint) + Anderson} & {\scriptsize{}25} & {\scriptsize{}242.8} & {\scriptsize{}201} & {\scriptsize{}224} & {\scriptsize{}244.5} & {\scriptsize{}254.5} & {\scriptsize{}284} & {\scriptsize{}0.91525} & {\scriptsize{}100} & {\scriptsize{}-15.1} & {\scriptsize{}100}\tabularnewline
{\scriptsize{}$V\delta$-(0) (joint) + Spectral} & {\scriptsize{}25} & {\scriptsize{}1131.45} & {\scriptsize{}913} & {\scriptsize{}1017.5} & {\scriptsize{}1097} & {\scriptsize{}1182} & {\scriptsize{}1647} & {\scriptsize{}4.0791} & {\scriptsize{}100} & {\scriptsize{}-14.5} & {\scriptsize{}100}\tabularnewline
{\scriptsize{}$V\delta$-(0) (joint) + SQUAREM} & {\scriptsize{}25} & {\scriptsize{}1162.3} & {\scriptsize{}920} & {\scriptsize{}1030} & {\scriptsize{}1112} & {\scriptsize{}1224} & {\scriptsize{}1684} & {\scriptsize{}4.24885} & {\scriptsize{}100} & {\scriptsize{}-14.3} & {\scriptsize{}100}\tabularnewline
{\scriptsize{}$V\delta$-(1) (joint)} & {\scriptsize{}25} & {\scriptsize{}2479.8} & {\scriptsize{}2029} & {\scriptsize{}2270} & {\scriptsize{}2414.5} & {\scriptsize{}2664} & {\scriptsize{}3000} & {\scriptsize{}9.0023} & {\scriptsize{}90} & {\scriptsize{}-13.9} & {\scriptsize{}100}\tabularnewline
{\scriptsize{}$V\delta$-(1) (joint) + Anderson} & {\scriptsize{}25} & {\scriptsize{}233.5} & {\scriptsize{}196} & {\scriptsize{}207} & {\scriptsize{}223.5} & {\scriptsize{}251.5} & {\scriptsize{}306} & {\scriptsize{}0.85245} & {\scriptsize{}100} & {\scriptsize{}-15.1} & {\scriptsize{}100}\tabularnewline
{\scriptsize{}$V\delta$-(1) (joint) + Spectral} & {\scriptsize{}25} & {\scriptsize{}1139.3} & {\scriptsize{}904} & {\scriptsize{}1024.5} & {\scriptsize{}1130.5} & {\scriptsize{}1218.5} & {\scriptsize{}1641} & {\scriptsize{}4.31415} & {\scriptsize{}100} & {\scriptsize{}-14.4} & {\scriptsize{}100}\tabularnewline
{\scriptsize{}$V\delta$-(1) (joint) + SQUAREM} & {\scriptsize{}25} & {\scriptsize{}1168.3} & {\scriptsize{}936} & {\scriptsize{}1037} & {\scriptsize{}1138} & {\scriptsize{}1233} & {\scriptsize{}1706} & {\scriptsize{}4.15965} & {\scriptsize{}100} & {\scriptsize{}-14.4} & {\scriptsize{}100}\tabularnewline
\hline 
\end{tabular}
\par\end{centering}
}

{\footnotesize{}\input{notes/note_dynamic_BLP.tex}}{\footnotesize\par}

{\footnotesize{}The minimum and median outside option CCPs are 0.693 and 0.998 respectively.}{\footnotesize\par}
\end{table}
\par\end{center}

\section{Conclusion\label{sec:Conclusion}}

\input{conclusion.tex}

\pagebreak{}

\appendix

\section{Additional results and discussions\label{sec:Additional-results}}

\subsection{Convergence properties of the mappings (Static BLP)\label{subsec:Convergence-properties-mappings}}

\subsubsection{Convergence properties of $\Phi^{\delta,\gamma}$}

The following proposition shows the convergence properties of the mapping $\Phi^{\delta,\gamma}$.
\begin{prop}
\label{prop:property_mapping_delta}

(a). Under $\gamma\in[0,1]$, the following inequalities hold: 
\begin{eqnarray*}
\left\Vert \Phi^{\delta,\gamma}(\delta_{1})-\Phi^{\delta,\gamma}(\delta_{2})\right\Vert _{\infty} & \leq & c_{\gamma}\left\Vert \delta_{1}-\delta_{2}\right\Vert _{\infty}\\
\left\Vert \Phi^{\delta,\gamma}(\delta_{1})-\Phi^{\delta,\gamma}(\delta_{2})\right\Vert _{\infty} & \leq & \widetilde{c_{\gamma}}\left(\left\Vert \max\{\delta_{1},\delta_{2}\}-\delta_{1}\right\Vert _{\infty}+\left\Vert \max\{\delta_{1},\delta_{2}\}-\delta_{2}\right\Vert _{\infty}\right)\\
 & \leq & 2\widetilde{c_{\gamma}}\cdot\left\Vert \delta_{1}-\delta_{2}\right\Vert _{\infty}
\end{eqnarray*}

where $c_{\gamma}\equiv\sup_{j\in\mathcal{J},\delta\in B_{\delta}}\sum_{k\in\mathcal{J}}\left|\sum_{i}s_{ik}\left(\frac{w_{i}s_{ij}(\delta)}{s_{j}(\delta)}-\gamma\frac{w_{i}s_{i0}(\delta)}{s_{0}(\delta)}\right)\right|$ and

$\widetilde{c_{\gamma}}\equiv\begin{cases}
\sup_{i\in\mathcal{I},J^{*}\subset\mathcal{J},\delta}s_{iJ^{*}}(\delta)=\sup_{i\in\mathcal{I},\delta}\left(1-s_{i0}(\delta)\right)\leq1 & \text{if}\ \gamma=0\\
\sup_{i\in\mathcal{I},J^{*}\subset\mathcal{J},\delta}s_{iJ^{*}}(\delta)-\inf_{i\in\mathcal{I},J^{*}\subset\mathcal{J},\delta}s_{iJ^{*}}(\delta)\leq1 & \text{if}\ \gamma=1\\
\max\left\{ \sup_{i\in\mathcal{I},J^{*}\subset\mathcal{J},\delta\in B_{\delta}}s_{iJ^{*}}(\delta)-\gamma\inf_{i\in\mathcal{I},J^{*}\subset\mathcal{J},\delta\in B_{\delta}}s_{iJ^{*}}(\delta),\right.\\
\left.\ \ \ \gamma\sup_{i\in\mathcal{I},J^{*}\subset\mathcal{J},\delta\in B_{\delta}}s_{iJ^{*}}(\delta)-\inf_{i\in\mathcal{I},J^{*}\subset\mathcal{J},\delta\in B_{\delta}}s_{iJ^{*}}(\delta)\right\}  & \text{Otherwise}
\end{cases}$

(b). Under $\gamma=0$,

$\left\Vert \Phi^{\delta,\gamma=0}(\delta_{1})-\Phi^{\delta,\gamma=0}(\delta_{2})\right\Vert _{\infty}\leq\sup_{i\in\mathcal{I},\delta\in B_{\delta}}\left(1-s_{i0}(\delta)\right)\cdot\left\Vert \delta_{1}-\delta_{2}\right\Vert _{\infty}\leq\left\Vert \delta_{1}-\delta_{2}\right\Vert _{\infty}$ holds.

(c). If $|\mathcal{J}|=1,$

\begin{eqnarray*}
\left\Vert \Phi^{\delta,\gamma}(\delta_{1})-\Phi^{\delta,\gamma}(\delta_{2})\right\Vert _{\infty} & \leq & \widetilde{c_{\gamma}}\cdot\left\Vert \delta_{1}-\delta_{2}\right\Vert _{\infty}\leq\left\Vert \delta_{1}-\delta_{2}\right\Vert _{\infty}
\end{eqnarray*}
\end{prop}
First, we consider the case of $|\mathcal{J}|=1$ (only one product). The proposition implies $\Phi^{\delta,\gamma=0}$ (BLP contraction mapping) has a contraction property\footnote{See also Remark \ref{rem:contraction_remark}.} with modulus $\sup_{i\in\mathcal{I},J^{*}\subset\mathcal{J}}s_{iJ^{*}}(\delta)=1-\inf_{i\in\mathcal{I}}s_{i0}(\delta)$. In contrast, $\Phi^{\delta,\gamma=1}$ has a contraction property with modulus $\sup_{i\in\mathcal{I},J^{*}\subset\mathcal{J}}s_{iJ^{*}}(\delta)-\inf_{i\in\mathcal{I},J^{*}\subset\mathcal{J}}s_{iJ^{*}}(\delta)<\sup_{i\in\mathcal{I},J^{*}\subset\mathcal{J}}s_{iJ^{*}}(\delta)$. Hence, we can expect that the latter attains faster convergence.

When $|\mathcal{J}|>1$ holds, there is no guarantee that $c_{\gamma}$ is less than 1. Hence, there is no guarantee that $\Phi^{\delta,\gamma=1}$ is a contraction. Nevertheless, $c_{\gamma=1}$ takes a small value when the consumer heterogeneity, measured by $\frac{s_{ij}}{s_{j}}-\frac{s_{i0}}{s_{0}}$, is small. If there is no consumer heterogeneity, $\frac{s_{ij}}{s_{j}}-\frac{s_{i0}}{s_{0}}=0$ holds, and the modulus is equal to 0. This implies that the iteration using the mapping immediately converges to the solution. The proposition also implies that the magnitude of consumer heterogeneity is important for the convergence speed of $\Phi^{\delta,\gamma=1}$. 

Besides, the inequality $\left\Vert \Phi^{\delta,\gamma}(\delta_{1})-\Phi^{\delta,\gamma}(\delta_{2})\right\Vert _{\infty}\leq\widetilde{c_{\gamma}}\left(\left\Vert \max\{\delta_{1},\delta_{2}\}-\delta_{1}\right\Vert _{\infty}+\left\Vert \max\{\delta_{1},\delta_{2}\}-\delta_{2}\right\Vert _{\infty}\right)$ for $|\mathcal{J}|\in\mathbb{N}$ implies 

$\left\Vert \Phi^{\delta,\gamma}(\delta_{1})-\Phi^{\delta,\gamma}(\delta_{2})\right\Vert _{\infty}\leq\widetilde{c_{\gamma}}\left\Vert \delta_{1}-\delta_{2}\right\Vert _{\infty}=\left(\sup_{i\in\mathcal{I},J^{*}\subset\mathcal{J},\delta}s_{iJ^{*}}(\delta)-\inf_{i\in\mathcal{I},J^{*}\subset\mathcal{J},\delta}s_{iJ^{*}}(\delta)\right)\left\Vert \delta_{1}-\delta_{2}\right\Vert _{\infty}\leq\left\Vert \delta_{1}-\delta_{2}\right\Vert _{\infty}$ under $\delta_{1}\geq\delta_{2}$. Hence, $\Phi^{V,\gamma=1}$ has a property similar to contractions.
\begin{rem}
\label{rem:contraction_remark}If we assume $B_{\delta}=\mathbb{R}^{|\delta|}$, the value of $\sup_{i\in\mathcal{I},\delta\in B_{\delta}}\left(1-s_{i0}(\delta)\right)$ can be 1, which implies $\Phi^{\delta,\gamma=0}$ is not a contraction, strictly speaking. Hence, we should restrict the domain or the range of the mapping $\Phi^{\delta,\gamma=0}$ to guarantee that the mapping is a contraction, as discussed in the Appendix of \citet{berry1995automobile}. 
\end{rem}

\subsubsection{Convergence properties of $\Phi^{\delta,\gamma}$}

We define the following terms:

\begin{eqnarray*}
prob_{0|i_{2}}(V) & \equiv & \frac{1}{1+\sum_{j\in\mathcal{J}}S_{j}^{(data)}\frac{\exp(\mu_{i_{2}j})}{\sum_{i\in\mathcal{I}}w_{i}\exp(\mu_{ij})\exp(-V_{i})}},\\
prob_{I^{*}|j}(V) & \equiv & \frac{\sum_{i\in I^{*}\subset\mathcal{I}}w_{i}\exp(\mu_{ij})\exp(-V_{i})}{\sum_{i\in\mathcal{I}}w_{i}\exp(\mu_{ij})\exp(-V_{i})},\\
prob_{I^{*}|0}(V) & \equiv & \frac{\sum_{i\in I^{*}\subset\mathcal{I}}w_{i}\exp(-V_{i})}{\sum_{i\in\mathcal{I}}w_{i}\exp(-V_{i})}.
\end{eqnarray*}

Intuitively, $prob_{0|i}$ denotes the probability that consumer $i$ chooses alternative $j$. $prob_{I^{*}\subset\mathcal{I}|j}$ denotes the fraction of consumers in group $I^{*}\subset\mathcal{I}$ choosing alternative $j$ among those choosing $j$. 

Then, we obtain the following results regarding the convergence properties of $\Phi^{V,\gamma}$. 
\begin{prop}
\label{prop:property_mapping_V}

(a). Under $\gamma\in[0,1]$, the following inequalities hold: 
\begin{eqnarray*}
\left\Vert \Phi^{V,\gamma}(V_{1})-\Phi^{V,\gamma}(V_{2})\right\Vert _{\infty} & \leq & 2\sup_{i_{2}\in\mathcal{I},V\in B_{V}}\left((1-prob_{0|i_{2}}(V))\cdot\left(\frac{\sum_{i\in\mathcal{I}}w_{i}\exp(-V_{i})}{S_{0}^{(data)}}\right)^{\gamma}\right)\cdot\left\Vert V_{1}-V_{2}\right\Vert _{\infty},\\
\left\Vert \Phi^{V,\gamma}(V_{1})-\Phi^{V,\gamma}(V_{2})\right\Vert _{\infty} & \leq & \sup_{i\in\mathcal{I},V\in B_{V}}\left((1-prob_{0|i}(V))\cdot\left(\frac{\sum_{i\in\mathcal{I}}w_{i}\exp(-V_{i})}{S_{0}^{(data)}}\right)^{\gamma}\right)\cdot\\
 &  & \sup_{I^{*}\subset\mathcal{I},j\in\mathcal{J},V\in B_{V}}\left|prob_{I^{*}|j}(V)-\gamma prob_{I^{*}|0}(V)\right|\cdot\\
 &  & \left(\left\Vert \max\{V_{1},V_{2}\}-V_{1}\right\Vert _{\infty}+\left\Vert \max\{V_{1},V_{2}\}-V_{2}\right\Vert _{\infty}\right).
\end{eqnarray*}

(b). Under $\gamma=0$, $\left\Vert \Phi^{V,\gamma}(V_{1})-\Phi^{V,\gamma}(V_{2})\right\Vert _{\infty}\leq\sup_{V\in B_{V},i_{2}\in\mathcal{I}}\left((1-prob_{0|i_{2}}(V))\right)\cdot\left\Vert V_{1}-V_{2}\right\Vert _{\infty}\leq\left\Vert V_{1}-V_{2}\right\Vert _{\infty}$ holds. 
\end{prop}
The result in (b) corresponds to the contraction property of the BLP contraction mapping $\Phi^{\delta,\gamma=0}$. Regarding the results in (a), the first inequality implies $\Phi^{V,\gamma}$ is not always a contraction mapping if $2\sup_{i_{2}\in\mathcal{I},V\in B_{V}}\left((1-prob_{0|i_{2}}(V))\cdot\left(\frac{\sum_{i\in\mathcal{I}}w_{i}\exp(-V_{i})}{S_{0}^{(data)}}\right)^{\gamma}\right)>1$. Nevertheless, the second inequality shows it has a property similar to contraction mappings and has good convergence properties. Here, consider the setting where $\frac{\sum_{i\in\mathcal{I}}w_{i}\exp(-V_{i})}{S_{0}^{(data)}}\approx1$, $V^{*}>V$, and $\gamma=1$. Then, $\left\Vert \Phi^{V,\gamma}(V)-\Phi^{V,\gamma}(V^{*})\right\Vert _{\infty}\leq\sup_{i\in\mathcal{I},B\in B_{V}}\left((1-prob_{0|i}(V))\right)\cdot\sup_{I^{*}\subset\mathcal{I},j\in\mathcal{J},V\in B_{V}}\left|prob_{I^{*}|j}(V)-prob_{I^{*}|0}(V)\right|\cdot\left\Vert V^{*}-V\right\Vert _{\infty}$ holds. $\sup_{I^{*}\subset\mathcal{I},j\in\mathcal{J},V\in B_{V}}\left|prob_{I^{*}|j}(V)-prob_{I^{*}|0}(V)\right|$ is the maximum absolute value of the difference between $prob_{I^{*}|0}(V)$ and $prob_{I^{*}|j}(V)$, and it is less than 1. Hence, $\left\Vert \Phi^{V,\gamma}(V)-\Phi^{V,\gamma}(V^{*})\right\Vert _{\infty}\leq\exists K\left\Vert V^{*}-V\right\Vert _{\infty}\leq\left\Vert V^{*}-V\right\Vert _{\infty}\ (K\leq1)$ holds under the setting. 

Besides, $prob_{I^{*}|j}(V)-prob_{I^{*}|0}(V)$ take small values when the consumer heterogeneity is small, implying we can attain high convergence speed in the model with small consumer heterogeneity. If the size of consumer heterogeneity is zero as in logit models without random coefficients, the mapping is equivalent to $\Phi^{V,\gamma=1}(V)=\log\left(1+\left(\sum_{j\in\mathcal{J}}S_{j}^{(data)}\right)\cdot\left(\frac{1}{S_{0}^{(data)}}\right)\right)=-\log\left(S_{0}^{(data)}\right)$, whose right-hand side is the solution of the fixed-point problem, and it immediately converges to the solution after applying the mapping $\Phi^{V,\gamma=1}$ once.

\subsection{Global convergence in the static BLP model\label{subsec:Global-convergence}}

As previously mentioned, there is no guarantee that the mapping $\Phi^{\delta,\gamma=1}$ is a contraction. However, we can easily guarantee the global convergence in the static BLP model by slightly modifying the original algorithm using $\Phi^{\delta,\gamma=1}$, utilizing the fact that $\Phi^{\delta,\gamma=0}$ is a contraction. The stabilized algorithm is worth considering for the static BLP model if practitioners are conservative about the convergence of the algorithm, though the algorithms without any stabilization techniques converge in most cases. Algorithm \ref{alg:static_BLP_delta_stabilized} shows the steps.

\begin{algorithm}[H]
Set initial values of $\delta^{(0)}$ and a tuning parameter $\eta\in(0,1)$. Let $m\left(\delta\right)\equiv\left\Vert \log\left(S^{(data)}\right)-\log\left(s\left(\delta\right)\right)\right\Vert _{\infty}$. Iterate the following $(n=0,1,2,\cdots)$:
\begin{enumerate}
\item Let $\widetilde{\delta}^{(n+1)}\leftarrow\Phi^{\delta,\gamma=1}(\delta^{(n)})$
\item If $m\left(\widetilde{\delta}^{(n+1)}\right)\leq\eta m\left(\delta^{(n)}\right)$, let $\delta^{(n+1)}\leftarrow\widetilde{\delta}^{(n+1)}\equiv\Phi^{\delta,\gamma=1}(\delta^{(n)})$.

Otherwise, let $\delta^{(n+1)}\leftarrow\Phi^{\delta,\gamma=0}(\delta^{(n)})$.
\item Exit the iteration if $m\left(\delta^{(n+1)}\right)<\epsilon_{\delta}$
\end{enumerate}
\caption{Inner loop algorithm of static BLP using $\Phi^{\delta,\gamma=1}$ (stabilized version)\label{alg:static_BLP_delta_stabilized}}
\end{algorithm}

The difference with the original algorithm \ref{alg:static_BLP_delta} is Step 2. $m\left(\delta\right)\equiv\left\Vert \log\left(S^{(data)}\right)-\log\left(s\left(\delta\right)\right)\right\Vert _{\infty}$ can be interpreted as the measure of the deviation from the solution, and we do not accept the value of $\Phi^{\delta,\gamma=1}(\delta^{(n)})$ as $\delta^{(n+1)}$ if the value of $m\left(\delta\right)$ does not decrease. In this case, we alternatively use $\Phi^{\delta,\gamma=0}(\delta^{(n)})$ as $\delta^{(n+1)}$.\footnote{Analogous procedure has been used to guarantee the global convergence of fixed-point iterations (e.g., \citealp{varadhan2008simple}, \citealp{zhang2020globally}).} We can implement the step by adding a few lines in the programming code. 

In the algorithm, we need to specify the value of a tuning parameter $\eta$, which determines the acceptance of $\Phi^{\delta,\gamma=1}(\delta^{(n)})$ as $\delta^{(n+1)}$. To avoid unnecessary rejection, values close to 1, such as 0.99, are recommended.\footnote{If we set $\eta=1$, we cannot rule out the possibility that the iteration converges to a point such that $m(\delta)>0$, strictly speaking. Hence, setting a value less than 1 is essential.} 

We can easily show that $\lim_{n\rightarrow\infty}\delta^{(n)}=\delta_{*}$ holds, where $\delta_{*}$ is the solution, by utilizing the fact that $\Phi^{\delta,\gamma=0}$ is a contraction. Let $\eta_{0}\in(0,1)$ be the modulus of the standard BLP contraction mapping $\Phi^{\delta,\gamma=0}$. Because $m\left(\delta\right)=\left\Vert \Phi^{\delta,\gamma=0}\left(\delta\right)-\delta\right\Vert _{\infty}$, 
\begin{eqnarray*}
m\left(\Phi^{\delta,\gamma=0}(\delta^{(n)})\right) & = & \left\Vert \Phi^{\delta,\gamma=0}\left(\Phi^{\delta,\gamma=0}(\delta^{(n)})\right)-\Phi^{\delta,\gamma=0}(\delta^{(n)})\right\Vert _{\infty}\\
 & \leq & \eta_{0}\left\Vert \Phi^{\delta,\gamma=0}(\delta^{(n)})-\delta^{(n)}\right\Vert _{\infty}\ \left(\because\text{Contraction property of }\Phi^{\delta,\gamma=0}\right)\\
 & = & \eta_{0}\cdot m\left(\delta^{(n)}\right)
\end{eqnarray*}

Hence, the algorithm satisfies $m\left(\delta^{(n+1)}\right)\leq\max\left\{ \eta_{0},\eta\right\} \cdot m\left(\delta^{(n)}\right)$, which implies $\lim_{n\rightarrow\infty}m\left(\delta^{(n)}\right)=\lim_{n\rightarrow\infty}\left(\max\left\{ \eta_{0},\eta\right\} \right)^{n}m\left(\delta^{(0)}\right)=0$ because $\eta_{0},\eta\in(0,1)$. Because $\delta$ satisfying $m\left(\delta\right)=0$ is unique, $\lim_{n\rightarrow\infty}\delta^{(n)}=\delta_{*}$ holds.

Although we have discussed the algorithm to guarantee the global convergence of $\Phi^{\delta,\gamma=1}$, the idea can also be used for the algorithms using acceleration methods, such as Anderson, spectral, and SQUAREM. In addition, the idea can be also used for $\Phi^{V,\gamma=1}$ in the static BLP model, utilizing the fact that $\Phi^{V,\gamma=0}$ is a contraction.

\subsection{Static Random Coefficient Nested Logit (RCNL) model\label{subsec:Extension:-RCNL-model}}

In the main part of this article, we have considered BLP models without nest structure (RCL models). Nevertheless, we can easily extend the discussion to random coefficient nested logit (RCNL) models, where the nest structure is introduced in the RCL models.

\subsubsection{Model}

Let consumer $i$'s utility when buying product $j$ be $v_{ij}=\delta_{j}+\mu_{ij}+\epsilon_{ij}$, and utility when buying nothing be $v_{i0}=\epsilon_{i0}$. $\delta_{j}$ denotes product $j$'s mean utility, and $\mu_{ij}$ denotes consumer $i$-specific utility of product $j$. $\epsilon$ denotes idiosyncratic utility shocks. Let $\mathcal{J}$ be the set of products. Here, we assume the following distributional assumption on $\epsilon_{ij}$:

\begin{eqnarray*}
\epsilon_{ij} & = & \bar{\xi}_{ig}+(1-\rho)\widetilde{\epsilon_{ij}}\ \ \ (j\in\mathcal{J}_{g}),
\end{eqnarray*}
where $\widetilde{\epsilon_{ij}}$ is distributed i.i.d. mean zero type-I extreme value, and $\bar{\xi}_{ig}$ is such that $\epsilon_{ij}$ is distributed extreme value. $\rho$ denotes a nest parameter such that $\rho\in[0,1)$.

Let $\mathcal{G}$ be the set of nests. We assume $\mathcal{G}$ does not include the outside option. Let $\mathcal{J}_{g}$ be the set of products in nest $g\in\mathcal{G}$. By definition, $\mathcal{J}=\cup_{g\in\mathcal{G}}\mathcal{J}_{g}$ holds.

Then, consumer $i$'s choice probability of product $j$ in nest $g$ is:

\begin{eqnarray*}
s_{ij} & = & \frac{\exp\left(\frac{\delta_{j}+\mu_{ij}}{1-\rho}\right)}{\exp\left(\frac{IV_{ig}}{1-\rho}\right)}\frac{\exp\left(IV_{ig}\right)}{1+\sum_{g\in\mathcal{G}}\exp\left(IV_{ig}\right)}
\end{eqnarray*}
Here, $IV_{ig}\equiv(1-\rho)\log\left(\sum_{j\in\mathcal{J}_{g}}\exp\left(\frac{\delta_{j}+\mu_{ij}}{1-\rho}\right)\right)$ denotes consumer $i$'s inclusive value of nest $g\in\mathcal{G}$.

Consumer $i$'s choice probability of the outside option is:

\begin{eqnarray*}
s_{i0} & = & \frac{1}{1+\sum_{g\in\mathcal{G}}\exp\left(IV_{ig}\right)}
\end{eqnarray*}

The market share of product $j$ is represented as $s_{j}=\sum_{i\in\mathcal{I}}w_{i}s_{ij}$. 

\subsubsection{Mappings of $\delta$}

We define the following mapping: 

\begin{eqnarray*}
\Phi_{j}^{\delta,\gamma}\left(\delta\right) & = & \delta_{j}+\left(1-\rho\right)\left[\log\left(S_{j}^{(data)}\right)-\log\left(s_{j}(\delta)\right)\right]+\\
 &  & \gamma\rho\left[\log\left(S_{g}^{(data)}\right)-\log\left(s_{g}(\delta)\right)\right]-\gamma\left[\log\left(S_{0}^{(data)}\right)-\log\left(s_{0}(\delta)\right)\right]
\end{eqnarray*}

Here, we define $s_{j}(\delta)\equiv\sum_{i}w_{i}\frac{\exp\left(\frac{\delta_{j}+\mu_{ij}}{1-\rho}\right)}{\exp\left(\frac{IV_{ig}}{1-\rho}\right)}\frac{\exp\left(IV_{ig}\right)}{1+\sum_{g\in\mathcal{G}}\exp\left(IV_{ig}\right)},$$s_{g}(\delta)\equiv\sum_{j\in\mathcal{J}_{g}}s_{j}(\delta)$ where $IV_{ig}=(1-\rho)\log\left(\sum_{j\in\mathcal{J}_{g}}\exp\left(\frac{\delta_{j}+\mu_{ij}}{1-\rho}\right)\right)$ and $s_{0}(\delta)\equiv1-\sum_{g\in\mathcal{G}}s_{g}(\delta)$.

The mapping $\Phi_{j}^{\delta,\gamma=0}$ corresponds to the traditional one considered in \citet{iizuka2007experts}, \citet{grigolon2014nested}, and \citet{conlon2020best}.\footnote{As pointed out in \citet{conlon2020best}, there is a typesetting error in the mapping in Equation (15) of \citet{grigolon2014nested}.}

The following proposition justifies the use of the mapping $\Phi^{\delta,\gamma\geq0}$:
\begin{prop}
(Static RCNL model)\label{prop:sol_static_RCNL}Solution of $\delta=\Phi^{\delta,\gamma\geq0}\left(\delta\right)$ satisfies $S_{j}^{(data)}=s_{j}(\delta)\ \forall j\in\mathcal{J}$.
\end{prop}
When no consumer heterogeneity exists, $\delta_{j}=(1-\rho)\log(S_{j}^{(data)})+\rho\log\left(S_{g}^{(data)}\right)-\log(S_{0}^{(data)})$ holds, as shown in \citet{berry1994estimating}. Regarding the mapping $\Phi^{\delta,\gamma=1}$, $\Phi_{j}^{\delta,\gamma=1}\left(\delta\right)=(1-\rho)\log(S_{j}^{(data)})+\rho\log\left(S_{g}^{(data)}\right)-\log(S_{0}^{(data)})$ holds, and the output of $\Phi^{\delta,\gamma=1}$ is equal to the true $\delta$ for any input. When consumer heterogeneity exists, the mapping does not immediately converge to the solution. Nevertheless, we can expect fast convergence especially in a setting with small consumer heterogeneity, as in the case of static RCL models.

\subsubsection{Mappings of $IV$}

When applying the RCNL model, generally we cannot analytically represent $\delta$ as a function of $V$ alone, unlike the case without nest structure. However, we can alternatively represent $\delta$ as a function of $IV\equiv\left\{ IV_{ig}\right\} _{g\in\mathcal{G}}$, which are nest-specific inclusive values of each consumer type.\footnote{\citet{doi2022simple} also derived the representation of $\delta$ as a function of nest-level choice probabilities by each consumer type.} The following equation shows the analytical formula:

\begin{eqnarray*}
\delta_{j} & =(1-\rho) & \left[\log\left(S_{j}^{(data)}\right)-\log\left(\sum_{i}w_{i}\frac{\exp\left(\frac{\mu_{ij}}{1-\rho}\right)}{\exp\left(\frac{IV_{ig}}{1-\rho}\right)}\frac{\exp\left(IV_{ig}\right)}{1+\sum_{g\in\mathcal{G}}\exp\left(IV_{ig}\right)}\right)\right]
\end{eqnarray*}

Motivated by this analytical formula, define the following mapping $\Phi^{IV,\gamma}:B_{IV}\rightarrow B_{IV}$: 

\begin{eqnarray*}
\Phi^{IV,\gamma}\left(IV\right) & \equiv & \iota_{\delta\rightarrow V}\left(\iota_{IV\rightarrow\delta}^{\gamma}\left(IV\right)\right),
\end{eqnarray*}
where $\iota_{\delta\rightarrow IV}:B_{\delta}\rightarrow B_{IV}$ is a mapping such that:

\begin{eqnarray*}
\iota_{\delta\rightarrow IV,i}\left(\delta\right) & \equiv & \log\left(1+\sum_{j\in\mathcal{J}_{g}}\exp\left(\delta_{j}+\mu_{ij}\right)\right),
\end{eqnarray*}
and $\iota_{IV\rightarrow\delta}:B_{IV}\rightarrow B_{\delta}$ is a mapping such that:

\begin{eqnarray*}
\iota_{IV\rightarrow\delta,j}^{\gamma}\left(IV\right) & = & (1-\rho)\left[\log\left(S_{j}^{(data)}\right)-\log\left(\sum_{i}w_{i}\frac{\exp\left(\frac{\mu_{ij}}{1-\rho}\right)}{\exp\left(\frac{IV_{ig}}{1-\rho}\right)}\frac{\exp\left(IV_{ig}\right)}{1+\sum_{g\in\mathcal{G}}\exp\left(IV_{ig}\right)}\right)\right]\\
 &  & +\gamma\rho\left[\log\left(S_{g}^{(data)}\right)-\log\left(s_{g}(IV)\right)\right]-\gamma\left[\log\left(S_{0}^{(data)}\right)-\log\left(s_{0}(IV)\right)\right].
\end{eqnarray*}

Here, we define $s_{g}(IV)\equiv\sum_{i}w_{i}\frac{\exp\left(IV_{ig}\right)}{1+\sum_{g\in\mathcal{G}}\exp\left(IV_{ig}\right)}$ and $s_{0}(IV)\equiv1-\sum_{g\in\mathcal{G}}s_{g}(IV)$. $B_{IV}$ denotes the space of $IV$.

The following proposition justifies the use of $\Phi^{IV,\gamma\geq0}$ .
\begin{prop}
\label{prop:sol_IV_mapping_static_RCNL}$\delta$ such that $IV=\Phi_{i}^{IV,\gamma\geq0}(IV),\ \delta=\iota_{IV\rightarrow\delta}^{\gamma\geq0}(IV)$ satisfies $S_{j}^{(data)}=s_{j}(\delta)\ \forall j\in\mathcal{J}$.
\end{prop}
Algorithm \ref{alg:static_RCNL} shows the steps to solve for $\delta$.

\begin{algorithm}[H]
\begin{itemize}
\item Algorithm using $\Phi^{\delta,\gamma}$

Set the initial values of $\delta^{(0)}$. Iterate the following $(n=0,1,2,\cdots)$:
\begin{enumerate}
\item Compute $\delta_{j}^{(n+1)}=\Phi_{j}^{\delta,\gamma}(\delta^{(n)})$
\item Exit the iteration if $\left\Vert \delta^{(n+1)}-\delta^{(n)}\right\Vert <\epsilon_{\delta}$
\end{enumerate}
\item Algorithm using $\Phi^{V,\gamma}$

Set the initial values of $IV^{(0)}$. Iterate the following $(n=0,1,2,\cdots)$:
\begin{enumerate}
\item Compute $\delta^{(n)}=\iota_{IV\rightarrow\delta}^{\gamma}\left(IV^{(n)}\right)$
\item Update $IV$ by $IV^{(n+1)}=\iota_{\delta\rightarrow IV}\left(\delta^{(n)}\right)$
\item Exit the iteration if $\left\Vert IV^{(n+1)}-IV^{(n)}\right\Vert <\epsilon_{IV}$
\end{enumerate}
\end{itemize}
\caption{Inner loop algorithms of static RCNL model\label{alg:static_RCNL}}
\end{algorithm}

As in the case of static BLP models without nest structure, $\Phi^{\delta,\gamma}$ and $\Phi^{IV,\gamma}$ have a dualistic relationship. The following proposition, which we can prove easily, shows a formal statement:
\begin{prop}
\label{prop:duality_mapping-RCNL}(Duality of mappings of $\delta$ and $IV$) The following holds for all $\gamma\in\mathbb{R}$:

\begin{eqnarray*}
\Phi^{IV,\gamma} & = & \iota_{\delta\rightarrow IV}\circ\iota_{IV\rightarrow\delta}^{\gamma},\\
\Phi^{\delta,\gamma} & = & \iota_{IV\rightarrow\delta}^{\gamma}\circ\iota_{\delta\rightarrow IV}.
\end{eqnarray*}
\end{prop}

\subsubsection{Numerical experiments}

The parameter settings are the same as the static RCL models, except for the existence of nests with $\rho=0.5$. We assume there are $G=3$ nests, and 25 products in each nest. Here, we denote $\delta$-(0) as the algorithm using the mapping $\Phi^{\delta,\gamma=0}$, and $\delta$-(1) as the algorithm using the mapping $\Phi^{\delta,\gamma=1}$. Similarly, we denote $IV$-(0) as the algorithm using the mapping $\Phi^{IV,\gamma=0}$, and $IV$-(1) as the algorithm using the mapping $\Phi^{IV,\gamma=1}$.

Table \ref{tab:static_BLP_Monte_Carlo_RCNL} shows the results. The results show that the new mappings $\Phi^{\delta,\gamma=1}$ and $\Phi^{IV,\gamma=1}$ perform better than $\Phi^{\delta,\gamma=0}$ and $\Phi^{IV,\gamma=0}$. In addition, combining the acceleration methods leads to faster convergence. As in the static RCL models, the Anderson acceleration is the best. We can also see mostly similar performance of $\Phi^{\delta,\gamma}$ and $\Phi^{IV,\gamma}$, which verifies the dualistic relations between the corresponding mappings.
\begin{center}
\begin{table}[H]
\caption{Results of the Monte Carlo simulation (Static RCNL model; Continuous consumer types)\label{tab:static_BLP_Monte_Carlo_RCNL}}

\scalebox{0.8}{
\begin{centering}
\begin{tabular}{cccccccccccc}
\hline 
 & \multirow{2}{*}{{\scriptsize{}$\rho$}} & \multicolumn{6}{c}{{\scriptsize{}Func. Evals.}} & {\scriptsize{}Mean} & {\scriptsize{}Conv.} & {\scriptsize{}Mean} & {\scriptsize{}$DIST<\epsilon_{tol}$}\tabularnewline
\cline{3-8} \cline{4-8} \cline{5-8} \cline{6-8} \cline{7-8} \cline{8-8} 
 &  & {\scriptsize{}Mean} & {\scriptsize{}Min.} & {\scriptsize{}25th} & {\scriptsize{}Median.} & {\scriptsize{}75th} & {\scriptsize{}Max.} & {\scriptsize{}CPU time (s)} & {\scriptsize{}(\%)} & {\scriptsize{}$\log_{10}(DIST)$} & {\scriptsize{}(\%)}\tabularnewline
\hline 
{\scriptsize{}$\delta$-(0) (BLP)} & {\scriptsize{}0.5} & {\scriptsize{}154.46} & {\scriptsize{}44} & {\scriptsize{}55} & {\scriptsize{}84} & {\scriptsize{}143} & {\scriptsize{}1000} & {\scriptsize{}0.12544} & {\scriptsize{}98} & {\scriptsize{}-13.6} & {\scriptsize{}98}\tabularnewline
{\scriptsize{}$\delta$-(0) (BLP) + Anderson} & {\scriptsize{}0.5} & {\scriptsize{}18.02} & {\scriptsize{}12} & {\scriptsize{}15} & {\scriptsize{}18} & {\scriptsize{}20} & {\scriptsize{}27} & {\scriptsize{}0.01636} & {\scriptsize{}100} & {\scriptsize{}-15.5} & {\scriptsize{}100}\tabularnewline
{\scriptsize{}$\delta$-(0) (BLP) + Spectral} & {\scriptsize{}0.5} & {\scriptsize{}42.7} & {\scriptsize{}19} & {\scriptsize{}23} & {\scriptsize{}28} & {\scriptsize{}36} & {\scriptsize{}471} & {\scriptsize{}0.03768} & {\scriptsize{}100} & {\scriptsize{}-14.6} & {\scriptsize{}100}\tabularnewline
{\scriptsize{}$\delta$-(0) (BLP) + SQUAREM} & {\scriptsize{}0.5} & {\scriptsize{}43.3} & {\scriptsize{}20} & {\scriptsize{}24} & {\scriptsize{}30} & {\scriptsize{}40} & {\scriptsize{}438} & {\scriptsize{}0.03848} & {\scriptsize{}100} & {\scriptsize{}-14.5} & {\scriptsize{}100}\tabularnewline
{\scriptsize{}$\delta$-(1)} & {\scriptsize{}0.5} & {\scriptsize{}29.98} & {\scriptsize{}10} & {\scriptsize{}20} & {\scriptsize{}27} & {\scriptsize{}35} & {\scriptsize{}69} & {\scriptsize{}0.0273} & {\scriptsize{}100} & {\scriptsize{}-15.3} & {\scriptsize{}100}\tabularnewline
{\scriptsize{}$\delta$-(1) (BLP) + Anderson} & {\scriptsize{}0.5} & {\scriptsize{}12.84} & {\scriptsize{}7} & {\scriptsize{}11} & {\scriptsize{}13} & {\scriptsize{}15} & {\scriptsize{}20} & {\scriptsize{}0.01196} & {\scriptsize{}100} & {\scriptsize{}-16.4} & {\scriptsize{}100}\tabularnewline
{\scriptsize{}$\delta$-(1) + Spectral} & {\scriptsize{}0.5} & {\scriptsize{}16.12} & {\scriptsize{}8} & {\scriptsize{}13} & {\scriptsize{}16} & {\scriptsize{}18} & {\scriptsize{}28} & {\scriptsize{}0.01548} & {\scriptsize{}100} & {\scriptsize{}-15.5} & {\scriptsize{}100}\tabularnewline
{\scriptsize{}$\delta$-(1) + SQUAREM} & {\scriptsize{}0.5} & {\scriptsize{}16.36} & {\scriptsize{}9} & {\scriptsize{}13} & {\scriptsize{}15.5} & {\scriptsize{}19} & {\scriptsize{}26} & {\scriptsize{}0.01442} & {\scriptsize{}100} & {\scriptsize{}-15.8} & {\scriptsize{}100}\tabularnewline
\hline 
{\scriptsize{}$IV$-(0)} & {\scriptsize{}0.5} & {\scriptsize{}161.26} & {\scriptsize{}44} & {\scriptsize{}56} & {\scriptsize{}88.5} & {\scriptsize{}172} & {\scriptsize{}1000} & {\scriptsize{}0.20378} & {\scriptsize{}98} & {\scriptsize{}-13.7} & {\scriptsize{}98}\tabularnewline
{\scriptsize{}$IV$-(0) + Anderson} & {\scriptsize{}0.5} & {\scriptsize{}18.12} & {\scriptsize{}11} & {\scriptsize{}15} & {\scriptsize{}18} & {\scriptsize{}21} & {\scriptsize{}29} & {\scriptsize{}0.02672} & {\scriptsize{}100} & {\scriptsize{}-15.7} & {\scriptsize{}100}\tabularnewline
{\scriptsize{}$IV$-(0) + Spectral} & {\scriptsize{}0.5} & {\scriptsize{}41.9} & {\scriptsize{}20} & {\scriptsize{}23} & {\scriptsize{}30} & {\scriptsize{}41} & {\scriptsize{}405} & {\scriptsize{}0.05702} & {\scriptsize{}100} & {\scriptsize{}-15.2} & {\scriptsize{}100}\tabularnewline
{\scriptsize{}$IV$-(0) + SQUAREM} & {\scriptsize{}0.5} & {\scriptsize{}39.94} & {\scriptsize{}21} & {\scriptsize{}25} & {\scriptsize{}29.5} & {\scriptsize{}40} & {\scriptsize{}276} & {\scriptsize{}0.05184} & {\scriptsize{}100} & {\scriptsize{}-15.1} & {\scriptsize{}100}\tabularnewline
{\scriptsize{}$IV$-(1)} & {\scriptsize{}0.5} & {\scriptsize{}31.2} & {\scriptsize{}11} & {\scriptsize{}20} & {\scriptsize{}29} & {\scriptsize{}37} & {\scriptsize{}74} & {\scriptsize{}0.04304} & {\scriptsize{}100} & {\scriptsize{}-15.3} & {\scriptsize{}100}\tabularnewline
{\scriptsize{}$IV$-(1) + Anderson} & {\scriptsize{}0.5} & {\scriptsize{}13.92} & {\scriptsize{}8} & {\scriptsize{}11} & {\scriptsize{}14} & {\scriptsize{}16} & {\scriptsize{}22} & {\scriptsize{}0.02112} & {\scriptsize{}100} & {\scriptsize{}-16} & {\scriptsize{}100}\tabularnewline
{\scriptsize{}$IV$-(1) + Spectral} & {\scriptsize{}0.5} & {\scriptsize{}17.4} & {\scriptsize{}9} & {\scriptsize{}14} & {\scriptsize{}17} & {\scriptsize{}21} & {\scriptsize{}28} & {\scriptsize{}0.02496} & {\scriptsize{}100} & {\scriptsize{}-15.6} & {\scriptsize{}100}\tabularnewline
{\scriptsize{}$IV$-(1) + SQUAREM} & {\scriptsize{}0.5} & {\scriptsize{}18.86} & {\scriptsize{}12} & {\scriptsize{}15} & {\scriptsize{}18} & {\scriptsize{}23} & {\scriptsize{}28} & {\scriptsize{}0.02508} & {\scriptsize{}100} & {\scriptsize{}-15.7} & {\scriptsize{}100}\tabularnewline
\hline 
\end{tabular}
\par\end{centering}
}

{\footnotesize{}\input{notes/note_static_BLP_results.tex}}{\footnotesize\par}

{\footnotesize{}The mean Outside option share is 0.660.}{\footnotesize\par}
\end{table}
\par\end{center}

\section{Proof\label{sec:Proof}}

\subsection{Justification for the algorithms}

\subsubsection{Proof of Propositions \ref{prop:sol_delta_mapping_static_BLP} and \ref{prop:sol_static_RCNL} ($\Phi^{\delta,\gamma}$ in the RCL / RCNL models)}

We show Proposition \ref{prop:sol_static_RCNL} , because RCNL models include BLP models without nest structure.
\begin{proof}
Solution of $\delta=\Phi_{j}^{\delta,\gamma}\left(\delta\right)$ satisfies:

\begin{eqnarray*}
0 & = & (1-\rho_{g})\left[\log(S_{j}^{(data)})-\log(s_{j}(\delta))\right]+\gamma\rho_{g}\left[\log(S_{g}^{(data)})-\log(s_{g}(\delta))\right]-\gamma\left[\log(S_{0}^{(data)})-\log(s_{0}(\delta))\right]
\end{eqnarray*}

Then, we have:

\begin{eqnarray}
\frac{S_{j}^{(data)}\cdot\left(S_{g}^{(data)}\right)^{\frac{\rho_{g}}{1-\rho_{g}}\gamma}}{\left(S_{0}^{(data)}\right)^{\frac{1}{1-\rho_{g}}\gamma}} & = & \frac{s_{j}(\delta)\cdot\left(s_{g}(\delta)\right)^{\frac{\rho_{g}}{1-\rho_{g}}\gamma}}{\left(s_{0}(\delta)\right)^{\frac{1}{1-\rho_{g}}\gamma}}\ j\in\mathcal{J}_{g},g\in\mathcal{G}.\label{eq:S_j_eq}
\end{eqnarray}

By summing up both sides for all $j\in\mathcal{J}_{g}$, we have:

\begin{eqnarray*}
\frac{S_{g}^{(data)}\cdot\left(S_{g}^{(data)}\right)^{\frac{\rho_{g}}{1-\rho_{g}}\gamma}}{\left(S_{0}^{(data)}\right)^{\frac{1}{1-\rho_{g}}\gamma}} & = & \frac{s_{g}(\delta)\cdot\left(s_{g}(\delta)\right)^{\frac{\rho_{g}}{1-\rho_{g}}\gamma}}{\left(s_{0}(\delta)\right)^{\frac{1}{1-\rho_{g}}\gamma}}\ g\in\mathcal{G},
\end{eqnarray*}

namely, 

\begin{eqnarray}
\frac{S_{g}^{(data)}}{\left(S_{0}^{(data)}\right)^{\frac{\gamma}{1-\rho_{g}+\rho_{g}\gamma}}} & = & \frac{s_{g}(\delta)}{\left(s_{0}(\delta)\right)^{\frac{\gamma}{1-\rho_{g}+\rho_{g}\gamma}}},g\in\mathcal{G}.\label{eq:S_g_eq}
\end{eqnarray}

By summing up both sides for all $g\in\mathcal{G}$, we have:

\begin{eqnarray}
\frac{1-S_{0}^{(data)}}{\left(S_{0}^{(data)}\right)^{\frac{\gamma}{1-\rho_{g}+\rho_{g}\gamma}}} & = & \frac{1-s_{0}(\delta)}{\left(s_{0}(\delta)\right)^{\frac{\gamma}{1-\rho_{g}+\rho_{g}\gamma}}}.\label{eq:S_0_eq}
\end{eqnarray}

We can interpret the equation above as an equation with unknown variable $s_{0}(\delta)$. The equation has an obvious solution $s_{0}(\delta)=S_{0}^{(data)}$. When $\gamma\geq0$ and $\rho_{g}\in[0,1)$ holds, $\frac{\gamma}{1-\rho_{g}+\rho_{g}\gamma}>0$ holds, and $\frac{1-s_{0}(\delta)}{\left(s_{0}(\delta)\right)^{\frac{\gamma}{1-\rho_{g}+\rho_{g}\gamma}}}$ is a decreasing function with respect to $s_{0}(\delta)$. Then, equation (\ref{eq:S_0_eq}) has an unique solution $s_{0}(\delta)=S_{0}^{(data)}$, and equation (\ref{eq:S_0_eq}) implies $s_{0}(\delta)=S_{0}^{(data)}$. By equation (\ref{eq:S_g_eq}), $s_{g}(\delta)=S_{g}^{(data)}$ holds, and by (\ref{eq:S_j_eq}), $s_{j}(\delta)=S_{j}^{(data)}$ holds.
\end{proof}

\subsubsection{Proof of Proposition \ref{prop:sol_V_mapping_static_BLP} ($\Phi^{V,\gamma}$ in the RCL model)}
\begin{proof}
$\delta_{j}=\iota_{V\rightarrow\delta,j}^{\gamma}(V)=\log\left(S_{j}^{(data)}\right)-\log\left(\sum_{i}w_{i}\exp\left(\mu_{ij}-V_{i}\right)\right)-\gamma\left(\log(S_{0}^{(data)})-\log\left(\sum_{i}w_{i}\exp(-V_{i})\right)\right)$ implies:

\begin{eqnarray*}
S_{j}^{(data)} & = & \sum_{i}w_{i}\frac{\exp\left(\delta_{j}+\mu_{ij}\right)}{\exp(V_{i})}\left(\frac{S_{0}^{(data)}}{\sum_{i}w_{i}\exp(-V_{i})}\right)^{\gamma}
\end{eqnarray*}

Because $V=\Phi^{V,\gamma}(V)$ implies $V_{i}=\log\left(1+\sum_{k\in\mathcal{J}}\exp\left(\delta_{k}+\mu_{ik}\right)\right)$, $\delta$ satisfies:

\begin{eqnarray*}
S_{j}^{(data)} & = & \sum_{i}w_{i}\frac{\exp\left(\delta_{j}+\mu_{ij}\right)}{1+\sum_{k\in\mathcal{J}}\exp\left(\delta_{k}+\mu_{ik}\right)}\left(\frac{S_{0}^{(data)}}{\sum_{i}w_{i}\frac{1}{1+\sum_{k\in\mathcal{J}}\exp\left(\delta_{k}+\mu_{ik}\right)}}\right)^{\gamma}\\
 & = & s_{j}(\delta)\cdot\left(\frac{S_{0}^{(data)}}{s_{0}(\delta)}\right)^{\gamma}.
\end{eqnarray*}

Then, $\delta$ satisfies $\delta_{j}=\Phi_{j}^{\delta,\gamma}\left(\delta\right)=\delta_{j}+\log\left(\frac{S_{j}^{(data)}}{s_{j}(\delta)}\slash\left(\frac{S_{0}^{(data)}}{s_{0}(\delta)}\right)^{\gamma}\right)\ \forall j\in\mathcal{J}$, and $S_{j}^{(data)}=s_{j}(\delta)$ holds $\forall\gamma\geq0$ by Proposition \ref{prop:sol_delta_mapping_static_BLP}.
\end{proof}

\subsubsection{Proof of Proposition \ref{prop:sol_IV_mapping_static_RCNL} ($\Phi^{IV,\gamma}$ in the RCNL model)}
\begin{proof}
First, $\delta_{j}=\iota_{IV\rightarrow\delta,j}^{\gamma}(IV)=(1-\rho)\log\left(\frac{S_{j}^{(data)}}{\sum_{i}w_{i}\frac{\exp\left(\frac{\mu_{ij}}{1-\rho}\right)}{\exp\left(\frac{IV_{ig}}{1-\rho}\right)}\frac{\exp\left(IV_{ig}\right)}{1+\sum_{g\in\mathcal{G}}\exp\left(IV_{ig}\right)}}\right)+\gamma\log\left(\left(\frac{S_{g}^{(data)}}{s_{g}(IV)}\right)^{\rho}\slash\frac{S_{0}^{(data)}}{s_{0}(IV)}\right)$ holds. Then, 

\begin{eqnarray*}
0 & = & (1-\rho)\log\left(\frac{S_{j}^{(data)}}{\sum_{i}w_{i}\frac{\exp\left(\frac{\delta_{j}+\mu_{ij}}{1-\rho}\right)}{\exp\left(\frac{IV_{ig}}{1-\rho}\right)}\frac{\exp\left(IV_{ig}\right)}{1+\sum_{g\in\mathcal{G}}\exp\left(IV_{ig}\right)}}\right)+\gamma\log\left(\left(\frac{S_{g}^{(data)}}{s_{g}(IV)}\right)^{\rho}\slash\frac{S_{0}^{(data)}}{s_{0}(IV)}\right)\\
 & = & (1-\rho)\log\left(\frac{S_{j}^{(data)}}{\sum_{i}w_{i}\frac{\exp\left(\frac{\delta_{j}+\mu_{ij}}{1-\rho}\right)}{\exp\left(\frac{IV_{ig}}{1-\rho}\right)}\frac{\exp\left(IV_{ig}\right)}{1+\sum_{g\in\mathcal{G}}\exp\left(IV_{ig}\right)}}\right)+\gamma\log\left(\left(\frac{S_{g}^{(data)}}{s_{g}\left(\delta\right)}\right)^{\rho}\slash\frac{S_{0}^{(data)}}{s_{0}\left(\delta\right)}\right)\\
 & = & (1-\rho)\log\left(\frac{S_{j}^{(data)}}{s_{j}(\delta)}\right)+\gamma\log\left(\left(\frac{S_{g}^{(data)}}{s_{g}\left(\delta\right)}\right)^{\rho}\slash\frac{S_{0}^{(data)}}{s_{0}\left(\delta\right)}\right).
\end{eqnarray*}

Hence, $\delta=\delta+(1-\rho)\log\left(\frac{S_{j}^{(data)}}{s_{j}(\delta)}\right)+\gamma\log\left(\left(\frac{S_{g}^{(data)}}{s_{g}\left(\delta\right)}\right)^{\rho}\slash\frac{S_{0}^{(data)}}{s_{0}\left(\delta\right)}\right)$ holds, and $S_{j}^{(data)}=s_{j}(\delta)$ holds $\forall\gamma\geq0$ by Proposition \ref{prop:sol_static_RCNL}.
\end{proof}

\subsection{Convergence properties of the mappings}

We first show the following lemma.

\begin{lem}
\label{lem:fx_diff_upper}Let $f:B_{x}\subset\mathbb{R}^{K}\rightarrow f(B_{x})\subset\mathbb{R}^{K}$ be a continuously differentiable function. Then, $f$ satisfies the following inequalities for any $x,x^{*}\in\mathbb{R}^{K}$:

(a). 
\begin{eqnarray*}
\left\Vert f(x)-f(x^{*})\right\Vert _{\infty} & \leq & \left[\sup_{I\subset\{1,\cdots,K\},x\in B_{x}}\left\Vert \sum_{i\in I}\frac{\partial f(x)}{\partial x_{i}}\right\Vert _{\infty}\right]\cdot\left[\left\Vert \max\{x,x^{*}\}-x\right\Vert _{\infty}+\left\Vert \max\{x,x^{*}\}-x^{*}\right\Vert _{\infty}\right]\\
 & \leq & \left[\sup_{I\subset\{1,\cdots,K\},x\in B_{x}}\left\Vert \sum_{i\in I}\frac{\partial f(x)}{\partial x_{i}}\right\Vert _{\infty}\right]\cdot2\left\Vert x-x^{*}\right\Vert _{\infty}
\end{eqnarray*}

(b). 

\begin{eqnarray*}
\left\Vert f(x)-f(x^{*})\right\Vert _{\infty} & \leq & \left[\sup_{x\in B_{x}}\left\Vert \nabla_{x}f(x)\right\Vert \right]\cdot\left\Vert x-x^{*}\right\Vert _{\infty}\\
 & = & \left[\sup_{x\in B_{x}}\left\Vert \sum_{i\in\{1,\cdots,K\}}\left|\frac{\partial f(x)}{\partial x_{i}}\right|\right\Vert _{\infty}\right]\cdot\left\Vert x-x^{*}\right\Vert _{\infty}
\end{eqnarray*}
\end{lem}
(b) is a well-known inequality, and $\sup_{x}\left\Vert \nabla_{x}f(x)\right\Vert $ is known as the Lipschitz constant.
\begin{proof}
We can derive (a) by: 

\begin{eqnarray*}
\left\Vert f(x)-f(x^{*})\right\Vert _{\infty} & = & \left\Vert -\left(f(\max\{x,x^{*}\})-f(x)\right)+\left(f(\max\{x,x^{*}\})-f(x^{*})\right)\right\Vert _{\infty}\\
 & = & \left\Vert -\int_{0}^{\left\Vert \max\{x,x^{*}\}-x\right\Vert _{\infty}}\left[\sum_{i=1}^{K}\frac{\partial f(x^{*}+\lambda\cdot1)}{\partial x_{i}}1[\lambda\leq\max\{x_{i},x_{i}^{*}\}-x_{i}]\right]d\lambda\right.\\
 &  & \left.+\int_{0}^{\left\Vert \max\{x,x^{*}\}-x^{*}\right\Vert _{\infty}}\left[\sum_{i=1}^{K}\frac{\partial f(x^{*}+\lambda\cdot1)}{\partial x_{i}}1[\lambda\leq\max\{x_{i},x_{i}^{*}\}-x_{i}^{*}]\right]d\lambda\right\Vert _{\infty}\\
 & \leq & \left(\sup_{I\subset\{1,\cdots,K\},x\in B_{x}}\left\Vert \sum_{i\in I}\frac{\partial f(x)}{\partial x_{i}}\right\Vert _{\infty}\right)\cdot\left(\left\Vert \max\{x,x^{*}\}-x\right\Vert _{\infty}+\left\Vert \max\{x,x^{*}\}-x^{*}\right\Vert _{\infty}\right)
\end{eqnarray*}
\end{proof}

\subsubsection*{Proof of Proposition \ref{prop:property_mapping_delta} }

We can prove Proposition \ref{prop:property_mapping_delta} using the following lemma.
\begin{lem}
\label{lem:mapping_delta_property}

(a). $\frac{\partial\Phi_{j}^{\gamma}(\delta)}{\partial\delta_{k}}=\sum_{i\in\mathcal{I}}w_{i}s_{ik}\left(\frac{s_{ij}}{s_{j}}-\gamma\frac{s_{i0}}{s_{0}}\right)\ \forall k\in\mathcal{J}$, where $s_{j}=\sum_{i}w_{i}\frac{\exp\left(\delta_{j}+\mu_{ij}\right)}{1+\sum_{k\in\mathcal{J}}\exp\left(\delta_{k}+\mu_{ik}\right)}$ and $s_{0}=\sum_{i\in\mathcal{I}}w_{i}\frac{1}{1+\sum_{k\in\mathcal{J}}\exp\left(\delta_{k}+\mu_{ik}\right)}$

(b). $\sup_{J^{*}\subset\mathcal{J},\delta\in B_{\delta}}\left|\sum_{k\in J^{*}\subset\mathcal{J}}\frac{\partial\Phi_{j}(\delta)}{\partial\delta_{k}}\right|\leq\widetilde{c_{\gamma}}$
\end{lem}
\begin{proof}
(a). 

\begin{eqnarray*}
\frac{\partial\Phi_{j}^{\gamma}(\delta)}{\partial\delta_{j}} & = & 1-\frac{1}{s_{j}}\frac{\partial s_{j}}{\partial\delta_{j}}+\gamma\frac{1}{s_{0}}\frac{\partial s_{0}}{\partial\delta_{j}}\\
 & = & 1-\frac{s_{j}-\sum_{i}w_{i}s_{ij}^{2}}{s_{j}}-\gamma\frac{\sum_{i}w_{i}s_{ij}s_{i0}}{s_{0}}\\
 & = & \sum_{i}w_{i}s_{ij}\left(\frac{s_{ij}}{s_{j}}-\gamma\frac{s_{i0}}{s_{0}}\right)\\
\frac{\partial\Phi_{j}^{\gamma}(\delta)}{\partial\delta_{k}} & = & -\frac{1}{s_{j}}\frac{\partial s_{j}}{\partial\delta_{k}}+\gamma\frac{1}{s_{0}}\frac{\partial s_{0}}{\partial\delta_{k}}\ \ \ (k\neq j)\\
 & = & \frac{\sum_{i}w_{i}s_{ij}s_{ik}}{s_{j}}-\gamma\frac{\sum_{i}w_{i}s_{ik}s_{i0}}{s_{0}}\\
 & = & \sum_{i}w_{i}s_{ik}\left(\frac{s_{ij}}{s_{j}}-\gamma\frac{s_{i0}}{s_{0}}\right)
\end{eqnarray*}

(b).

\begin{eqnarray*}
\sup_{J^{*}\subset\mathcal{J},\delta}\left|\sum_{k\in J^{*}\subset\mathcal{J}}\frac{\partial\Phi_{j}(\delta)}{\partial\delta_{k}}\right| & = & \sup_{J^{*}\subset\mathcal{J},\delta}\left|\sum_{k\in J^{*}\subset\mathcal{J}}\sum_{i}w_{i}s_{ik}\left(\frac{s_{ij}}{s_{j}}-\gamma\frac{s_{i0}}{s_{0}}\right)\right|\\
 & = & \sup_{J^{*}\subset\mathcal{J},\delta}\left|\sum_{i}w_{i}s_{iJ^{*}}\left(\frac{s_{ij}}{s_{j}}-\gamma\frac{s_{i0}}{s_{0}}\right)\right|
\end{eqnarray*}

Here, 

\begin{eqnarray*}
\sum_{i}w_{i}s_{iJ^{*}}\left(\frac{s_{ij}}{s_{j}}-\gamma\frac{s_{i0}}{s_{0}}\right) & \leq & \sum_{i}w_{i}\left(\max_{i}s_{iJ^{*}}\right)\frac{s_{ij}}{s_{j}}-\gamma\sum_{i}w_{i}\left(\min_{i}s_{iJ^{*}}\right)\frac{s_{i0}}{s_{0}}\\
 & = & \max_{i}s_{iJ^{*}}-\gamma\min_{i}s_{iJ^{*}}\\
 & \leq & \widetilde{c_{\gamma}}
\end{eqnarray*}

\begin{eqnarray*}
\sum_{i}w_{i}s_{iJ^{*}}\left(\frac{s_{ij}}{s_{j}}-\gamma\frac{s_{i0}}{s_{0}}\right) & \geq & \sum_{i}w_{i}\left(\min_{i}s_{iJ^{*}}\right)\frac{s_{ij}}{s_{j}}-\gamma\sum_{i}w_{i}\left(\max_{i}s_{iJ^{*}}\right)\frac{s_{i0}}{s_{0}}\\
 & = & \min_{i}s_{iJ^{*}}-\gamma\max_{i}s_{iJ^{*}}\\
 & \geq & -\widetilde{c_{\gamma}}
\end{eqnarray*}

Hence,

\begin{eqnarray*}
\sup_{J^{*}\subset\mathcal{J},\delta}\left|\sum_{k\in J^{*}\subset\mathcal{J}}\frac{\partial\Phi_{j}(\delta)}{\partial\delta_{k}}\right| & \leq & \widetilde{c_{\gamma}}.
\end{eqnarray*}
\end{proof}

\subsubsection*{Proof of Proposition \ref{prop:property_mapping_V}}

We can show Proposition \ref{prop:property_mapping_V} based on the following lemmas.
\begin{lem}
\label{lem:mapping_V_property}

(a). 
\begin{eqnarray*}
\frac{\partial\Phi_{i_{2}}^{V,\gamma}(V)}{\partial V_{i_{1}}} & = & \frac{\sum_{j\in\mathcal{J}}S_{j}^{(data)}\frac{\exp(\mu_{i_{2}j})}{\sum_{i}w_{i}\exp(\mu_{ij})\exp(-V_{i})}\left(\frac{\sum_{i\in\mathcal{I}}w_{i}\exp(-V_{i})}{S_{0}^{(data)}}\right)^{\gamma}\left[\frac{w_{i_{1}}\exp(\mu_{i_{1}j})\exp(-V_{i_{1}})}{\sum_{i}w_{i}\exp(\mu_{ij})\exp(-V_{i})}-\gamma\frac{w_{i_{1}}\exp(-V_{i_{1}})}{\sum_{i}w_{i}\exp(-V_{i})}\right]}{1+\sum_{j\in\mathcal{J}}S_{j}^{(data)}\frac{\exp(\mu_{i_{2}j})}{\sum_{i}w_{i}\exp(\mu_{ij})\exp(-V_{i})}}
\end{eqnarray*}

(b). $\sup_{I^{*}\subset\mathcal{I},V\in B_{V}}\left|\sum_{i_{1}\in I^{*}\subset\mathcal{I}}\frac{\partial\Phi_{i_{2}}^{V,\gamma}(V)}{\partial V_{i_{1}}}\right|\leq\sup_{V\in B_{V}}\left((1-prob_{0|i_{2}}(V))\cdot\left(\frac{\sum_{i\in\mathcal{I}}w_{i}\exp(-V_{i})}{S_{0}^{(data)}}\right)^{\gamma}\right)\cdot\sup_{I^{*}\subset\mathcal{I},j\in\mathcal{J},V\in B_{V}}\left|prob_{I^{*}|j}(V)-\gamma prob_{I^{*}|0}(V)\right|$, 

(c). $\sup_{V\in B_{V}}\sum_{i_{1}\in\mathcal{I}}\left|\frac{\partial\Phi_{i_{2}}^{V,\gamma}(V)}{\partial V_{i_{1}}}\right|\leq\begin{cases}
2\sup_{V\in B_{V}}\left((1-prob_{0|i_{2}}(V))\cdot\left(\frac{\sum_{i\in\mathcal{I}}w_{i}\exp(-V_{i})}{S_{0}^{(data)}}\right)^{\gamma}\right) & \forall\gamma\in[0,1],\\
\sup_{V\in B_{V}}\left((1-prob_{0|i_{2}}(V))\right) & \text{if}\ \gamma=0.
\end{cases}$
\end{lem}
\begin{proof}
(a)

\begin{eqnarray*}
\frac{\partial\Phi_{i_{2}}^{V,\gamma}(V)}{\partial V_{i_{1}}} & = & \frac{1}{1+\sum_{j\in\mathcal{J}}S_{j}^{(data)}\frac{\exp(\mu_{i_{2}j})}{\sum_{i}w_{i}\exp(\mu_{ij})\exp(-V_{i})}\cdot\left(\frac{\sum_{i}w_{i}\exp(-V_{i})}{S_{0}^{(data)}}\right)^{\gamma}}\cdot\\
 &  & \left(\sum_{j\in\mathcal{J}}S_{j}^{(data)}\exp(\mu_{i_{2}j})\frac{w_{i_{1}}\exp(\mu_{i_{1}j})\exp(-V_{i_{1}})}{\left(\sum_{i}w_{i}\exp(\mu_{ij})\exp(-V_{i})\right)^{2}}\cdot\left(\frac{\sum_{i}w_{i}\exp(-V_{i})}{S_{0}^{(data)}}\right)^{\gamma}+\right.\\
 &  & \left.\sum_{j\in\mathcal{J}}S_{j}^{(data)}\frac{\exp(\mu_{i_{2}j})}{\sum_{i}w_{i}\exp(\mu_{ij})\exp(-V_{i})}\cdot(-1)\cdot\left(\frac{\sum_{i}w_{i}\exp(-V_{i})}{S_{0}^{(data)}}\right)^{\gamma}\right)\\
 & = & \frac{\sum_{j\in\mathcal{J}}S_{j}^{(data)}\frac{\exp(\mu_{i_{2}j})}{\sum_{i}w_{i}\exp(\mu_{ij})\exp(-V_{i})}\left(\frac{\sum_{i}w_{i}\exp(-V_{i})}{S_{0}^{(data)}}\right)^{\gamma}\left[\frac{w_{i_{1}}\exp(\mu_{i_{1}j})\exp(-V_{i_{1}})}{\sum_{i}w_{i}\exp(\mu_{ij})\exp(-V_{i})}-\gamma\frac{w_{i_{1}}\exp(-V_{i_{1}})}{\sum_{i}w_{i}\exp(-V_{i})}\right]}{1+\sum_{j\in\mathcal{J}}S_{j}^{(data)}\frac{\exp(\mu_{i_{2}j})}{\sum_{i}w_{i}\exp(\mu_{ij})\exp(-V_{i})}}.
\end{eqnarray*}

(b).

\begin{eqnarray*}
 &  & \left|\sum_{i_{1}\in I^{*}\subset\mathcal{I}}\frac{\partial\Phi_{i_{2}}^{V,\gamma}(V)}{\partial V_{i_{1}}}\right|\\
 & = & \left|\frac{\sum_{j\in\mathcal{J}}S_{j}^{(data)}\frac{\exp(\mu_{i_{2}j})}{\sum_{i}w_{i}\exp(\mu_{ij})\exp(-V_{i})}\left(\frac{\sum_{i}w_{i}\exp(-V_{i})}{S_{0}^{(data)}}\right)^{\gamma}\cdot\sum_{i_{1}\in I^{*}\subset\mathcal{I}}\left[prob_{i_{1}|j}-prob_{i_{1}|0}\right]}{1+\sum_{j\in\mathcal{J}}S_{j}^{(data)}\frac{\exp(\mu_{i_{2}j})}{\sum_{i}w_{i}\exp(\mu_{ij})\exp(-V_{i})}}\right|\\
 & = & \left|\frac{\sum_{j\in\mathcal{J}}S_{j}^{(data)}\frac{\exp(\mu_{i_{2}j})}{\sum_{i}w_{i}\exp(\mu_{ij})\exp(-V_{i})}\cdot\left(\frac{\sum_{i}w_{i}\exp(-V_{i})}{S_{0}^{(data)}}\right)^{\gamma}\left[prob_{I^{*}|j}-\gamma prob_{I^{*}|0}\right]}{1+\sum_{j\in\mathcal{J}}S_{j}^{(data)}\frac{\exp(\mu_{i_{2}j})}{\sum_{i}w_{i}\exp(\mu_{ij})\exp(-V_{i})}}\right|\\
 & \leq & \frac{\sum_{j\in\mathcal{J}}S_{j}^{(data)}\frac{\exp(\mu_{i_{2}j})}{\sum_{i}w_{i}\exp(\mu_{ij})\exp(-V_{i})}}{1+\sum_{j\in\mathcal{J}}S_{j}^{(data)}\frac{\exp(\mu_{i_{2}j})}{\sum_{i}w_{i}\exp(\mu_{ij})\exp(-V_{i})}}\left(\frac{\sum_{i}w_{i}\exp(-V_{i})}{S_{0}^{(data)}}\right)^{\gamma}\cdot\sup_{I^{*}\subset\mathcal{I},j\in\mathcal{J}}\left|prob_{I^{*}|j}-\gamma prob_{I^{*}|0}\right|\\
 & = & \left(1-prob_{0|i_{2}}\right)\cdot\left(\frac{\sum_{i}w_{i}\exp(-V_{i})}{S_{0}^{(data)}}\right)^{\gamma}\cdot\sup_{I^{*}\subset\mathcal{I},j\in\mathcal{J},V}\left|prob_{I^{*}|j}-\gamma prob_{I^{*}|0}\right|.
\end{eqnarray*}

(c). 

Because $\sup_{I^{*}\subset\mathcal{I},j\in\mathcal{J},V}\sum_{i_{1}\in\mathcal{I}}\left|prob_{i_{1}|j}-prob_{i_{1}|0}\right|\leq2$ for $\gamma\in[0,1]$, $\sup_{V}\sum_{i_{1}\in\mathcal{I}}\left|\frac{\partial\Phi_{i_{2}}^{V,\gamma}(V)}{\partial V_{i_{1}}}\right|\leq2\sup_{V}\left((1-prob_{0|i_{2}}(V))\cdot\left(\frac{\sum_{i}w_{i}\exp(-V_{i})}{S_{0}^{(data)}}\right)^{\gamma}\right)$ holds for all $\gamma\in[0,1]$.

Under $\gamma=0$, 

\begin{eqnarray*}
\sum_{i_{1}\in\mathcal{I}}\left|\frac{\partial\Phi_{i_{2}}^{V,\gamma=0}(V)}{\partial V_{i_{1}}}\right| & = & \frac{\sum_{j\in\mathcal{J}}S_{j}^{(data)}\frac{\exp(\mu_{i_{2}j})}{\sum_{i}w_{i}\exp(\mu_{ij})\exp(-V_{i})}\left[\sum_{i_{1}\in\mathcal{I}}\frac{w_{i_{1}}\exp(\mu_{i_{1}j})\exp(-V_{i_{1}})}{\sum_{i}w_{i}\exp(\mu_{ij})\exp(-V_{i})}\right]}{1+\sum_{j\in\mathcal{J}}S_{j}^{(data)}\frac{\exp(\mu_{i_{2}j})}{\sum_{i}w_{i}\exp(\mu_{ij})\exp(-V_{i})}}\\
 & = & \frac{\sum_{j\in\mathcal{J}}S_{j}^{(data)}\frac{\exp(\mu_{i_{2}j})}{\sum_{i}w_{i}\exp(\mu_{ij})\exp(-V_{i})}}{1+\sum_{j\in\mathcal{J}}S_{j}^{(data)}\frac{\exp(\mu_{i_{2}j})}{\sum_{i}w_{i}\exp(\mu_{ij})\exp(-V_{i})}}
\end{eqnarray*}
holds.
\end{proof}
\pagebreak{}
\begin{center}
\textbf{\LARGE{}Supplemental Appendix to ``Fast and simple inner-loop algorithms of static / dynamic BLP estimations''}{\LARGE\par}
\par\end{center}

\begin{center}
{\Large{}Takeshi Fukasawa}\footnote{Waseda Institute for Advanced Study, Waseda University; fukasawa3431@gmail.com\\
Replication code of the numerical experiments in this article is available at \url{https://github.com/takeshi-fukasawa/BLP_algorithm}.}
\par\end{center}

\appendix

\section{Additional results and discussions\label{sec:Additional-results-1}}

\subsection{Case with large consumer heterogeneity\label{subsec:large_hetero}}

As discussed in Appendix A.1 of the main article, the size of the consumer heterogeneity is critical for the performance of the new mappings $\Phi^{\delta,\gamma=1},\Phi^{V,\gamma=1}$. This section shows results under large consumer heterogeneity. Unlike the standard setting with relatively small consumer heterogeneity experimented in \citet{dube2012improving}, \citet{lee2015computationally} and others, applying the mappings $\Phi^{\delta,\gamma=1},\Phi^{V,\gamma=1}$ does not necessarily lead to convergence. Nevertheless, the spectral algorithm works well even under this setting.

\subsubsection*{Settings}

Consumer $i$'s utility when choosing product $j$ is $U_{ij}=\delta_{j}+\mu_{ij}+\epsilon_{ij}$, and utility when not buying anything is $U_{i0}=\epsilon_{i0}$. Suppose $\epsilon$ follows Gumbel distribution. Let $|\mathcal{J}|=2$ (2 products), $|\mathcal{I}|=2$ (2 consumer types), $\delta_{j=1}=0$, $\delta_{j=2}=-1,$ and $\mu_{i=1,j=1}=\mu_{i=2,j=2}=10,$ $\mu_{i=1,j=2}=\mu_{i=2,j=1}=0$, $w_{i=1}=0.1,w_{i=2}=0.9$.

Table \ref{tab:Choice-probabilities-large-hetero} shows choice probabilities for each consumer type under the parameter setting. As the table shows, there is large consumer heterogeneity: Type 1 consumers strongly prefer product 1, and type 2 consumers strongly prefer product 2.

\begin{table}[H]
\caption{Choice probabilities by consumer type\label{tab:Choice-probabilities-large-hetero}}

\centering{}{\footnotesize{}}%
\begin{tabular}{ccc}
\hline 
 & {\footnotesize{}$i=1$} & {\footnotesize{}$i=2$}\tabularnewline
\hline 
\hline 
{\footnotesize{}$j=1$} & {\footnotesize{}0.9999} & {\footnotesize{}0.0001}\tabularnewline
{\footnotesize{}$j=2$} & {\footnotesize{}0.0000} & {\footnotesize{}0.9998}\tabularnewline
{\footnotesize{}$j=0$} & {\footnotesize{}0.0000} & {\footnotesize{}0.0001}\tabularnewline
\hline 
\end{tabular}{\footnotesize\par}
\end{table}

Unlike the setting in Section 6, I use the true $\mu_{ij}(i=1,2;j=1,2)$ to solve for $\delta$, to make the setting replicable.

\subsubsection*{Results}

Table \ref{tab:Results-static-BLP-large-hetero} shows the results. When we apply the algorithm $\delta$-(1) (mapping $\Phi^{\delta,\gamma=1}$), the algorithm does not converge. Regarding $\delta$-(1), Figure \ref{fig:Example-of-non-contraction} shows the trend of the norm $\left\Vert \delta^{(n+1)}-\delta^{(n)}\right\Vert _{\infty}$. If the mapping $\Phi^{\delta,\gamma=1}$ is a contraction, the value of $\left\Vert \delta^{(n+1)}-\delta^{(n)}\right\Vert _{\infty}$ should decrease as the number of iterations increases, and $\left\Vert \delta^{(n+2)}-\delta^{(n+1)}\right\Vert _{\infty}-\left\Vert \delta^{(n+1)}-\delta^{(n)}\right\Vert _{\infty}<0$ should hold. Nevertheless, they do not hold in the figure. This implies that $\delta$-(1) is not a contraction in the current setting.

Nevertheless, even when using $\delta$-(1), combining the spectral algorithm works well, as shown in the number of function evaluations. 

\begin{table}[H]
\caption{Results of the Monte Carlo simulation (Static BLP model; Too large consumer heterogeneity)\label{tab:Results-static-BLP-large-hetero}}

\scalebox{0.8}{
\begin{centering}
\begin{tabular}{cccccccccccc}
\hline 
 & \multirow{2}{*}{{\scriptsize{}$J$}} & \multicolumn{6}{c}{{\scriptsize{}Func. Eval.}} & {\scriptsize{}Mean} & {\scriptsize{}Conv.} & {\footnotesize{}Mean} & {\scriptsize{}$DIST<\epsilon_{tol}$}\tabularnewline
\cline{3-8} \cline{4-8} \cline{5-8} \cline{6-8} \cline{7-8} \cline{8-8} 
 &  & {\scriptsize{}Mean} & {\scriptsize{}Min.} & {\scriptsize{}25th} & {\scriptsize{}Median.} & {\scriptsize{}75th} & {\scriptsize{}Max.} & {\scriptsize{}CPU time (s)} & {\scriptsize{}(\%)} & {\footnotesize{}$\log_{10}\left(DIST\right)$} & {\scriptsize{}(\%)}\tabularnewline
\hline 
{\scriptsize{}$\delta$-(0) (BLP)} & {\scriptsize{}25} & {\scriptsize{}2000} & {\scriptsize{}2000} & {\scriptsize{}2000} & {\scriptsize{}2000} & {\scriptsize{}2000} & {\scriptsize{}2000} & {\scriptsize{}0.103} & {\scriptsize{}0} & {\scriptsize{}-3.8} & {\scriptsize{}0}\tabularnewline
{\scriptsize{}$\delta$-(0) (BLP) + Anderson} & {\scriptsize{}25} & {\scriptsize{}6} & {\scriptsize{}6} & {\scriptsize{}6} & {\scriptsize{}6} & {\scriptsize{}6} & {\scriptsize{}6} & {\scriptsize{}0.002} & {\scriptsize{}0} & {\scriptsize{}NaN} & {\scriptsize{}0}\tabularnewline
{\scriptsize{}$\delta$-(0) (BLP) + Spectral} & {\scriptsize{}25} & {\scriptsize{}41} & {\scriptsize{}41} & {\scriptsize{}41} & {\scriptsize{}41} & {\scriptsize{}41} & {\scriptsize{}41} & {\scriptsize{}0.004} & {\scriptsize{}100} & {\scriptsize{}-13.7} & {\scriptsize{}100}\tabularnewline
{\scriptsize{}$\delta$-(0) (BLP) + SQUAREM} & {\scriptsize{}25} & {\scriptsize{}119} & {\scriptsize{}119} & {\scriptsize{}119} & {\scriptsize{}119} & {\scriptsize{}119} & {\scriptsize{}119} & {\scriptsize{}0.007} & {\scriptsize{}100} & {\scriptsize{}-13.1} & {\scriptsize{}100}\tabularnewline
{\scriptsize{}$\delta$-(1)} & {\scriptsize{}25} & {\scriptsize{}2000} & {\scriptsize{}2000} & {\scriptsize{}2000} & {\scriptsize{}2000} & {\scriptsize{}2000} & {\scriptsize{}2000} & {\scriptsize{}0.111} & {\scriptsize{}0} & {\scriptsize{}-4.2} & {\scriptsize{}0}\tabularnewline
{\scriptsize{}$\delta$-(1) (BLP) + Anderson} & {\scriptsize{}25} & {\scriptsize{}17} & {\scriptsize{}17} & {\scriptsize{}17} & {\scriptsize{}17} & {\scriptsize{}17} & {\scriptsize{}17} & {\scriptsize{}0.004} & {\scriptsize{}100} & {\scriptsize{}NaN} & {\scriptsize{}100}\tabularnewline
{\scriptsize{}$\delta$-(1) + Spectral} & {\scriptsize{}25} & {\scriptsize{}98} & {\scriptsize{}98} & {\scriptsize{}98} & {\scriptsize{}98} & {\scriptsize{}98} & {\scriptsize{}98} & {\scriptsize{}0.008} & {\scriptsize{}100} & {\scriptsize{}-14.9} & {\scriptsize{}100}\tabularnewline
{\scriptsize{}$\delta$-(1) + SQUAREM} & {\scriptsize{}25} & {\scriptsize{}35} & {\scriptsize{}35} & {\scriptsize{}35} & {\scriptsize{}35} & {\scriptsize{}35} & {\scriptsize{}35} & {\scriptsize{}0.002} & {\scriptsize{}100} & {\scriptsize{}-14.2} & {\scriptsize{}100}\tabularnewline
\hline 
{\scriptsize{}$V$-(0)} & {\scriptsize{}25} & {\scriptsize{}2000} & {\scriptsize{}2000} & {\scriptsize{}2000} & {\scriptsize{}2000} & {\scriptsize{}2000} & {\scriptsize{}2000} & {\scriptsize{}0.146} & {\scriptsize{}0} & {\scriptsize{}-3.4} & {\scriptsize{}0}\tabularnewline
{\scriptsize{}$V$-(0) + Anderson} & {\scriptsize{}25} & {\scriptsize{}25} & {\scriptsize{}25} & {\scriptsize{}25} & {\scriptsize{}25} & {\scriptsize{}25} & {\scriptsize{}25} & {\scriptsize{}0.006} & {\scriptsize{}0} & {\scriptsize{}NaN} & {\scriptsize{}0}\tabularnewline
{\scriptsize{}$V$-(0) + Spectral} & {\scriptsize{}25} & {\scriptsize{}36} & {\scriptsize{}36} & {\scriptsize{}36} & {\scriptsize{}36} & {\scriptsize{}36} & {\scriptsize{}36} & {\scriptsize{}0.004} & {\scriptsize{}100} & {\scriptsize{}-14.9} & {\scriptsize{}100}\tabularnewline
{\scriptsize{}$V$-(0) + SQUAREM} & {\scriptsize{}25} & {\scriptsize{}1013} & {\scriptsize{}1013} & {\scriptsize{}1013} & {\scriptsize{}1013} & {\scriptsize{}1013} & {\scriptsize{}1013} & {\scriptsize{}0.072} & {\scriptsize{}100} & {\scriptsize{}-13.1} & {\scriptsize{}100}\tabularnewline
{\scriptsize{}$V$-(1)} & {\scriptsize{}25} & {\scriptsize{}2000} & {\scriptsize{}2000} & {\scriptsize{}2000} & {\scriptsize{}2000} & {\scriptsize{}2000} & {\scriptsize{}2000} & {\scriptsize{}0.117} & {\scriptsize{}0} & {\scriptsize{}-4.9} & {\scriptsize{}0}\tabularnewline
{\scriptsize{}$V$-(1) + Anderson} & {\scriptsize{}25} & {\scriptsize{}7} & {\scriptsize{}7} & {\scriptsize{}7} & {\scriptsize{}7} & {\scriptsize{}7} & {\scriptsize{}7} & {\scriptsize{}0.001} & {\scriptsize{}0} & {\scriptsize{}NaN} & {\scriptsize{}0}\tabularnewline
{\scriptsize{}$V$-(1) + Spectral} & {\scriptsize{}25} & {\scriptsize{}36} & {\scriptsize{}36} & {\scriptsize{}36} & {\scriptsize{}36} & {\scriptsize{}36} & {\scriptsize{}36} & {\scriptsize{}0.003} & {\scriptsize{}100} & {\scriptsize{}-15.5} & {\scriptsize{}100}\tabularnewline
{\scriptsize{}$V$-(1) + SQUAREM} & {\scriptsize{}25} & {\scriptsize{}58} & {\scriptsize{}58} & {\scriptsize{}58} & {\scriptsize{}58} & {\scriptsize{}58} & {\scriptsize{}58} & {\scriptsize{}0.003} & {\scriptsize{}100} & {\scriptsize{}-14} & {\scriptsize{}100}\tabularnewline
\hline 
\end{tabular}
\par\end{centering}
}

{\footnotesize{}Notes.}{\footnotesize\par}

{\footnotesize{}$\ensuremath{DIST\equiv\left\Vert \log(S^{(data)})-\log(s)\right\Vert _{\infty}},$$\epsilon_{tol}=$1E-12.}{\footnotesize\par}

{\footnotesize{}The maximum number of function evaluations is set to 2000.}{\footnotesize\par}
\end{table}

\begin{figure}[H]
\caption{Example of non-contraction mapping $\delta$-(1)\label{fig:Example-of-non-contraction}}

\begin{centering}
\includegraphics[scale=0.5]{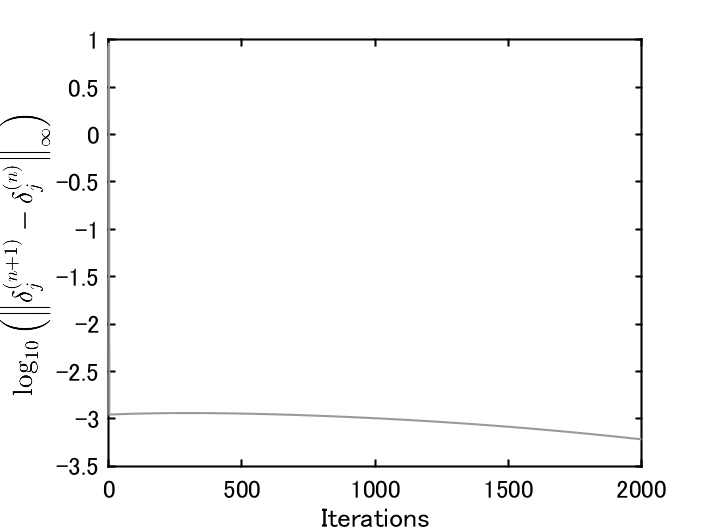}\includegraphics[scale=0.5]{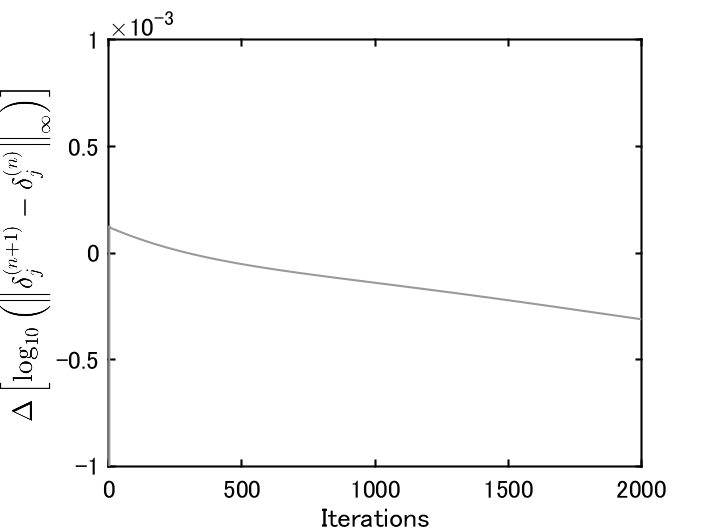}
\par\end{centering}
{\footnotesize{}Note.}{\footnotesize\par}

{\footnotesize{}$\Delta\log_{10}\left(\left\Vert \delta^{(n+1)}-\delta^{(n)}\right\Vert _{\infty}\right)$ is defined by $\log_{10}\left(\left\Vert \delta^{(n+2)}-\delta^{(n+1)}\right\Vert _{\infty}\right)-\log_{10}\left(\left\Vert \delta^{(n+1)}-\delta^{(n)}\right\Vert _{\infty}\right)$.}{\footnotesize\par}
\end{figure}

\subsection{Choice of step sizes in the spectral/SQUAREM algorithm\label{subsec:Choice-of-step-sizes}}

\subsubsection*{Desirable sign of the step size $\alpha^{(n)}$}

As discussed in Section 5 of the main article, $\alpha^{(n)}$ should be chosen to accelerate convergence unless it gets unstable. Generally, choosing $\alpha^{(n)}<0$ is not a good choice, because it might lead to the opposite direction of the original fixed-point iteration $x^{(n+1)}=\Phi(x^{(n)})$ when $\Phi$ is a contraction. Figure \ref{fig:Idea-of-Extrapolation} illustrates it, and the following proposition shows a formal statement: 
\begin{prop}
\label{prop:spectral_step_size_discussion}Suppose that $\Phi:\mathbb{R}^{n}\rightarrow\mathbb{R}^{n}$ is a contraction mapping of modulus $K\in[0,1)$ on $L^{2}$ metric space. Let $x^{*}$ be the unique solution of $x=\Phi(x)$. Then, the following holds for $\alpha<0$:

\begin{eqnarray*}
\left\Vert \left(\alpha\Phi(x)+(1-\alpha)x\right)-x^{*}\right\Vert _{2} & > & \left\Vert x-x^{*}\right\Vert _{2}.
\end{eqnarray*}
\end{prop}
Proposition \ref{prop:spectral_step_size_discussion} shows that $\alpha^{(n)}\Phi(x)+(1-\alpha^{(n)})x$ is further from the solution $x^{*}$ than $x$, and it might lead to divergence of the iterations, when $\Phi$ is a contraction mapping and $\alpha^{(n)}$ takes a negative value.\footnote{Note that we cannot deny the possibility that the iteration converges faster when setting $\alpha^{(n)}<0$ for some $n$. Nevertheless, in general, it would be natural to expect that setting $\alpha^{(n)}<0$ for some $n$ would make the convergence slower and more unstable.} Even when there is no guarantee that $\Phi$ is a contraction, choosing a negative $\alpha$ might not be adequate when $\Phi$ is ``close to'' a contraction.

\begin{figure}[H]
\caption{Idea of the extrapolation method and the spectral algorithm\label{fig:Idea-of-Extrapolation}}

\begin{centering}
\includegraphics[scale=0.6]{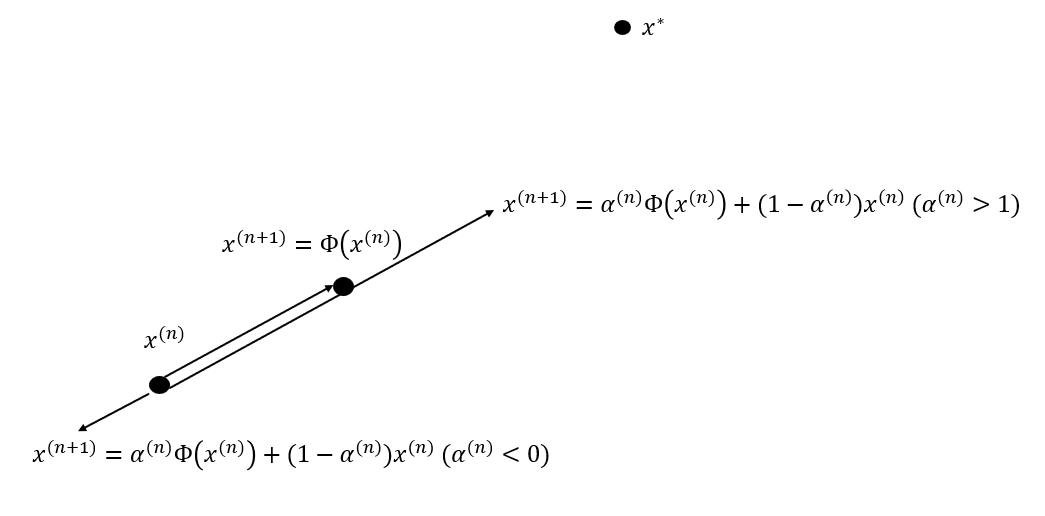}
\par\end{centering}
{\footnotesize{}Notes.}{\footnotesize\par}

{\footnotesize{}$x^{*}$ denotes the solution of the fixed-point problem $x=\Phi(x)$.}{\footnotesize\par}

{\footnotesize{}Proposition \ref{prop:spectral_step_size_discussion} implies $x^{(n+1)}=\alpha^{(n)}\Phi(x^{(n)})+(1-\alpha^{(n)})x^{(n)}$ is further than $x^{(n)}$ from $x^{*}$, when $\alpha^{(n)}<0$ holds. In contrast, if a positive value $\alpha^{(n)}$ is appropriately chosen, $x^{(n+1)}$ can be closer to $x^{*}$ than the case of $x^{(n)}$.}{\footnotesize\par}
\end{figure}

\subsubsection*{Numerical results}

To clarify the importance of the choice of step sizes $\alpha$, I show the results of static BLP Monte Carlo simulations under alternative choices of step sizes ($\alpha_{S1}\equiv-\frac{s^{(n)\prime}y^{(n)}}{y^{(n)\prime}y^{(n)}},\alpha_{S2}\equiv-\frac{s^{(n)\prime}s^{(n)}}{s^{(n)\prime}y^{(n)}}$), which are mentioned in the previous literature (\citealp{reynaerts2012enhencing}; \citealp{conlon2020best}; \citealp{pal2023comparing}). In addition to these step sizes, I also experiment the step size $\alpha_{S3}^{\prime}\equiv sgn\left(s^{(n)\prime}y^{(n)}\right)\frac{\left\Vert s^{(n)}\right\Vert _{2}}{\left\Vert y^{(n)}\right\Vert _{2}}$, which is introduced in the BB package of R language. Here, $sgn\left(s^{(n)\prime}y^{(n)}\right)$ denotes the sign of $s^{(n)\prime}y^{(n)}$. Results are shown in Tables \ref{tab:static_BLP_Monte_Carlo_J_250_alpha_S1}, \ref{tab:static_BLP_Monte_Carlo_J_250_alpha_S2}, and \ref{tab:static_BLP_Monte_Carlo_J_250_alpha_S4}. Unlike the case of $\alpha_{S3}$, the algorithms sometimes do not converge, and lead to slower convergence under these specifications. The results imply choosing the step size $\alpha_{S3}$ is better from the viewpoint of stable and fast convergence.

Note that one remedial measure for using step sizes other than $\alpha_{S3}$ is to set $\alpha\leftarrow\min\left(\alpha,\alpha_{min}\right)$, where $\alpha_{min}$ is a small positive number like 1E-10.\footnote{DF-SANE function (spectral) in scipy package used in PyBLP introduces $\alpha_{min}$ whose default size is 1E-10.} However, it might stagnate the convergence process, and it seems that choosing $\alpha_{S3}$ is desirable.
\begin{center}
\begin{table}[H]
\caption{Results of the Monte Carlo simulation (Static BLP model; Continuous consumer types; $\alpha_{S1}$)\label{tab:static_BLP_Monte_Carlo_J_250_alpha_S1}}

\scalebox{0.8}{
\begin{centering}
\begin{tabular}{cccccccccccc}
\hline 
 & \multirow{2}{*}{{\scriptsize{}$J$}} & \multicolumn{6}{c}{{\scriptsize{}Func. Evals.}} & {\scriptsize{}Mean} & {\scriptsize{}Conv.} & {\footnotesize{}Mean} & {\scriptsize{}$DIST<\epsilon_{tol}$}\tabularnewline
\cline{3-8} \cline{4-8} \cline{5-8} \cline{6-8} \cline{7-8} \cline{8-8} 
 &  & {\scriptsize{}Mean} & {\scriptsize{}Min.} & {\scriptsize{}25th} & {\scriptsize{}Median.} & {\scriptsize{}75th} & {\scriptsize{}Max.} & {\scriptsize{}CPU time (s)} & {\scriptsize{}(\%)} & {\footnotesize{}$\log_{10}\left(DIST\right)$} & {\scriptsize{}(\%)}\tabularnewline
\hline 
{\scriptsize{}$\delta$-(0) (BLP) + Spectral} & {\scriptsize{}250} & {\scriptsize{}12.06} & {\scriptsize{}9} & {\scriptsize{}11} & {\scriptsize{}12} & {\scriptsize{}13} & {\scriptsize{}22} & {\scriptsize{}0.0468} & {\scriptsize{}100} & {\scriptsize{}-15.5} & {\scriptsize{}100}\tabularnewline
{\scriptsize{}$\delta$-(0) (BLP) + SQUAREM} & {\scriptsize{}250} & {\scriptsize{}92.76} & {\scriptsize{}13} & {\scriptsize{}22} & {\scriptsize{}27.5} & {\scriptsize{}44} & {\scriptsize{}1000} & {\scriptsize{}0.3513} & {\scriptsize{}94} & {\scriptsize{}-13.9} & {\scriptsize{}94}\tabularnewline
{\scriptsize{}$\delta$-(1) + Spectral} & {\scriptsize{}250} & {\scriptsize{}282.2} & {\scriptsize{}14} & {\scriptsize{}23} & {\scriptsize{}30.5} & {\scriptsize{}1000} & {\scriptsize{}1000} & {\scriptsize{}1.08664} & {\scriptsize{}74} & {\scriptsize{}-12.2} & {\scriptsize{}76}\tabularnewline
{\scriptsize{}$\delta$-(1) + SQUAREM} & {\scriptsize{}250} & {\scriptsize{}9.76} & {\scriptsize{}7} & {\scriptsize{}9} & {\scriptsize{}10} & {\scriptsize{}11} & {\scriptsize{}14} & {\scriptsize{}0.0381} & {\scriptsize{}100} & {\scriptsize{}-15.9} & {\scriptsize{}100}\tabularnewline
\hline 
{\scriptsize{}$V$-(0) + Spectral} & {\scriptsize{}250} & {\scriptsize{}13.8} & {\scriptsize{}8} & {\scriptsize{}11} & {\scriptsize{}14} & {\scriptsize{}16} & {\scriptsize{}21} & {\scriptsize{}0.05284} & {\scriptsize{}100} & {\scriptsize{}-15.1} & {\scriptsize{}100}\tabularnewline
{\scriptsize{}$V$-(0) + SQUAREM} & {\scriptsize{}250} & {\scriptsize{}14.18} & {\scriptsize{}8} & {\scriptsize{}12} & {\scriptsize{}14} & {\scriptsize{}17} & {\scriptsize{}23} & {\scriptsize{}0.05314} & {\scriptsize{}100} & {\scriptsize{}-15.4} & {\scriptsize{}100}\tabularnewline
{\scriptsize{}$V$-(1) + Spectral} & {\scriptsize{}250} & {\scriptsize{}27.4} & {\scriptsize{}12} & {\scriptsize{}21} & {\scriptsize{}26.5} & {\scriptsize{}33} & {\scriptsize{}47} & {\scriptsize{}0.1242} & {\scriptsize{}100} & {\scriptsize{}-14.9} & {\scriptsize{}100}\tabularnewline
{\scriptsize{}$V$-(1) + SQUAREM} & {\scriptsize{}250} & {\scriptsize{}27.82} & {\scriptsize{}11} & {\scriptsize{}22} & {\scriptsize{}26} & {\scriptsize{}32} & {\scriptsize{}51} & {\scriptsize{}0.12224} & {\scriptsize{}100} & {\scriptsize{}-14.7} & {\scriptsize{}100}\tabularnewline
\hline 
\end{tabular}
\par\end{centering}
}

{\footnotesize{}\input{notes/note_static_BLP_results.tex}}{\footnotesize\par}

{\footnotesize{}The mean outside option share is 0.307.}{\footnotesize\par}
\end{table}
\par\end{center}

\begin{center}
\begin{table}[H]
\caption{Results of the Monte Carlo simulation (Static BLP model; Continuous consumer types; $\alpha_{S2}$)\label{tab:static_BLP_Monte_Carlo_J_250_alpha_S2}}

\scalebox{0.8}{
\begin{centering}
\begin{tabular}{cccccccccccc}
\hline 
 & \multirow{2}{*}{{\scriptsize{}$J$}} & \multicolumn{6}{c}{{\scriptsize{}Number of iterations}} & {\scriptsize{}Mean} & {\scriptsize{}Conv.} & {\footnotesize{}Mean} & {\scriptsize{}$DIST<\epsilon_{tol}$}\tabularnewline
\cline{3-8} \cline{4-8} \cline{5-8} \cline{6-8} \cline{7-8} \cline{8-8} 
 &  & {\scriptsize{}Mean} & {\scriptsize{}Min.} & {\scriptsize{}25th} & {\scriptsize{}Median.} & {\scriptsize{}75th} & {\scriptsize{}Max.} & {\scriptsize{}CPU time (s)} & {\scriptsize{}(\%)} & {\footnotesize{}$\log_{10}\left(DIST\right)$} & {\scriptsize{}(\%)}\tabularnewline
\hline 
{\scriptsize{}$\delta$-(0) (BLP) + Spectral} & {\scriptsize{}250} & {\scriptsize{}12.06} & {\scriptsize{}9} & {\scriptsize{}11} & {\scriptsize{}12} & {\scriptsize{}13} & {\scriptsize{}22} & {\scriptsize{}0.0434} & {\scriptsize{}100} & {\scriptsize{}-15.5} & {\scriptsize{}100}\tabularnewline
{\scriptsize{}$\delta$-(0) (BLP) + SQUAREM} & {\scriptsize{}250} & {\scriptsize{}77.62} & {\scriptsize{}13} & {\scriptsize{}23} & {\scriptsize{}29} & {\scriptsize{}45} & {\scriptsize{}1000} & {\scriptsize{}0.24764} & {\scriptsize{}96} & {\scriptsize{}-14} & {\scriptsize{}96}\tabularnewline
{\scriptsize{}$\delta$-(1) + Spectral} & {\scriptsize{}250} & {\scriptsize{}64.72} & {\scriptsize{}14} & {\scriptsize{}23} & {\scriptsize{}30.5} & {\scriptsize{}63} & {\scriptsize{}1000} & {\scriptsize{}0.20604} & {\scriptsize{}94} & {\scriptsize{}-13.6} & {\scriptsize{}94}\tabularnewline
{\scriptsize{}$\delta$-(1) + SQUAREM} & {\scriptsize{}250} & {\scriptsize{}9.76} & {\scriptsize{}7} & {\scriptsize{}9} & {\scriptsize{}10} & {\scriptsize{}11} & {\scriptsize{}14} & {\scriptsize{}0.03558} & {\scriptsize{}100} & {\scriptsize{}-15.9} & {\scriptsize{}100}\tabularnewline
\hline 
{\scriptsize{}$V$-(0) + Spectral} & {\scriptsize{}250} & {\scriptsize{}13.94} & {\scriptsize{}8} & {\scriptsize{}12} & {\scriptsize{}14} & {\scriptsize{}16} & {\scriptsize{}22} & {\scriptsize{}0.04926} & {\scriptsize{}100} & {\scriptsize{}-15} & {\scriptsize{}100}\tabularnewline
{\scriptsize{}$V$-(0) + SQUAREM} & {\scriptsize{}250} & {\scriptsize{}14.02} & {\scriptsize{}8} & {\scriptsize{}11} & {\scriptsize{}14} & {\scriptsize{}16} & {\scriptsize{}21} & {\scriptsize{}0.0491} & {\scriptsize{}100} & {\scriptsize{}-15.5} & {\scriptsize{}100}\tabularnewline
{\scriptsize{}$V$-(1) + Spectral} & {\scriptsize{}250} & {\scriptsize{}30.46} & {\scriptsize{}12} & {\scriptsize{}24} & {\scriptsize{}26} & {\scriptsize{}32} & {\scriptsize{}114} & {\scriptsize{}0.12136} & {\scriptsize{}100} & {\scriptsize{}-15.1} & {\scriptsize{}100}\tabularnewline
{\scriptsize{}$V$-(1) + SQUAREM} & {\scriptsize{}250} & {\scriptsize{}30.3} & {\scriptsize{}11} & {\scriptsize{}23} & {\scriptsize{}28.5} & {\scriptsize{}34} & {\scriptsize{}63} & {\scriptsize{}0.1199} & {\scriptsize{}100} & {\scriptsize{}-14.9} & {\scriptsize{}100}\tabularnewline
\hline 
\end{tabular}
\par\end{centering}
}

{\footnotesize{}\input{notes/note_static_BLP_results.tex}}{\footnotesize\par}

{\footnotesize{}The mean outside option share is 0.307.}{\footnotesize\par}
\end{table}
\par\end{center}

\begin{center}
\begin{table}[H]
\caption{Results of the Monte Carlo simulation (Static BLP model; Continuous consumer types; $\alpha_{S3}^{\prime}$)\label{tab:static_BLP_Monte_Carlo_J_250_alpha_S4}}

\scalebox{0.8}{
\begin{centering}
\begin{tabular}{cccccccccccc}
\hline 
 & \multirow{2}{*}{{\scriptsize{}$J$}} & \multicolumn{6}{c}{{\scriptsize{}Func. Evals.}} & {\scriptsize{}Mean} & {\scriptsize{}Conv.} & {\footnotesize{}Mean} & {\scriptsize{}$DIST<\epsilon_{tol}$}\tabularnewline
\cline{3-8} \cline{4-8} \cline{5-8} \cline{6-8} \cline{7-8} \cline{8-8} 
 &  & {\scriptsize{}Mean} & {\scriptsize{}Min.} & {\scriptsize{}25th} & {\scriptsize{}Median.} & {\scriptsize{}75th} & {\scriptsize{}Max.} & {\scriptsize{}CPU time (s)} & {\scriptsize{}(\%)} & {\footnotesize{}$\log_{10}\left(DIST\right)$} & {\scriptsize{}(\%)}\tabularnewline
\hline 
{\scriptsize{}$\delta$-(0) (BLP) + Spectral} & {\scriptsize{}250} & {\scriptsize{}12.06} & {\scriptsize{}9} & {\scriptsize{}11} & {\scriptsize{}12} & {\scriptsize{}13} & {\scriptsize{}22} & {\scriptsize{}0.05182} & {\scriptsize{}100} & {\scriptsize{}-15.5} & {\scriptsize{}100}\tabularnewline
{\scriptsize{}$\delta$-(0) (BLP) + SQUAREM} & {\scriptsize{}250} & {\scriptsize{}58.84} & {\scriptsize{}13} & {\scriptsize{}23} & {\scriptsize{}27.5} & {\scriptsize{}36} & {\scriptsize{}1000} & {\scriptsize{}0.22302} & {\scriptsize{}98} & {\scriptsize{}-14} & {\scriptsize{}98}\tabularnewline
{\scriptsize{}$\delta$-(1) + Spectral} & {\scriptsize{}250} & {\scriptsize{}47.38} & {\scriptsize{}14} & {\scriptsize{}22} & {\scriptsize{}30} & {\scriptsize{}43} & {\scriptsize{}435} & {\scriptsize{}0.19196} & {\scriptsize{}100} & {\scriptsize{}-14.5} & {\scriptsize{}100}\tabularnewline
{\scriptsize{}$\delta$-(1) + SQUAREM} & {\scriptsize{}250} & {\scriptsize{}9.76} & {\scriptsize{}7} & {\scriptsize{}9} & {\scriptsize{}10} & {\scriptsize{}11} & {\scriptsize{}14} & {\scriptsize{}0.04252} & {\scriptsize{}100} & {\scriptsize{}-15.9} & {\scriptsize{}100}\tabularnewline
\hline 
{\scriptsize{}$V$-(0) + Spectral} & {\scriptsize{}250} & {\scriptsize{}13.88} & {\scriptsize{}8} & {\scriptsize{}12} & {\scriptsize{}14} & {\scriptsize{}16} & {\scriptsize{}20} & {\scriptsize{}0.06062} & {\scriptsize{}100} & {\scriptsize{}-15.1} & {\scriptsize{}100}\tabularnewline
{\scriptsize{}$V$-(0) + SQUAREM} & {\scriptsize{}250} & {\scriptsize{}14.02} & {\scriptsize{}8} & {\scriptsize{}12} & {\scriptsize{}14} & {\scriptsize{}16} & {\scriptsize{}22} & {\scriptsize{}0.05864} & {\scriptsize{}100} & {\scriptsize{}-15.4} & {\scriptsize{}100}\tabularnewline
{\scriptsize{}$V$-(1) + Spectral} & {\scriptsize{}250} & {\scriptsize{}29.86} & {\scriptsize{}12} & {\scriptsize{}22} & {\scriptsize{}27} & {\scriptsize{}35} & {\scriptsize{}67} & {\scriptsize{}0.14622} & {\scriptsize{}100} & {\scriptsize{}-14.9} & {\scriptsize{}100}\tabularnewline
{\scriptsize{}$V$-(1) + SQUAREM} & {\scriptsize{}250} & {\scriptsize{}30} & {\scriptsize{}11} & {\scriptsize{}23} & {\scriptsize{}28} & {\scriptsize{}33} & {\scriptsize{}62} & {\scriptsize{}0.14686} & {\scriptsize{}100} & {\scriptsize{}-14.7} & {\scriptsize{}100}\tabularnewline
\hline 
\end{tabular}
\par\end{centering}
}

{\footnotesize{}\input{notes/note_static_BLP_results.tex}}{\footnotesize\par}

{\footnotesize{}The mean outside option share is 0.307.}{\footnotesize\par}
\end{table}
\par\end{center}

\subsubsection*{Proof of Proposition \ref{prop:spectral_step_size_discussion}}

Below, we show Proposition \ref{prop:spectral_step_size_discussion}.
\begin{lem}
\label{lem:step_size_lem}$\left\Vert \alpha y+(1-\alpha)z\right\Vert _{2}>\left\Vert z\right\Vert _{2}$ holds for $\alpha<0$ and $y,z$ such that $\left\Vert y\right\Vert _{2}<\left\Vert z\right\Vert _{2}$
\end{lem}
\begin{proof}
Under $\alpha<0$ and $\left\Vert y\right\Vert _{2}<\left\Vert z\right\Vert _{2}$,

\begin{eqnarray*}
\left\Vert \alpha y+(1-\alpha)z\right\Vert _{2}^{2}-\left\Vert z\right\Vert _{2}^{2} & = & \text{\ensuremath{\alpha\sum_{i}\left[\alpha y_{i}^{2}+2(1-\alpha)y_{i}z_{i}+(\alpha-2)z_{i}^{2}\right]}}\\
 & > & \alpha\sum_{i}\left[\alpha z_{i}^{2}+2(1-\alpha)y_{i}z_{i}+(\alpha-2)z_{i}^{2}\right]\ \left(\because\alpha<0,\left\Vert y\right\Vert _{2}<\left\Vert z\right\Vert _{2}\right)\\
 & = & 2\alpha(\alpha-1)\left[\left\Vert z\right\Vert ^{2}-y\cdot z\right].
\end{eqnarray*}

By Cauchy-Schwartz inequality, $y\cdot z\leq\left\Vert y\right\Vert _{2}\left\Vert z\right\Vert _{2}$ holds. Then, by $\left\Vert y\right\Vert _{2}<\left\Vert z\right\Vert _{2}$ and $\alpha<0$, we have $\left\Vert \alpha y+(1-\alpha)z\right\Vert _{2}^{2}-\left\Vert z\right\Vert _{2}^{2}>2\alpha(\alpha-1)\left[\left\Vert z\right\Vert ^{2}-y\cdot z\right]\geq0$.
\end{proof}
\begin{lem}
\label{lem:step_size_lem2}$\left\Vert \alpha\widetilde{x}+(1-\alpha)x-x^{*}\right\Vert _{2}>\left\Vert x-x^{*}\right\Vert _{2}$ holds for $\alpha<0$ and $\widetilde{x}$ such that $\left\Vert \widetilde{x}-x^{*}\right\Vert _{2}<\left\Vert x-x^{*}\right\Vert _{2}$
\end{lem}
\begin{proof}
By letting $y=\widetilde{x}-x^{*}$ and $z=x-x^{*}$ in Lemma \ref{lem:step_size_lem}, we obtain $\left\Vert \alpha\widetilde{x}+(1-\alpha)x-x^{*}\right\Vert _{2}=\left\Vert \alpha y+(1-\alpha)z\right\Vert _{2}>\left\Vert z\right\Vert _{2}=\left\Vert x-x^{*}\right\Vert _{2}$.
\end{proof}
\begin{lem}
\label{lem:contraction_distance}For a contraction mapping $\Phi$ with modulus $K\in[0,1)$, let $x^{*}$ be the solution of $\Phi(x)=x$. Then, $\left\Vert \Phi(x)-x^{*}\right\Vert _{2}<\left\Vert x-x^{*}\right\Vert _{2}$ holds.
\end{lem}
\begin{proof}
$\left\Vert \Phi(x)-x^{*}\right\Vert _{2}=\lim_{n\rightarrow\infty}\left\Vert \Phi(x)-\Phi^{n+1}(x)\right\Vert _{2}\leq K\lim_{n\rightarrow\infty}\left\Vert x-\Phi^{n}(x)\right\Vert _{2}=K\left\Vert x-x^{*}\right\Vert _{2}<\left\Vert x-x^{*}\right\Vert _{2}$.
\end{proof}

\paragraph*{Proof of Proposition \ref{prop:spectral_step_size_discussion}:}
\begin{proof}
By Lemmas \ref{lem:step_size_lem2} and \ref{lem:contraction_distance}, $\left\Vert \alpha\Phi(x)+(1-\alpha)x-x^{*}\right\Vert _{2}>\left\Vert x-x^{*}\right\Vert _{2}$ holds.
\end{proof}

\subsection{\citet{kalouptsidi2012market}'s method\label{subsec:Kalouptsidi-method}}

Here, we discuss the method proposed in \citet{kalouptsidi2012market}. The method was originally developed for static BLP models with a few consumer types.

\citet{kalouptsidi2012market} defined a term $r_{i}\equiv\log(w_{i}s_{i0})$ where $s_{i0}\equiv\frac{1}{\exp(V_{i})}$, and proposed to solve for $r$ by a mapping $F:\mathbb{R}^{I}\rightarrow\mathbb{R}^{I}$ defined by:

\begin{eqnarray*}
F_{i}(r) & = & \begin{cases}
r_{i}+\log\left(w_{i}\right)-\log\left(\sum_{j=1}^{J}S_{j}^{(data)}\frac{\exp\left(\mu_{ij}+r_{i}\right)}{\sum_{i=1}^{I}\exp\left(\mu_{ij}+r_{i}\right)}+S_{0}^{(data)}\frac{\exp\left(r_{i}\right)}{\sum_{i=1}^{I}\exp\left(r_{i}\right)}\right) & i=1,\cdots,I-1,\\
\log\left(S_{0}^{(data)}-\sum_{i=1}^{I-1}\exp\left(F_{i}(r)\right)\right) & i=I.
\end{cases}
\end{eqnarray*}

Note that $\delta_{j}=\log\left(S_{j}^{(data)}\right)-\log\left(\sum_{i=1}^{I}w_{i}\exp\left(\mu_{ij}+r_{i}\right)\right)$ holds. 

Nevertheless, the problem of the algorithm is that the outputs of the mapping $F$ do not necessarily take real numbers when $S_{0}^{(data)}-\sum_{i=1}^{I-1}\exp\left(F_{i}(r)\right)\leq0$. Hence, \citet{kalouptsidi2012market} defined $\widetilde{r}$ such that $\widetilde{r_{i}}\equiv r_{i}-r_{i=I}\ (i=1,\cdots,I-1),\ r_{i=I}=0$, and proposed the alternative mapping $\widetilde{F}:\mathbb{R}^{I-1}\rightarrow\mathbb{R}^{I-1}$ to solve for $\widetilde{r}$:

\begin{eqnarray*}
\widetilde{F_{i}}(\widetilde{r}) & = & \widetilde{r_{i}}+\log\left(w_{i}\right)-\log\left(\sum_{j=1}^{J}S_{j}^{(data)}\frac{\exp\left(\mu_{ij}+\widetilde{r_{i}}\right)}{\sum_{i=1}^{I}\exp\left(\mu_{ij}+\widetilde{r_{i}}\right)}+S_{0}^{(data)}\frac{\exp\left(\widetilde{r_{i}}\right)}{\sum_{i=1}^{I}\exp\left(\widetilde{r_{i}}\right)}\right)\ (i=1,\cdots,I-1).
\end{eqnarray*}
Note that $r_{i}=\widetilde{r_{i}}+\log(S_{0}^{(data)})-\log\left(\sum_{i=1}^{I}\exp(\widetilde{r_{i}})\right)$ holds. 

Based on the observations that the algorithm using $F$ usually converged, and the algorithms using $\widetilde{F}$ were slower, \citet{kalouptsidi2012market} proposed a mixed algorithm, which uses $F$ by default, but switches to $\widetilde{F}$ when $S_{0}^{(data)}-\sum_{i=1}^{I-1}\exp\left(F_{i}(r)\right)<0$ occurs. 

Because $r_{i}=\log(w_{i}s_{i0}^{(ccp)})=\log(w_{i})-V_{i}$, the mapping $F(r)$ is essentially equivalent to the following mapping of $V$:

\begin{eqnarray*}
G_{i}(V) & = & \begin{cases}
\log\left(\sum_{j=1}^{J}S_{j}^{(data)}\frac{\exp\left(\mu_{ij}\right)}{\sum_{i=1}^{I}w_{i}\exp\left(\mu_{ij}-V_{i}\right)}+\frac{S_{0}^{(data)}}{\sum_{i=1}^{I}w_{i}\exp\left(-V_{i}\right)}\right) & i=1,\cdots,I-1,\\
-\log\left(\frac{S_{0}^{(data)}-\sum_{i=1}^{I-1}w_{i}\exp\left(-G_{i}(V)\right)}{w_{i=I}}\right) & i=I.
\end{cases}
\end{eqnarray*}

When trying \citet{kalouptsidi2012market}'s two algorithms in the settings of the main article (continuous consumer types) by introducing 1000 simulation draws, I found they always failed to converge. To validate the performance of the algorithm in the setting with a few consumer types, which \citet{kalouptsidi2012market} originally considered, I also evaluate the performance of the algorithms by setting the number of simulation draws to 2.

Table \ref{tab:static_BLP_Monte_Carlo_2_consumer_types} shows the results. The results show that algorithms $\delta$-(1) and $V$-(1) perform much better than \citet{kalouptsidi2012market}'s algorithms, though the latter performs on average better than the BLP contraction mapping. The results suggest that the performance of \citet{kalouptsidi2012market}'s algorithms is sensitive to the number of consumer types. Still, algorithms $\delta$-(1) and $V$-(1) perform well regardless of the number of consumer types.
\begin{center}
\begin{table}[H]
\caption{Results of the Monte Carlo simulation (Static BLP model; 2 consumer types)\label{tab:static_BLP_Monte_Carlo_2_consumer_types}}

\scalebox{0.8}{
\begin{centering}
\begin{tabular}{cccccccccccc}
\hline 
 & \multirow{2}{*}{{\scriptsize{}$J$}} & \multicolumn{6}{c}{{\scriptsize{}Func. Evals.}} & {\scriptsize{}Mean} & {\scriptsize{}Conv.} & {\footnotesize{}Mean} & {\scriptsize{}$DIST<\epsilon_{tol}$}\tabularnewline
\cline{3-8} \cline{4-8} \cline{5-8} \cline{6-8} \cline{7-8} \cline{8-8} 
 &  & {\scriptsize{}Mean} & {\scriptsize{}Min.} & {\scriptsize{}25th} & {\scriptsize{}Median.} & {\scriptsize{}75th} & {\scriptsize{}Max.} & {\scriptsize{}CPU time (s)} & {\scriptsize{}(\%)} & {\footnotesize{}$\log_{10}\left(DIST\right)$} & {\scriptsize{}(\%)}\tabularnewline
\hline 
{\scriptsize{}$\delta$-(0) (BLP)} & {\scriptsize{}250} & {\scriptsize{}245.7} & {\scriptsize{}17} & {\scriptsize{}50} & {\scriptsize{}97.5} & {\scriptsize{}233} & {\scriptsize{}1000} & {\scriptsize{}0.02254} & {\scriptsize{}88} & {\scriptsize{}-12.8} & {\scriptsize{}90}\tabularnewline
{\scriptsize{}$\delta$-(0) (BLP) + Anderson} & {\scriptsize{}250} & {\scriptsize{}13.04} & {\scriptsize{}7} & {\scriptsize{}10} & {\scriptsize{}12} & {\scriptsize{}14} & {\scriptsize{}27} & {\scriptsize{}0.00266} & {\scriptsize{}96} & {\scriptsize{}NaN} & {\scriptsize{}96}\tabularnewline
{\scriptsize{}$\delta$-(0) (BLP) + Spectral} & {\scriptsize{}250} & {\scriptsize{}41.64} & {\scriptsize{}10} & {\scriptsize{}19} & {\scriptsize{}26.5} & {\scriptsize{}38} & {\scriptsize{}498} & {\scriptsize{}0.00552} & {\scriptsize{}100} & {\scriptsize{}-14.5} & {\scriptsize{}100}\tabularnewline
{\scriptsize{}$\delta$-(0) (BLP) + SQUAREM} & {\scriptsize{}250} & {\scriptsize{}40.54} & {\scriptsize{}12} & {\scriptsize{}21} & {\scriptsize{}27.5} & {\scriptsize{}42} & {\scriptsize{}185} & {\scriptsize{}0.00362} & {\scriptsize{}100} & {\scriptsize{}-14.5} & {\scriptsize{}100}\tabularnewline
{\scriptsize{}$\delta$-(1)} & {\scriptsize{}250} & {\scriptsize{}14.86} & {\scriptsize{}6} & {\scriptsize{}10} & {\scriptsize{}12} & {\scriptsize{}17} & {\scriptsize{}46} & {\scriptsize{}0.00168} & {\scriptsize{}100} & {\scriptsize{}-14.8} & {\scriptsize{}100}\tabularnewline
{\scriptsize{}$\delta$-(1) + Anderson} & {\scriptsize{}250} & {\scriptsize{}8.1} & {\scriptsize{}5} & {\scriptsize{}6} & {\scriptsize{}9} & {\scriptsize{}9} & {\scriptsize{}12} & {\scriptsize{}0.00166} & {\scriptsize{}100} & {\scriptsize{}-16.2} & {\scriptsize{}100}\tabularnewline
{\scriptsize{}$\delta$-(1) + Spectral} & {\scriptsize{}250} & {\scriptsize{}9.06} & {\scriptsize{}5} & {\scriptsize{}7} & {\scriptsize{}8.5} & {\scriptsize{}10} & {\scriptsize{}20} & {\scriptsize{}0.0013} & {\scriptsize{}100} & {\scriptsize{}-15.6} & {\scriptsize{}100}\tabularnewline
{\scriptsize{}$\delta$-(1) + SQUAREM} & {\scriptsize{}250} & {\scriptsize{}9.8} & {\scriptsize{}5} & {\scriptsize{}8} & {\scriptsize{}10} & {\scriptsize{}11} & {\scriptsize{}17} & {\scriptsize{}0.001} & {\scriptsize{}100} & {\scriptsize{}-16} & {\scriptsize{}100}\tabularnewline
\hline 
{\scriptsize{}$V$-(0)} & {\scriptsize{}250} & {\scriptsize{}250.04} & {\scriptsize{}16} & {\scriptsize{}51} & {\scriptsize{}110.5} & {\scriptsize{}242} & {\scriptsize{}1000} & {\scriptsize{}0.0295} & {\scriptsize{}88} & {\scriptsize{}-12.8} & {\scriptsize{}90}\tabularnewline
{\scriptsize{}$V$-(0) + Anderson} & {\scriptsize{}250} & {\scriptsize{}12.12} & {\scriptsize{}5} & {\scriptsize{}10} & {\scriptsize{}12} & {\scriptsize{}13} & {\scriptsize{}24} & {\scriptsize{}0.00282} & {\scriptsize{}64} & {\scriptsize{}NaN} & {\scriptsize{}64}\tabularnewline
{\scriptsize{}$V$-(0) + Spectral} & {\scriptsize{}250} & {\scriptsize{}16.88} & {\scriptsize{}8} & {\scriptsize{}12} & {\scriptsize{}15} & {\scriptsize{}18} & {\scriptsize{}45} & {\scriptsize{}0.0024} & {\scriptsize{}100} & {\scriptsize{}-15.4} & {\scriptsize{}100}\tabularnewline
{\scriptsize{}$V$-(0) + SQUAREM} & {\scriptsize{}250} & {\scriptsize{}20.5} & {\scriptsize{}9} & {\scriptsize{}14} & {\scriptsize{}17.5} & {\scriptsize{}25} & {\scriptsize{}52} & {\scriptsize{}0.00234} & {\scriptsize{}100} & {\scriptsize{}-15.2} & {\scriptsize{}100}\tabularnewline
{\scriptsize{}$V$-(1)} & {\scriptsize{}250} & {\scriptsize{}15.18} & {\scriptsize{}7} & {\scriptsize{}10} & {\scriptsize{}12} & {\scriptsize{}16} & {\scriptsize{}49} & {\scriptsize{}0.00206} & {\scriptsize{}100} & {\scriptsize{}-14.7} & {\scriptsize{}100}\tabularnewline
{\scriptsize{}$V$-(1) + Anderson} & {\scriptsize{}250} & {\scriptsize{}7.32} & {\scriptsize{}5} & {\scriptsize{}6} & {\scriptsize{}7} & {\scriptsize{}9} & {\scriptsize{}13} & {\scriptsize{}0.00158} & {\scriptsize{}90} & {\scriptsize{}NaN} & {\scriptsize{}90}\tabularnewline
{\scriptsize{}$V$-(1) + Spectral} & {\scriptsize{}250} & {\scriptsize{}10.88} & {\scriptsize{}6} & {\scriptsize{}7} & {\scriptsize{}9.5} & {\scriptsize{}13} & {\scriptsize{}25} & {\scriptsize{}0.00166} & {\scriptsize{}100} & {\scriptsize{}-15.2} & {\scriptsize{}100}\tabularnewline
{\scriptsize{}$V$-(1) + SQUAREM} & {\scriptsize{}250} & {\scriptsize{}12.54} & {\scriptsize{}6} & {\scriptsize{}9} & {\scriptsize{}10.5} & {\scriptsize{}17} & {\scriptsize{}26} & {\scriptsize{}0.00158} & {\scriptsize{}100} & {\scriptsize{}-15} & {\scriptsize{}100}\tabularnewline
\hline 
{\scriptsize{}Kalouptsidi (2012) (1)} & {\scriptsize{}250} & {\scriptsize{}56.46} & {\scriptsize{}45} & {\scriptsize{}46} & {\scriptsize{}48} & {\scriptsize{}54} & {\scriptsize{}250} & {\scriptsize{}0.00246} & {\scriptsize{}100} & {\scriptsize{}-14.4} & {\scriptsize{}100}\tabularnewline
{\scriptsize{}Kalouptsidi (2012) (2)} & {\scriptsize{}250} & {\scriptsize{}53.72} & {\scriptsize{}45} & {\scriptsize{}46} & {\scriptsize{}48} & {\scriptsize{}54} & {\scriptsize{}113} & {\scriptsize{}0.00236} & {\scriptsize{}100} & {\scriptsize{}-14.4} & {\scriptsize{}100}\tabularnewline
\hline 
\end{tabular}
\par\end{centering}
}

{\footnotesize{}Notes.}{\footnotesize\par}

{\footnotesize{}The number of simulation draws is set to 2 to represent the model with two consumer types.}{\footnotesize\par}

{\footnotesize{}Kalouptsidi (2012) (1) denotes the mixed algorithm updating $r$. Kalouptsidi (2012) (2) denotes the algorithm updating $\widetilde{r}$.}{\footnotesize\par}
\end{table}
\par\end{center}

\citet{kalouptsidi2012market} argued that her algorithm works well under $I\ll J$ because the algorithm solves for $I$ (lower)-dimensional vector rather than $J$ (higher)-dimensional vector. Nevertheless, it seems the discussion is not correct, because the algorithms $V$-(0), $V$-(1) worked mostly the same as the algorithms $\delta$-(0), $\delta$-(1), even when $I=1000>25=J$, as demonstrated in Section 6.1. As discussed in Section 3, mappings of $V$ and $\delta$ have dualistic relations. Based on the discussion in Section 3, in the context of static BLP estimation, it seems that what determines the convergence speed is not the dimensions of variables we solve for, but the form of fixed-point mappings. 

\subsection{Performance of time-dependent step sizes in dynamic BLP models\label{subsec:Performance-of-time-dependent}}

In this section, I compare the performance of the algorithm $V$-(1) combined with the spectral algorithm with time-dependent step size $\alpha_{t}^{(n)}$ ($V$-(1) + Spectral ($t$-dep)), and the algorithm $V$-(1) combined with the spectral algorithm with time-independent step size $\alpha^{(n)}$ ($V$-(1) + Spectral ($t$-indep)). I picked up one parameter setting among the 20 settings experimented in the dynamic BLP model under perfect foresight, and assessed these performances. Table \ref{tab:Comparison-time-dep-step-size} shows the results. The results imply that introducing time-dependent step sizes in the spectral algorithm leads to faster convergence.

\begin{table}[H]
\caption{Performance of time-dependent step sizes\label{tab:Comparison-time-dep-step-size}}

\scalebox{0.8}{
\begin{centering}
\begin{tabular}{cccccc}
\hline 
 & \multirow{2}{*}{{\footnotesize{}Func evals ($\Psi_{V\delta\rightarrow V}^{V,\gamma}$)}} & {\footnotesize{}Mean} & \multirow{2}{*}{{\footnotesize{}Conv.}} & \multirow{2}{*}{{\footnotesize{}$\log_{10}\left(DIST\right)$}} & \multirow{2}{*}{{\footnotesize{}$DIST<\epsilon_{tol}$}}\tabularnewline
 &  & {\footnotesize{}CPU time (s)} &  &  & \tabularnewline
\hline 
{\footnotesize{}$V$-(1) + Spectral ($t$-dep)} & {\footnotesize{}2177} & {\footnotesize{}4.339} & {\footnotesize{}1} & {\footnotesize{}-14.8574} & {\footnotesize{}1}\tabularnewline
{\footnotesize{}$V$-(1) + Spectral ($t$-indep)} & {\footnotesize{}388} & {\footnotesize{}0.6} & {\footnotesize{}1} & {\footnotesize{}-14.7972} & {\footnotesize{}1}\tabularnewline
\hline 
\end{tabular}
\par\end{centering}
}
\end{table}

\subsection{Comparison of two traditional dynamic BLP algorithms (nested vs joint update)\label{subsec:Comparison-between-nested-joint}}

In this section, I compare the performance of the nested version ($\delta V$-(1) (nested)) and joint-update version ($\delta V$-(1) (joint)) of the traditional dynamic BLP algorithms. I picked up one parameter setting among the 20 settings experimented in Section 6.2, and assessed these performances. In both algorithms, I combine the spectral algorithm. Regarding the nested version of the algorithm, I introduce ``hot-start'' procedure to speed up the convergence, and let $V^{(0,n)}=V^{*(n-1)}\ (n\geq1)$. 

Tables \ref{tab:Comparison-perfect-foresight-traditional} and \ref{tab:Comparison-IVS-traditional} show the results. The results show that the algorithms jointly updating the variables require more computation time than the algorithm using nested updating steps. Though the proposed algorithm $V$-(1) is much faster, $\delta V$-(1) (joint) is faster than $\delta V$-(1) (nested), at least in the current setting.

\begin{table}[H]
\caption{Comparison between two traditional dynamic BLP algorithms (joint vs nested; Perfect foresight)\label{tab:Comparison-perfect-foresight-traditional}}

\scalebox{0.8}{
\begin{centering}
\begin{tabular}{ccccccc}
\hline 
 & \multirow{2}{*}{{\footnotesize{}Func evals ($\Psi_{V\delta\rightarrow V}^{V,\gamma}$)}} & \multirow{2}{*}{{\footnotesize{}Func evals ($\Phi_{\delta V\rightarrow\delta}^{\gamma,\phi}$)}} & \multirow{2}{*}{{\footnotesize{}CPU time (s)}} & \multirow{2}{*}{{\footnotesize{}Conv.}} & \multirow{2}{*}{{\footnotesize{}$\log_{10}\left(DIST\right)$}} & \multirow{2}{*}{{\footnotesize{}$DIST<\epsilon_{tol}$}}\tabularnewline
 &  &  &  &  &  & \tabularnewline
\hline 
{\footnotesize{}$\delta V$-(1) (joint) + Spectral} & {\footnotesize{}516} & {\footnotesize{}516} & {\footnotesize{}1.119} & {\footnotesize{}1} & {\footnotesize{}-14.5588} & {\footnotesize{}1}\tabularnewline
{\footnotesize{}$\delta V$-(1) (nested)+ Spectral} & {\footnotesize{}2725} & {\footnotesize{}17} & {\footnotesize{}2.562} & {\footnotesize{}1} & {\footnotesize{}-14.5166} & {\footnotesize{}1}\tabularnewline
\hline 
\end{tabular}
\par\end{centering}
}
\end{table}

\begin{table}[H]
\caption{Comparison between two traditional dynamic BLP algorithms (joint vs nested; Inclusive value sufficiency)\label{tab:Comparison-IVS-traditional}}

\scalebox{0.8}{
\begin{centering}
\begin{tabular}{ccccccc}
\hline 
 & \multirow{2}{*}{{\footnotesize{}Func evals ($\Psi_{V\delta\rightarrow V}^{V,\gamma}$)}} & \multirow{2}{*}{{\footnotesize{}Func evals ($\Phi_{\delta V\rightarrow\delta}^{\gamma,\phi}$)}} & \multirow{2}{*}{{\footnotesize{}CPU time (s)}} & \multirow{2}{*}{{\footnotesize{}Conv.}} & \multirow{2}{*}{{\footnotesize{}$\log_{10}\left(DIST\right)$}} & \multirow{2}{*}{{\footnotesize{}$DIST<\epsilon_{tol}$}}\tabularnewline
 &  &  &  &  &  & \tabularnewline
\hline 
{\footnotesize{}$\delta V$-(1) (joint) + Spectral} & {\footnotesize{}241} & {\footnotesize{}241} & {\footnotesize{}1.077} & {\footnotesize{}1} & {\footnotesize{}-15.2409} & {\footnotesize{}1}\tabularnewline
{\footnotesize{}$\delta V$-(1) (nested)+ Spectral} & {\footnotesize{}1968} & {\footnotesize{}19} & {\footnotesize{}6.843} & {\footnotesize{}1} & {\footnotesize{}-14.8388} & {\footnotesize{}1}\tabularnewline
\hline 
\end{tabular}
\par\end{centering}
}
\end{table}

\subsection{The effect of inner loop error on the outer loop convergence\label{subsec:The-effect-of-inner-loop-error}}

In Table 4 of Section 6, I show the results of replication exercises using Berry et al. (\citeyear{berry1995automobile}, \citeyear{berry1999voluntary})'s dataset by setting the outer-loop tolerance to 1E-4. Below, I show the results under the setting where I set the outer-loop tolerance to 1E-8 as in \citet{conlon2020best}. Table \ref{tab:Berry-et-al.-replication-tight-outer-tolerance} shows the results. 
\begin{center}
\begin{table}[H]
\caption{Estimation using the Berry et al.\citeyearpar{berry1995automobile,berry1999voluntary}'s dataset (Tight outer-loop tolerance case)\label{tab:Berry-et-al.-replication-tight-outer-tolerance}}

\begin{centering}
{\footnotesize{}}%
\begin{tabular}{ccccc}
\hline 
 & {\footnotesize{}Mean feval} & {\footnotesize{}Total obj eval} & {\footnotesize{}Total feval} & {\footnotesize{}Objective}\tabularnewline
\hline 
\hline 
{\footnotesize{}$\delta$-(1)} & {\footnotesize{}199.936} & {\footnotesize{}124} & {\footnotesize{}495841} & {\footnotesize{}497.336}\tabularnewline
{\footnotesize{}$\delta$-(1) + Anderson} & {\footnotesize{}15.444} & {\footnotesize{}184} & {\footnotesize{}56833} & {\footnotesize{}497.366}\tabularnewline
{\footnotesize{}$\delta$-(1) + Spectral} & {\footnotesize{}42.319} & {\footnotesize{}129} & {\footnotesize{}109182} & {\footnotesize{}497.336}\tabularnewline
{\footnotesize{}$\delta$-(1) + SQUAREM} & {\footnotesize{}45.244} & {\footnotesize{}117} & {\footnotesize{}105870} & {\footnotesize{}497.336}\tabularnewline
{\footnotesize{}$\delta$-(0)} & {\footnotesize{}232.532} & {\footnotesize{}106} & {\footnotesize{}492968} & {\footnotesize{}497.336}\tabularnewline
{\footnotesize{}$\delta$-(0) + Anderson} & {\footnotesize{}15.508} & {\footnotesize{}154} & {\footnotesize{}47765} & {\footnotesize{}497.336}\tabularnewline
{\footnotesize{}$\delta$-(0) + Spectral} & {\footnotesize{}45.236} & {\footnotesize{}151} & {\footnotesize{}136613} & {\footnotesize{}497.336}\tabularnewline
{\footnotesize{}$\delta$-(0) + SQUAREM} & {\footnotesize{}43.871} & {\footnotesize{}177} & {\footnotesize{}155302} & {\footnotesize{}497.336}\tabularnewline
\hline 
\end{tabular}{\footnotesize\par}
\par\end{centering}

\input{notes/note_pyblp_results.tex}
{\footnotesize{}The outer loop tolerance is set to 1E-8.}{\footnotesize\par}
\end{table}
\par\end{center}

One remarkable point is that the total number of objective evaluations differs across different inner-loop algorithms. Though the results might seem unnatural considering the tight inner-loop tolerance level 1E-14 and mostly the same estimated parameters, they are not unusual taking account of the effect of inner-loop error on the outer-loop optimization algorithm. In the Monte Carlo simulation above, I apply the L-BFGS-B optimization algorithm as the outer-loop optimization method, which is the default in PyBLP and classified as one of the quasi-Newton methods. Intuitively, when we apply quasi-Newton-type outer-loop optimization algorithms, outer-loop step sizes might be sensitive to inner-loop numerical errors, even when search directions are not so biased. This implies that the convergence speed of the outer-loop optimization algorithm might be sensitive to inner-loop numerical errors.\footnote{\citet{dube2012improving} and \citet{lee2016revisiting} formally discussed the effect of inner-loop numerical errors on estimated parameters. In contrast, our current focus is on the convergence speed of the outer loop. Note that the argument holds even when we analytically compute the derivative of the GMM objective function concerning the candidate parameters. When we compute the derivative by numerical derivation, the problem might worsen.} 

Here, by developing a simple analytical framework, we discuss why the convergence speed of the outer-loop optimization algorithm might be sensitive to inner-loop numerical errors. To clarify the point, we assume there is only one nonlinear parameter $\theta$. Let $Q(\theta)$ be the GMM objective function given $\theta$, and let $g=\frac{\partial Q}{\partial\theta}$ be the true derivative of the GMM objective function with respect to $\theta$. When we apply a nested-fixed point algorithm, we incur inner-loop numerical errors, and we obtain another GMM objective function $\widetilde{Q}(\theta)$, which might be different from $Q$. We further assume $\widetilde{Q}(\theta)$ is differentiable and let $\widetilde{g}$ be the derivative.\footnote{As discussed in \citet{dube2012improving}, in general there is no guarantee that $\widetilde{Q}(\theta)$ is differentiable. Nevertheless, essential ideas would not be lost with this assumption.}

When we apply the quasi-Newton type optimization algorithms, parameter $\theta$ is iteratively updated until convergence in the following way: $\theta^{(n+1)}=\theta^{(n)}-\lambda_{n}B_{n}^{-1}\widetilde{g}\left(\theta^{(n)}\right)$. Here, $B_{n}$ is the approximation of $\nabla_{\theta}\widetilde{g}$. Generally it is computed using the past values of $\theta^{(n)}$ and $\widetilde{g}\left(\theta^{(n)}\right)$. $B_{n}^{-1}\widetilde{g}\left(\theta^{(n)}\right)$ denotes the search direction, and $\lambda_{n}\in(0,1]$ denotes the step size. The value of the scalar $\lambda_{n}$ is chosen based on the values of $\widetilde{g}\left(\theta^{(n)}\right)$ and $B_{n}$ in the line search process. 

This updating equation implies the convergence speed of $\theta$ largely depends on the values of $\nabla_{\theta}\widetilde{g}$. If $\nabla_{\theta}\widetilde{g}$ is largely affected by the inner-loop numerical errors, inner-loop numerical errors might largely affect the convergence speed. The following simple proposition implies $\nabla_{\theta}\widetilde{g}$ might be largely biased even when $\widetilde{g}$ is not: 
\begin{prop}
\label{prop:error_g_diff}Suppose $\sup_{x\in\mathbb{R}}\left|g(x)-\widetilde{g}(x)\right|\leq\epsilon$ and $\sup_{x\in\mathbb{R}}\left|g(x)\right|\leq C,\ \sup_{x\in\mathbb{R}}\left|\widetilde{g}(x)\right|<C$. Then, $\sup_{x\in\mathbb{R}}\left|g^{\prime}(x)-\widetilde{g}^{\prime}(x)\right|\leq4\sqrt{\epsilon C}$ holds.
\end{prop}
\begin{proof}
By Taylor's theorem, for all $\Delta>0,$ there exists $\overline{x}\in(0,\Delta)$ such that:

\begin{eqnarray*}
g(x+\Delta) & = & g(x)+g^{\prime}(x)\left((x+\Delta)-x\right)+g^{\prime\prime}(\overline{x})\cdot\left((x+\Delta)-x\right)^{2}\\
 & = & g(x)+g^{\prime}(x)\Delta+g^{\prime\prime}(\overline{x})\cdot\Delta^{2}.
\end{eqnarray*}

Similarly, for all $\Delta>0,$ there exists $\widetilde{x}\in(0,\Delta)$ such that:

\begin{eqnarray*}
\widetilde{g}(x+\Delta) & = & \widetilde{g}(x)+\widetilde{g}^{\prime}(x)\Delta+\widetilde{g}^{\prime\prime}(\widetilde{x})\cdot\Delta^{2}.
\end{eqnarray*}

Then, 

\begin{eqnarray*}
\left|g^{\prime}(x)-\widetilde{g}^{\prime}(x)\right| & = & \left|\frac{g(x+\Delta)-g(x)-g^{\prime\prime}(\overline{x})\cdot\Delta^{2}}{\Delta}-\frac{\widetilde{g}(x+\Delta)-\widetilde{g}(x)-\widetilde{g}^{\prime\prime}(\overline{x})\cdot\Delta^{2}}{\Delta}\right|\\
 & \leq & \left|\frac{g(x+\Delta)-\widetilde{g}(x+\Delta)}{\Delta}\right|+\left|\frac{g(x)-\widetilde{g}(x)}{\Delta}\right|+\left|\left(\widetilde{g}^{\prime\prime}(\overline{x})-g^{\prime\prime}(\widetilde{x})\right)\cdot\Delta\right|\\
 & \leq & \frac{2\epsilon}{\Delta}+2C\Delta
\end{eqnarray*}

By the inequality of arithmetic and geometric means, $\frac{2\epsilon}{\Delta}+2C\Delta\geq2\sqrt{\frac{2\epsilon}{\Delta}\cdot2C\Delta}=4\sqrt{\epsilon C}$ holds, and the equality holds only when $\frac{2\epsilon}{\Delta}=2C\Delta$, i.e. $\Delta=\sqrt{\frac{\epsilon}{C}}$.
\end{proof}
The proposition implies the error of $\left|g^{\prime}(x)-\widetilde{g}^{\prime}(x)\right|$ is of order $\epsilon^{\frac{1}{2}}$ when the error of $\left|g(x)-\widetilde{g}(x)\right|$ is of order $\epsilon$. For instance, even when we set the inner-loop tolerance level to $\epsilon$=1E-14 and the numerical error of $\widetilde{g}$ relative to $g$ is of order 1E-14, the bias in $\widetilde{g}^{\prime}$ might be of order $\sqrt{\epsilon}=$1E-7. In PyBLP, the default outer-loop tolerance level of the L-BFGS-B algorithm is 1E-8, and the numerical error in $\widetilde{g}^{\prime}$ is not negligible when we evaluate the convergence. 

\subsection{Details of the data-generating process and algorithms in the numerical experiments}

\subsubsection{Static BLP model\label{subsec:DGP-static-BLP}}

As in \citet{dube2012improving} and \citet{lee2015computationally}, let $X_{jt}=\{1,x_{j1},x_{j2},x_{j3},p_{jt}\}$, and assume $\left\{ x_{j1},x_{j2},x_{j3}\right\} $ follows:

\[
\left(\begin{array}{c}
x_{j1}\\
x_{j2}\\
x_{j3}
\end{array}\right)\sim N\left(\left(\begin{array}{c}
0\\
0\\
0
\end{array}\right),\left(\begin{array}{ccc}
1 & -0.8 & 0.3\\
-0.8 & 1 & 0.3\\
0.3 & 0.3 & 1
\end{array}\right)\right).
\]

$\xi_{jt}$ follows $N(0,1)$, and $p_{jt}$ is generated by $p_{jt}=3+\xi_{jt}\cdot1.5+u_{jt}+\sum_{m=1}^{3}x_{jm}$, where $u_{jt}\sim U[0,5]$. Finally, let $\beta_{i}=\{\theta_{i}^{0},\theta_{i}^{1},\theta_{i}^{2},\theta_{i}^{3},\theta_{i}^{p}\}$, each distributed independently normal with $E[\theta_{i}]=\{0,1.5,1.5,0.5,-3\}$ and $Var[\theta_{i}]=\{0.5^{2},0.5^{2},0.5^{2},0.5^{2},0.2^{2}\}$ in the baseline setting. 

\subsubsection{Dynamic BLP model\label{subsec:DGP-dynamic-BLP}}

As in \citet{Sun2019}, let $X_{jt}=[1,\chi_{jt},-p_{jt}]$. Product characteristics $\chi_{jt}$ and $\xi_{jt}$ are generated as $\chi_{jt}\equiv\left[\begin{array}{c}
\chi_{1jt}\\
\chi_{2jt}\\
\chi_{3jt}
\end{array}\right]\sim N\left(\left[\begin{array}{c}
0\\
0\\
0
\end{array}\right],\left[\begin{array}{ccc}
0.5^{2}\\
 & 0.5^{2}\\
 &  & 0.5^{2}
\end{array}\right]\right)$ and $\xi_{jt}\sim N(0,1)$. The price $p_{jt}$ is generated from $p_{jt}=\gamma_{0}+\gamma_{x}^{\prime}\chi_{jt}+\gamma_{z}z_{jt}+\gamma_{w}w_{jt}+\gamma_{\xi}\xi_{jt}-\gamma_{p}^{\prime}\sum_{k\neq j}\chi_{kt}+u_{jt}$, where $z_{jt}=\rho_{0}+\rho_{1}z_{jt-1}+\eta_{jt}$, $\eta_{jt}\sim N(0,0.1^{2})$, $w_{jt}\sim N(0,1)$, $u_{jt}\sim N(0,0.01^{2})$, $[\gamma_{0},\gamma_{X1},\gamma_{X2},\gamma_{X3},\gamma_{z},\gamma_{w},\gamma_{\xi}]=[1,0.2,0.2,0.1,1,0.2,0.7],\gamma_{p}=[0.1,0.1,0.1]$, $z_{j0}=8$, $[\rho_{0},\rho_{z}]=[0.1,0.95]$. 

Regarding the demand parameters $\theta_{i}$, let $\theta_{i}=[\theta_{i}^{X0},\theta_{i}^{X1},\theta_{i}^{X2},\theta_{i}^{X3},\theta_{i}^{p}]$ and $\left[\begin{array}{c}
\theta_{i}^{X1}\\
\theta_{i}^{X2}\\
\theta_{i}^{p}
\end{array}\right]\sim N\left(\left[\begin{array}{c}
1\\
1\\
2
\end{array}\right],\left[\begin{array}{ccc}
0.5^{2}\\
 & 0.5^{2}\\
 &  & 0.25^{2}
\end{array}\right]\right)$, $\theta_{i}^{X0}=6,\text{\ensuremath{\theta_{i}^{X3}=0.5}}$.

\subsubsection*{Details of the dynamic BLP algorithm under inclusive value sufficiency}

Algorithm \ref{alg:durable-IVS} shows the steps to solve for $\delta$ under inclusive value sufficiency.

\begin{algorithm}[H]
Take grid points $\omega_{h}^{(grid)}(h=1,\cdots,N_{grid})$. Set initial values of $V_{i}^{(0)}\left(\omega_{it}^{(data)}\right)\ i\in\mathcal{I},t=1,\cdots,T$ and $V_{i}^{(0)}\left(\omega_{h}^{(grid)}\right)\ i\in\mathcal{I},h=1,\cdots,N_{grid}$. Iterate the following $(n=0,1,2,\cdots)$:
\begin{enumerate}
\item For $t=1:T$,
\begin{enumerate}
\item Compute $\delta_{jt}^{(n)}=\iota_{V\rightarrow\delta,jt}^{\gamma}\left(V^{(n)}\right)=\log\left(S_{jt}^{(data)}\right)-\log\left(\sum_{i}w_{i}Pr0_{it}\cdot\frac{\exp\left(\mu_{ijt}\right)}{\exp\left(V_{i}^{(n)}(\omega_{it}^{(data)})\right)}\right)$
\item Compute $s_{ijt}^{(ccp)}=\frac{\exp\left(\delta_{jt}^{(n)}+\mu_{ijt}\right)}{\exp\left(V_{i}^{(n)}(\omega_{it}^{(data)})\right)}$ for $i\in\mathcal{I},j\in\mathcal{J}$
\item Compute $s_{i0t}^{(ccp)}=1-\sum_{j\in\mathcal{J}_{t}}s_{ijt}^{(ccp)}$
\item Update $Pr0_{it+1}=Pr0_{it}\cdot s_{i0t}^{(ccp)}$
\end{enumerate}
\item Compute $E_{t}\left[V_{i}(\omega_{it+1})|\omega_{it}\right]$:
\begin{enumerate}
\item Compute $\omega_{it}^{(data)(n)}=\log\left(1+\sum_{j\in\mathcal{J}_{t}}\exp\left(\delta_{jt}^{(n)}+\mu_{ijt}\right)\right)\ (i\in\mathcal{I},t=1,\cdots,T)$
\item Estimate the state transition probabilities $Pr(\omega_{it+1}^{(data)}|\omega_{it}^{(data)})$. We assume AR(1) process $\omega_{it+1}^{(data)}=\theta_{i0}+\theta_{i1}\omega_{it}^{(data)}+u_{it}\ (u_{it}\sim N(0,\sigma_{i}^{2}))$, and estimate the parameters $\theta_{i0},\theta_{i1},\sigma_{i}$.
\item Compute $E_{t}\left[V_{i}^{(n)}(\omega_{it+1})|\omega_{it}^{(data)(n)}\right]=\int V_{i}^{(n)}(\omega_{it+1})Pr(\omega_{it+1}|\omega_{it}^{(data)(n)})d\omega_{it+1}$. and $E_{t}\left[V_{i}^{(n)}(\omega_{it+1})|\omega_{it}^{(grid)}\right]$. Here, the values of $V_{i}^{(n)}(\omega_{it+1})$ are interpolated using the values of $V_{i}^{(n)}(\omega_{it}^{(grid)})$.
\end{enumerate}
\item Update $V\left(\omega^{(data)}\right)$ by:{\footnotesize{}
\begin{eqnarray*}
V_{i}^{(n+1)}\left(\omega_{i}^{(data)(n)}\right) & = & \Psi_{V\delta\rightarrow V,i}^{\gamma}\left(V^{(n)},\delta^{(n)}\right)\\
 & = & \log\left(\exp\left(\beta E_{t}\left[V_{i}^{(n)}(\omega_{it+1})|\omega_{it}^{(data)(n)}\right]\right)+\exp\left(\omega_{it}^{(data)(n)}\right)\cdot\left(\frac{s_{0t}(V^{(n)},\delta^{(n)})}{S_{0t}^{(data)}}\right)^{\gamma}\right)
\end{eqnarray*}
}{\footnotesize\par}

Here, $s_{0t}$ is computed by $s_{0t}(V^{(n)},\delta^{(n)})=\frac{\exp\left(\beta E_{t}\left[V_{i}^{(n)}(\omega_{it+1})|\omega_{it}^{(data)(n)}\right]\right)}{\exp\left(V_{i}^{(n)}(\omega_{it}^{(data)(n)})\right)}$.
\item Update $V\left(\omega^{(grid)}\right)$ by $V_{i}^{(n+1)}\left(\omega_{i}^{(grid)}\right)=\log\left(\exp\left(\beta E_{t}\left[V_{i}^{(n)}(\omega_{it+1})|\omega_{i}^{(grid)}\right]\right)+\exp\left(\omega_{i}^{(grid)}\right)\right)$
\item Exit the iteration if $\left\Vert V^{(n+1)}-V^{(n)}\right\Vert <\epsilon$
\end{enumerate}
\caption{Inner-loop Algorithm of dynamic BLP (Perfectly durable goods; Inclusive value sufficiency)\label{alg:durable-IVS}}
\end{algorithm}

In this study, I introduce 10 Chebyshev polynomial grid points in the range {[}-40,30{]}.\footnote{In the numerical experiments in the Supplemental Appendix \ref{subsec:Numerical-results-dynamic_BLP_smaller_CCP}, we take gridpoints in the range {[}-40,50{]} to stabilize convergence.} Also, the integral $\int V_{i}^{(n)}(\omega_{it+1})Pr(\omega_{it+1}|\omega_{it}^{(data)(n)})d\omega_{it+1}$ is computed by introducing Gauss-Hermite quadrature with order 5. 

\subsection{Numerical results of the dynamic BLP model with smaller outside option CCPs\label{subsec:Numerical-results-dynamic_BLP_smaller_CCP}}

This section shows additional numerical results of the dynamic BLP model by setting the value of $\theta_{i}^{X0}$ to 30, in contrast to the baseline setting $\theta_{i}^{X0}=6$. 
\begin{center}
\begin{table}[H]
\caption{Results of the Dynamic BLP Monte Carlo simulation (Perfectly durable goods; Perfect foresight; Smaller outside option CCPs ($\theta_{i}^{X0}=30$))\label{tab:dynamic_BLP_Monte_Carlo_Perfect-foresight-smaller-outside-CCPs}}

\scalebox{0.8}{
\begin{centering}
\begin{tabular}{cccccccccccc}
\hline 
 & \multirow{2}{*}{{\scriptsize{}$J$}} & \multicolumn{6}{c}{{\scriptsize{}Func. Evals.}} & {\scriptsize{}Mean} & {\scriptsize{}Conv.} & {\footnotesize{}Mean} & {\scriptsize{}$DIST<\epsilon_{tol}$}\tabularnewline
\cline{3-8} \cline{4-8} \cline{5-8} \cline{6-8} \cline{7-8} \cline{8-8} 
 &  & {\scriptsize{}Mean} & {\scriptsize{}Min.} & {\scriptsize{}25th} & {\scriptsize{}Median.} & {\scriptsize{}75th} & {\scriptsize{}Max.} & {\scriptsize{}CPU time (s)} & {\scriptsize{}(\%)} & {\footnotesize{}$\log_{10}\left(DIST\right)$} & {\scriptsize{}(\%)}\tabularnewline
\hline 
{\scriptsize{}$V$-(0)} & {\scriptsize{}25} & {\scriptsize{}3000} & {\scriptsize{}3000} & {\scriptsize{}3000} & {\scriptsize{}3000} & {\scriptsize{}3000} & {\scriptsize{}3000} & {\scriptsize{}4.1597} & {\scriptsize{}0} & {\scriptsize{}-12.5} & {\scriptsize{}100}\tabularnewline
{\scriptsize{}$V$-(0) + Anderson} & {\scriptsize{}25} & {\scriptsize{}790.4} & {\scriptsize{}619} & {\scriptsize{}652.5} & {\scriptsize{}810} & {\scriptsize{}873} & {\scriptsize{}1003} & {\scriptsize{}4.5921} & {\scriptsize{}100} & {\scriptsize{}-14} & {\scriptsize{}100}\tabularnewline
{\scriptsize{}$V$-(0) + Spectral} & {\scriptsize{}25} & {\scriptsize{}825.4} & {\scriptsize{}645} & {\scriptsize{}745} & {\scriptsize{}797} & {\scriptsize{}896} & {\scriptsize{}1208} & {\scriptsize{}1.18405} & {\scriptsize{}100} & {\scriptsize{}-14.8} & {\scriptsize{}100}\tabularnewline
{\scriptsize{}$V$-(0) + SQUAREM} & {\scriptsize{}25} & {\scriptsize{}895.2} & {\scriptsize{}651} & {\scriptsize{}806.5} & {\scriptsize{}916.5} & {\scriptsize{}971.5} & {\scriptsize{}1148} & {\scriptsize{}1.26145} & {\scriptsize{}100} & {\scriptsize{}-14.3} & {\scriptsize{}100}\tabularnewline
{\scriptsize{}$V$-(1)} & {\scriptsize{}25} & {\scriptsize{}3000} & {\scriptsize{}3000} & {\scriptsize{}3000} & {\scriptsize{}3000} & {\scriptsize{}3000} & {\scriptsize{}3000} & {\scriptsize{}4.2312} & {\scriptsize{}0} & {\scriptsize{}-12.5} & {\scriptsize{}100}\tabularnewline
{\scriptsize{}$V$-(1) + Anderson} & {\scriptsize{}25} & {\scriptsize{}708.15} & {\scriptsize{}510} & {\scriptsize{}582.5} & {\scriptsize{}682.5} & {\scriptsize{}787.5} & {\scriptsize{}1136} & {\scriptsize{}3.85515} & {\scriptsize{}100} & {\scriptsize{}-14.1} & {\scriptsize{}100}\tabularnewline
{\scriptsize{}$V$-(1) + Spectral} & {\scriptsize{}25} & {\scriptsize{}757.7} & {\scriptsize{}623} & {\scriptsize{}668.5} & {\scriptsize{}747} & {\scriptsize{}817.5} & {\scriptsize{}971} & {\scriptsize{}1.1115} & {\scriptsize{}100} & {\scriptsize{}-14.8} & {\scriptsize{}100}\tabularnewline
{\scriptsize{}$V$-(1) + SQUAREM} & {\scriptsize{}25} & {\scriptsize{}739.55} & {\scriptsize{}76} & {\scriptsize{}704.5} & {\scriptsize{}837} & {\scriptsize{}915.5} & {\scriptsize{}997} & {\scriptsize{}1.0655} & {\scriptsize{}85} & {\scriptsize{}-12.7} & {\scriptsize{}85}\tabularnewline
\hline 
\end{tabular}
\par\end{centering}
}

{\footnotesize{}\input{notes/note_dynamic_BLP.tex}}{\footnotesize\par}

{\footnotesize{}The minimum and median outside option CCPs are 0.104 and 0.923 respectively.}{\footnotesize\par}
\end{table}
\par\end{center}

\begin{center}
\begin{table}[H]
\caption{Results of the Dynamic BLP Monte Carlo simulation (Perfectly durable goods; Inclusive value sufficiency; Smaller outside option CCPs)\label{tab:dynamic_BLP_Monte_Carlo_IVS-smaller-outside-CCPs}}

\scalebox{0.8}{
\begin{centering}
\begin{tabular}{cccccccccccc}
\hline 
 & \multirow{2}{*}{{\scriptsize{}$J$}} & \multicolumn{6}{c}{{\scriptsize{}Func. Evals.}} & {\scriptsize{}Mean} & {\scriptsize{}Conv.} & {\footnotesize{}Mean} & {\scriptsize{}$DIST<\epsilon_{tol}$}\tabularnewline
\cline{3-8} \cline{4-8} \cline{5-8} \cline{6-8} \cline{7-8} \cline{8-8} 
 &  & {\scriptsize{}Mean} & {\scriptsize{}Min.} & {\scriptsize{}25th} & {\scriptsize{}Median.} & {\scriptsize{}75th} & {\scriptsize{}Max.} & {\scriptsize{}CPU time (s)} & {\scriptsize{}(\%)} & {\footnotesize{}$\log_{10}\left(DIST\right)$} & {\scriptsize{}(\%)}\tabularnewline
\hline 
{\scriptsize{}$V$-(0)} & {\scriptsize{}25} & {\scriptsize{}2103.85} & {\scriptsize{}1563} & {\scriptsize{}1829.5} & {\scriptsize{}2060.5} & {\scriptsize{}2329.5} & {\scriptsize{}3000} & {\scriptsize{}5.9077} & {\scriptsize{}95} & {\scriptsize{}-13.8} & {\scriptsize{}100}\tabularnewline
{\scriptsize{}$V$-(0) + Anderson} & {\scriptsize{}25} & {\scriptsize{}271.45} & {\scriptsize{}198} & {\scriptsize{}232.5} & {\scriptsize{}265} & {\scriptsize{}306} & {\scriptsize{}355} & {\scriptsize{}1.241} & {\scriptsize{}100} & {\scriptsize{}-14.7} & {\scriptsize{}100}\tabularnewline
{\scriptsize{}$V$-(0) + Spectral} & {\scriptsize{}25} & {\scriptsize{}627.95} & {\scriptsize{}523} & {\scriptsize{}552.5} & {\scriptsize{}615} & {\scriptsize{}683.5} & {\scriptsize{}868} & {\scriptsize{}1.7775} & {\scriptsize{}100} & {\scriptsize{}-14.3} & {\scriptsize{}100}\tabularnewline
{\scriptsize{}$V$-(0) + SQUAREM} & {\scriptsize{}25} & {\scriptsize{}629.4} & {\scriptsize{}512} & {\scriptsize{}566} & {\scriptsize{}622} & {\scriptsize{}696} & {\scriptsize{}828} & {\scriptsize{}1.75075} & {\scriptsize{}100} & {\scriptsize{}-14.2} & {\scriptsize{}100}\tabularnewline
{\scriptsize{}$V$-(1)} & {\scriptsize{}25} & {\scriptsize{}2029.5} & {\scriptsize{}1445} & {\scriptsize{}1726} & {\scriptsize{}1834} & {\scriptsize{}2060} & {\scriptsize{}3000} & {\scriptsize{}5.67115} & {\scriptsize{}80} & {\scriptsize{}-12} & {\scriptsize{}80}\tabularnewline
{\scriptsize{}$V$-(1) + Anderson} & {\scriptsize{}25} & {\scriptsize{}254.95} & {\scriptsize{}205} & {\scriptsize{}228.5} & {\scriptsize{}245.5} & {\scriptsize{}279.5} & {\scriptsize{}346} & {\scriptsize{}1.13825} & {\scriptsize{}100} & {\scriptsize{}-15.1} & {\scriptsize{}100}\tabularnewline
{\scriptsize{}$V$-(1) + Spectral} & {\scriptsize{}25} & {\scriptsize{}582.85} & {\scriptsize{}68} & {\scriptsize{}529.5} & {\scriptsize{}564} & {\scriptsize{}677.5} & {\scriptsize{}855} & {\scriptsize{}1.6541} & {\scriptsize{}95} & {\scriptsize{}-14.1} & {\scriptsize{}95}\tabularnewline
{\scriptsize{}$V$-(1) + SQUAREM} & {\scriptsize{}25} & {\scriptsize{}572.45} & {\scriptsize{}24} & {\scriptsize{}475.5} & {\scriptsize{}599.5} & {\scriptsize{}653.5} & {\scriptsize{}868} & {\scriptsize{}1.6135} & {\scriptsize{}95} & {\scriptsize{}-14.2} & {\scriptsize{}95}\tabularnewline
\hline 
\end{tabular}
\par\end{centering}
}

{\footnotesize{}\input{notes/note_dynamic_BLP.tex}}{\footnotesize\par}

{\footnotesize{}The minimum and median outside option CCPs are 0.320 and 0.927 respectively.}{\footnotesize\par}
\end{table}
\par\end{center}

\section{Dynamic discrete choice models with unobserved payoffs and unobserved heterogeneity\label{sec:DDC}}

\input{DDC.tex}

\bibliographystyle{apalike}
\bibliography{literature}

\end{document}

%% file: introduction_QE.tex
Demand estimation is the basis for many economic analyses. The estimation method based on the random coefficient logit model proposed by \citet{berry1995automobile} (henceforth BLP) has been widely used in the literature. Although \citet{berry1995automobile} initially considered a static demand model, the model and estimation procedure have been extended to dynamic demand models, such as durable goods and goods with switching costs (e.g., \citealp{gowrisankaran2012dynamics}; \citealp{shcherbakov2016measuring}). These models are known as dynamic BLP. In these models, the products' mean utilities $\delta$ are numerically solved in the inner loop, and econometricians find the parameter values that minimize the GMM objective in the outer loop. 

Nevertheless, significant challenge of this method is the computational time, as fixed-point iterations in the inner loop can be burdensome, especially for complex demand model estimations or large datasets. This study explores computationally efficient inner-loop algorithms for estimating both static and dynamic BLP models. While computationally efficient inner-loop algorithms for static BLP models have been extensively studied in the literature, including \citet{kalouptsidi2012market}, \citet{reynaerts2012enhencing}, \citet{conlon2020best}, and \citet{pal2023comparing}, there remains considerable potential for improvement in static BLP models, and even more so in dynamic BLP models. 

Regarding static BLP models, I introduce a new fixed-point iteration mapping $\delta_{j}^{(n+1)}=\Phi_{j}^{\delta,\gamma=1}(\delta_{j}^{(n)})\equiv \delta_{j}^{(n)}+\left(\log\left(S_{j}^{(data)}\right)-\log\left(s_{j}(\delta^{(n)})\right)\right)-\left(\log\left(S_{0}^{(data)}\right)-\log\left(s_{0}(\delta^{(n)})\right)\right)$. This slightly modifies the traditional BLP contraction mapping $\delta_{j}^{(n+1)}=\delta_{j}^{(n)}+\left(\log\left(S_{j}^{(data)}\right)-\log\left(s_{j}(\delta^{(n)})\right)\right)$.\footnote{$S_{j}^{(data)}$ denotes product $j$'s observed market share, and $s_{j}(\delta)$ denotes the market share predicted by the structural model.} The difference with the traditional BLP contraction mapping is the term $\log\left(S_{0}^{(data)}\right)-\log\left(s_{0}(\delta^{(n)})\right)$, which is straightforward to implement in any programming language. Interestingly, when consumer heterogeneity is absent, the iteration converges after one iteration, regardless of initial values $\delta^{(0)}$. In such cases, $\delta_{j}=\log(S_{j}^{(data)})-\log(S_{0}^{(data)})$ holds, as shown by \citet{berry1994estimating}. For the mapping $\Phi^{\delta,\gamma=1}$, $\Phi_{j}^{\delta,\gamma=1}\left(\delta\right)=\log\left(S_{j}^{(data)}\right)-\log\left(S_{0}^{(data)}\right)$ holds, and the output of $\Phi^{\delta,\gamma=1}$ is equal to the true $\delta$ for any input $\delta$. When the consumer heterogeneity exists,
the iteration may not converge after one iteration. However, it still inherits the good convergence property. This study finds that the new mapping converges and reduces the number of iterations in the Monte Carlo simulation settings experimented in the previous studies (\citealp{dube2012improving}; \citealp{lee2015computationally}) and datasets from \citet{nevo2001measuring} and Berry et al. \citeyearpar{berry1995automobile,berry1999voluntary}. The speed-up is especially prominent when the consumer heterogeneity is relatively small, as theoretically shown in Appendix \ref{subsec:Convergence-properties-mappings}.\footnote{Though there is no guarantee that the mapping is a contraction, we can guarantee the global convergence of the iterations using the new mapping by adding a few lines in the programming code, by using the fact that the traditional BLP contraction mapping is a contraction. For details, see Appendix \ref{subsec:Global-convergence}. If practitioners are conservative concerning the convergence, the procedure is worth considering.}

This study also compares various fixed-point iterations acceleration methods (Anderson acceleration, Spectral, SQUAREM), and finds that Anderson acceleration, which has been understudied in the literature\footnote{\citet{duch2023evaluating} is the exception. As discussed in the online appendix of their paper, they applied the stabilized version of the Anderson acceleration method proposed by \citet{zhang2020globally} to estimate a random coefficient nested logit model with a large nest parameter.}\footnote{\citet{conlon2020best} pointed out in footnote 81 of their paper that ``Anderson acceleration'' using ``anderson'' function in Python Scipy package was too slow and unreliable to be worth considering. The current study verify it by using PyBLP. However, the ``Anderson acceleration'' method in Python Scipy package is aimed at solving general nonlinear equations, and the algorithm is not the same as the ``Anderson acceleration'' for fixed point iterations discussed in the current paper.}, performs better than the others, and further reduces both the number of iterations and computation time. The strategy remains effective even when the size of consumer heterogeneity is large. As the new mapping can be easily coded by adding a few lines to the traditional BLP contraction mapping, the current study suggests coding the new mapping first, and then considering combining the acceleration of fixed point iterations, especially Anderson acceleration, if needed. Anderson acceleration method is intuitive and relatively easy to code.

Regarding dynamic BLP models, I propose analytically representing the mean product utility $\delta$ as a function of value function $V$, solveing for $V$ by applying a mapping of $V$, and recovering $\delta$ using the analytical formula and $V$ values. Previous studies (e.g., \citealp{gowrisankaran2012dynamics}) applied an inner loop algorithm essentially consisted of solving for two types of variables $\delta$ and $V$. In the current algorithm, we only have to solve for only one type of variable $V$, thus reducing inner loop iterations. Combining fixed-point iteration acceleration methods, especially Anderson acceleration further reduces the number of iterations.

Fixed-point iteration acceleration methods are useful not only for the inner-loop algorithms of static/dynamic BLP estimations, but also for efficiently solving various fixed-point problems, including dynamic models with value functions and supply-side pricing equilibrium. Built-in packages of the acceleration methods, including Anderson acceleration, Spectral, and SQUAREM, are available in some programming languages. Even when unavailable, they are easy to code because of their simple formulas. Hence, being familializing oneself with  accustomed to using such built-in packages or developing custom coding to implement the acceleration methods to meet practitioners' needs would be useful. Note that when we apply the spectral or SQUAREM algorithm, we should be careful about the choice of step sizes $\alpha$, though not paid much attention to in the previous economics literature. When we do not use appropriate step sizes, the algorithms may cause divergence. This is briefly discussed in Section \ref{sec:Acceleration-methods}.


This study contributes to the literature by presenting key insights aimed at reducing inner-loop iterations in static/dynamic BLP estimations: (1). New mapping $\delta_{j}^{(n+1)}=\delta_{j}^{(n)}+\left(\log\left(S_{j}^{(data)}\right)-\log\left(s_{j}(\delta^{(n)})\right)\right)-\left(\log\left(S_{0}^{(data)}\right)-\log\left(s_{0}(\delta^{(n)})\right)\right)$; (2). Analytically represent mean product utilities $\delta$ as a function of value functions $V$ and solve for $V$ (for dynamic BLP models); (3). Combine the acceleration method of the fixed-point iterations, especially the Anderson acceleration. These methods are independent, and relatively easy to implement. Practitioners can selectively implement some when they face computational challenges.\footnote{\citet{fukasawa2024lightbulb} applies some of these ideas to the dynamic demand estimation of the light bulb market by specifying a durable goods model with forward-looking consumers and replacement demand. The computation time is reduced by more than 10 times with these strategies. Without them, the expected estimation time would be more than 10 hours. }

The rest of this article is organized as follows. Section \ref{sec:Literature} examines the relationships with the previous studies. Section \ref{sec:Static-BLP-model} discusses inner-loop algorithms for static BLP. Section \ref{sec:Dynamic-BLP-model} explores  inner-loop algorithms for dynamic BLP models. Readers  focusing on static BLP models can skip this section. Section \ref{sec:Acceleration-methods} discusses the acceleration methods of fixed point iterations, including Anderson acceleration, Spectral, and  SQUAREM. Note that the discussion in this section applies to any fixed-point iteration. Section \ref{sec:Numerical-Experiments} presents the results of numerical experiments on static/dynamic BLP models. Finally, Section \ref{sec:Conclusion} presents the conclusions.

Appendix \ref{sec:Additional-results} shows additional results and discussions. Appendix \ref{sec:Proof} contains all the proofs of the mathematical statements. The Supplemental Appendix shows further results and discussions. In the Supplemental Appendix, we discuss that the idea of the proposed inner-loop algorithm for estimating dynamic BLP models can be used in estimating dynamic discrete choice models with unobserved payoffs, which have been considered in \citet{kalouptsidi2020linear} and others.

%% file: literature_review_QE.tex
First, this study relates to the literature on computationally efficient estimation of static/dynamic BLP models.

\citet{conlon2020best} developed the PyBLP package written in Python for implementing static BLP models' estimations and simulation. Although the package is highly convenient to use and state-of-the-art, incorporating the knowledge from existing literature, it is not suitable for BLP-type problems that PyBLP cannot address.\footnote{For instance, we cannot directly use PyBLP to estimate static limited consideration set models (e.g., \citealp{goeree2008limited}) or dynamic BLP models.} In such cases, we should write our own code, and the techniques discussed in the current study should be considered.

\citet{kalouptsidi2012market} proposed inner-loop estimation algorithms of static BLP models that achieve faster convergence in the setting with a small number of discrete consumer types. \citet{doi2022simple} proposed an estimation method for static BLP models with discrete consumer types that eliminates the need for  computationally costly fixed-point iterations, given the total sales data availability for each consumer type. The main concept of these studies is that mean utilities or unobserved product characteristics can be analytically represented as a function of value functions specific to consumer types, given parameters and market share data. 

The idea is beneficial, especially in the context of dynamic BLP models. While the direct application of \citet{kalouptsidi2012market}'s algorithm to static BLP models do not work well when the number of consumer types is large. Nevertheless, I show that it works well by simplifying and slightly modifying the original algorithm. The new algorithm corresponds to the mapping of value function $V$ discussed in Section \ref{subsec:Mappings-on-delta-static-BLP}. We discuss these issues in detail in the Supplemental Appendix. 

Although we focus on improving the inner-loop algorithms for estimating static/dynamic BLP models through Nested fixed-point (NFXP) approach, several alternative estimation procedures have been proposed so far.\footnote{Besides the studies mentioned above, \citet{bonnet2022yogurts} proposed an inner-loop estimation algorithm of static demand models, using the idea of two-sided matching. The algorithm is applicable to static demand models other than static BLP models, including the pure characteristics model considered in \citet{berry2007pure}. \citet{salanie2022fast} proposed an estimator of static BLP models that does not require solving fixed-point problems and is approximately correct. The estimator can be used as the initial values of parameters in the standard nested fixed point estimation. } \citet{dube2012improving} proposed the MPEC (Mathematical Programming with Equilibrium Constraint) approach for static and dynamic BLP models. \citet{lee2015computationally} proposed the ABLP (approximate BLP) method for static BLP models, which iterates the process inspired by Newton's method to estimate parameters. Regarding dynamic BLP, \citet{Sun2019} proposed a Bayesian-type algorithm. 

Although promising, recent studies on the MPEC have indicated that its performance is not necessarily as good as that of the NFXP approach. \citet{pal2023comparing} showed in static BLP models that the computational times using the MPEC are longer than those with the NFXP approach using the traditional BLP contraction mapping and the spectral/SQUAREM algorithm as the acceleration method. In addition, the frequency of reaching the global optimum of the optimization problem is lower for MPEC.\footnote{In general, there is no guarantee that the GMM objective function in the BLP estimation is convex, and we may reach a local minimum of the GMM objective function.} \citet{Sun2019} also showed that the MPEC performed worse than the NFXP. As mentioned in \cite{dube2012improving}, it is known that the MPEC is computationally costly when the Jacobian matrix of the constraint of the constrained optimization problem is dense, and the approach may not be always practical.

Furthermore, in general, there are two types of costs for any algorithm. The first one is ``thinking costs'', which are associated with thinking about the validity of an algorithm for practitioners' needs and coding the algorithm. The second one is ``computer costs'', which are associated with the computation time required to execute the algorithm on a computer.\footnote{\citet{jamshidian1997acceleration} discuss this point in the context of the acceleration of the EM algorithm.} The proposed methods in the current article reduce the computation costs with small additional thinking costs, as they require only minor changes to the estimation procedure based on the standard NFXP approach widely applied by practitioners. Unlike the MPEC method, we do not have to install new software or packages to solve constrained optimization problems.\footnote{In MPEC, sometimes we must specify a sparsity structure to reduce the memory requirements. Our current method, however, does not require this procedure.} Unlike the ABLP method\footnote{\citet{pal2023comparing} showed numerical results where the ABLP performs worse than the NFXP method using the traditional BLP contraction mapping and the spectral/SQUAREM algorithm as the acceleration method.}, the proposed methods do not require computing the analytical derivatives of functions\footnote{As discussed in \citet{miranda2004applied} in the context of the Newton's method for solving nonlinear equations, analysts might make coding errors in coding derivatives of functions, and such a procedure should be avoided if possible. The algorithms proposed in the current article do not require coding derivatives, and are attractive.}, and the proposed algorithms are not restricted to static BLP models. Moreover, unlike the Baysian-type methods, there is no need to introduce Baysian-type techniques, such as Markov Chain Monte Carlo (MCMC) methods.\footnote{The method proposed by \citet{Sun2019} relies on a fixed-point mapping whose convergence is not necessarily fast in the NFXP approach, as shown in the current paper. The insights presented in the current paper might be useful for improving the performance of \citet{Sun2019}'s method.} 

Finally, this study contributes to the literature discussing acceleration methods of fixed point iterations for BLP estimations. \citet{reynaerts2012enhencing}, \citet{conlon2020best}, and \citet{pal2023comparing} found that combining the spectral or SQUAREM algorithms accelerates the inner-loop convergence of the static BLP estimation. The current study finds that the Anderson acceleration method, which is not discussed in these studies, outperforms the spectral and SQUAREM algorithms. Note that the choice of the step sizes in the spectral and SQUAREM algorithms largely affects the algorithm's performance, though not paid much attention to in the literature, and the current study briefly discusses it.

%% file: choice_acceleration_methods.tex
As numerically shown in Section \ref{sec:Numerical-Experiments}, the Anderson acceleration outperforms both the spectral and SQUAREM algorithms in static and dynamic BLP applications, and the current study recommends using the Anderson acceleration as the acceleration method. Because even the spectral and SQUAREM algorithms work fairly well, they can be good alternatives in case the Anderson acceleration does not work well.

Regarding convergence, Anderson acceleration is mainly designed to accelerate the convergence of a contraction mapping, and the iteration may not converge when the mapping does not have contraction properties.\footnote{\citet{zhang2020globally} developed a globalization strategy for nonexpansive mappings, which is a slight extension of contraction mappings, in Anderson acceleration. However, the strategy may not work when there is no guarantee that the iteration is nonexpansive.} In contrast, regarding the spectral algorithm, there are many studies on globalization strategies to stabilize convergence (e.g., \citet{la2006spectral}; \citet{huang2017new}), and the spectral algorithm can be a good alternative for practitioners prioritizing convergence. Note that, in static BLP models, the global convergence can be ensured regardless of acceleration methods used, simplify by adding a few lines in the programming code, utilizing the fact that the BLP contraction mapping is a contraction. For details, see Appendix \ref{subsec:Global-convergence}.

%% file: notes/note_static_BLP_results.tex
{\footnotesize{}Notes. $DIST\equiv\left\Vert \log(S^{(data)})-\log(s)\right\Vert _{\infty}$, $\epsilon_{tol}=$1E-12.}{\footnotesize\par}

{\footnotesize{}The number of simulation draws is set to 1000.}{\footnotesize\par}

{\footnotesize{}The maximum number of function evaluations is set to 1000.}{\footnotesize\par}

%% file: notes/note_pyblp_results.tex
{\footnotesize{}Notes.}{\footnotesize\par}

{\footnotesize{}``Total obj eval'' denotes the total number of objective evaluations in the GMM estimation. }{\footnotesize\par}

{\footnotesize{}``Total feval'' denotes the total number of function evaluations in the GMM estimation. }{\footnotesize\par}

{\footnotesize{}``Objective'' denotes the GMM objective value. }{\footnotesize\par}

{\footnotesize{}``Mean feval'' denotes the mean number of function evaluations, defined by Total feval / (Number of markets $\times$ Total objective evalutions).}{\footnotesize\par}

%% file: notes/note_dynamic_BLP.tex
{\footnotesize{}Notes. $DIST\equiv\left\Vert \log(S^{(data)})-\log(s)\right\Vert _{\infty}$, $\epsilon_{tol}=$1E-12.}{\footnotesize\par}

{\footnotesize{}The maximum number of function evaluations is set to 3000.}{\footnotesize\par}

%% file: conclusion.tex
This study examined computationally fast inner-loop algorithms for estimating static and dynamic BLP models. To minimize the number of inner-loop iterations, the following ideas and insights are proposed: (1). New mapping $\delta^{(n+1)}=\delta^{(n)}+\left(\log\left(S_{j}^{(data)}\right)-\log\left(s_{j}(\delta^{(n)})\right)\right)-\left(\log\left(S_{0}^{(data)}\right)-\log\left(s_{0}(\delta^{(n)})\right)\right)$; (2). Analytically represent the mean product utilities $\delta$ as a function of value functions $V$ and solve for $V$ (for dynamic BLP); (3). Combine an acceleration method of fixed point iterations, especially Anderson acceleration. They are independent and easy to implement. These proposed methods would ease empirical analyses of markets with large datasets or complex demand models under static/dynamic BLP framework, facing problems of computational burden.

Although BLP demand models are considered, whether the ideas and insights can also be applied to other demand models, such as pure characteristics models in \citet{berry2007pure}, is an interesting topic for further research. 

Finally, although this study focuses on improving the convergence speed of the inner-loop algorithms, the mappings with fast convergence could potentially applied to other estimation procedures, such as the MPEC method (e.g., \citealp{dube2012improving}) and sequential estimation algorithms (e.g., \citealp{lee2015computationally}). Future research could examine these possibilities.

%% file: DDC.tex
The algorithms discussed in this study give insights into the estimation of Dynamic Discrete Choice (DDC) models with unobserved payoffs discussed in \citet{kalouptsidi2020linear},\footnote{They argued that DDC models with unobserved payoffs are attractive in settings where not all the market-level state variables are observable.} which can be regarded as a broader concept of dynamic demand or dynamic BLP models. This section briefly discusses it. Note that discussion in this section is closely related to the one in Section 5. Here, we allow for unobserved heterogeneity and do not restrict the focus on the models with finite dependence, unlike \citet{kalouptsidi2020linear}, but assume idiosyncratic utility shocks $\epsilon$ follow type-I extreme value distribution. We consider a nonstationary environment,\footnote{When we consider a stationary model, the transition of $\xi$ should also be specified and estimated. Or, we might be able to use the idea of inclusive value sufficiency applied in the dynamic demand literature (\citealp{hendel2006measuring}; \citealp{gowrisankaran2012dynamics}).} and we assume agents' decisions are observed until the terminal period, and their payoffs in the terminal period are correctly specified given parameters.

Discounted sum of utility of an agent at observed state $(x_{t},\Omega_{t})$ and persistent unobserved state $s$ when choosing alternative $j\in\mathcal{A}_{t}(x_{t})\subset\mathcal{J}_{t}$ at time $t$ is:

\begin{eqnarray*}
v_{jt}(x_{t},\Omega_{t},s)= & \overline{\pi}(x_{t},\Omega_{t},a_{t}=j,s,\theta)+\xi_{jt}(x_{t})+\beta E_{t}\left[V_{t+1}(x_{t+1},\Omega_{t+1},s)|x_{t},\Omega_{t},s,a_{t}=j\right]+\epsilon_{t}(a_{t}),
\end{eqnarray*}
where $\epsilon_{t}(a_{t})$ denotes idiosyncratic utility shocks. $\beta$ denotes agents' discount factor. $x_{t}$ denotes individual state variables, such as the durable goods holding of the agent. $\Omega_{t}$ denotes the market-level state variables, such as economic conditions. $\mathcal{A}_{t}(x_{t})$ denotes the consideration set of agents at individual state $x_{t}$ at time $t$.

Discounted sum of utility of an agent at observed state $(x_{t},\Omega_{t})$ and persistent unobserved state $s$ when choosing reference choice $0$ at time $t$ is:

\begin{eqnarray*}
v_{0t}(x_{t},\Omega_{t},s) & =\overline{\pi}(x_{t},\Omega_{t},a_{t}=0,s,\theta)+ & \beta E_{t}\left[V_{t+1}(x_{t+1},\Omega_{t+1},s)|x_{t},\Omega_{t},s,a_{t}=j\right]+\epsilon_{t}(a_{t}).
\end{eqnarray*}

Let $\widehat{p}_{t}(a_{t}=j|x_{t})$ be the observed ratio of consumers at state $x$ choosing the alternative $j$ at time $t$. We assume the value is nonparametrically estimated in the first stage. Besides, let $e_{kt}(x_{t})\equiv E_{t}\left[V_{t+1}(x_{t+1},\Omega_{t+1},s)|x_{t},\Omega_{t},s,a_{t}=k\right]-E_{x}\left[V_{t+1}(x_{t+1},\Omega_{t+1}^{(data)},s)|x_{t},s,a_{t}=k\right]\ (k\in\mathcal{J}_{t}\cup\{0\})$ be the expectation error when choosing alternative $k$ at time $t$ and state $x_{t}$. $\Omega_{t+1}^{(data)}$ denotes the realized market-level state variables at time $t+1$. We assume there is no expectation errors after the terminal period.

Here, we assume that $\epsilon(a)$ follows i.i.d. mean zero type-I extreme value distribution. In addition, we assume $\widehat{p_{t}}(a_{t}=j|x_{t})$ derived from the data is equal to the counterpart of the structural model, as in the BLP models. Then the following equations hold:

{\scriptsize{}
\begin{eqnarray*}
\widehat{p_{t}}(a_{t}=j|x_{t}) & = & \frac{\sum_{s}w_{s}Pr_{t}(x_{t}|s)\cdot\frac{\exp\left(\overline{\pi}(x_{t},\Omega_{t}^{(data)},a_{t}=j,s,\theta)+\xi_{jt}(x_{t})+\beta E_{x}\left[V_{t+1}(x_{t+1},\Omega_{t+1}^{(data)},s)|x_{t},s,a_{t}=j\right]+\beta e_{jt}(x_{t})\right)}{\exp\left(V_{t}(x_{t},\Omega_{t}^{(data)},s)\right)}}{\sum_{s}w_{s}Pr_{t}(x_{t}|s)},\\
\widehat{p_{t}}(a_{t}=0|x_{t}) & = & \frac{\sum_{s}w_{s}Pr_{t}(x_{t}|s)\cdot\frac{\exp\left(\overline{\pi}(x_{t},\Omega_{t}^{(data)},a_{t}=0,s,\theta)+\beta E_{x}\left[V_{t+1}(x_{t+1},\Omega_{t+1}^{(data)},,s)|x_{t},s,a_{t}=0\right]+\beta e_{0t}(x_{t})\right)}{\exp\left(V(x_{t},\Omega_{t}^{(data)},s)\right)}}{\sum_{s}w_{s}Pr_{t}(x_{t}|s)},\\
V_{t}(x_{t},\Omega_{t}^{(data)},s) & = & \log\left(\exp\left(\overline{\pi}(x_{t},\Omega_{t}^{(data)},a_{t}=0,s,\theta)+\beta E_{x}\left[V_{t+1}(x_{t+1},\Omega_{t+1}^{(data)},s)|x_{t},s,a_{t}=0\right]+\beta e_{kt}(x_{t})\right)+\right.\\
 &  & \left.\sum_{j\in\mathcal{J}_{t}}\exp\left(\overline{\pi}(x_{t},\Omega_{t}^{(data)},a_{t}=j,s,\theta)+\xi_{jt}(x_{t})+\beta E_{x}\left[V_{t+1}(x_{t+1},\Omega_{t+1}^{(data)},s)|x_{t},s,a_{t}=j\right]+\beta e_{0t}(x_{t})\right)\right),
\end{eqnarray*}
}where $Pr_{t}(x_{t}|s)$ denotes the probability that consumers at persistent unobserved state $s$ is at observed state $x_{t}$ at time $t$. 

Here, suppose state transitions of $x$ are not consumer-type specific. Then, by choosing $\eta_{t}(x_{t})$ so that $-\beta E_{x}\left[\eta_{t+1}(x_{t+1})|x_{t},s,a_{t}=0\right]+\beta e_{0t}(x_{t})+\eta_{t}(x_{t})=0$ and by defining $\widehat{V_{t}}(x,s)\equiv V_{t}(x,s)+\eta_{t}(x)$, $\widehat{\xi_{jt}}(x_{t})\equiv\xi_{jt}(x_{t})-\beta E_{x}\left[\eta_{t+1}(x_{t+1})|x_{t},s,a_{t}=j\right]+\beta e_{jt}(x_{t})+\eta_{t}(x)$, we have the following equations:

{\scriptsize{}
\begin{eqnarray*}
\widehat{V_{t}}(x_{t},\Omega_{t}^{(data)},s) & = & \log\left(\exp\left(\pi_{0t}(x_{t},s,\theta)+\beta E_{x}\left[\widehat{V_{t+1}}(x_{t+1},s)|x_{t},s,a_{t}=0\right]\right)+\right.\\
 &  & \ \ \ \ \left.\sum_{j\in\mathcal{J}_{t}}\exp\left(\pi_{jt}(x_{t},s,\theta)+\widehat{\xi_{jt}}(x_{t})+\beta E_{x}\left[\widehat{V_{t+1}}(x_{t+1},s)|x_{t},s,a_{t}=j\right]\right)\cdot\left(\frac{p_{t}(a_{t}=j|x_{t})}{s_{0t}(x_{t},\widehat{V})}\right)^{\gamma}\right)\\
 & \equiv & \Psi_{V\widehat{\xi}\rightarrow V,x_{t},s}(\widehat{V},\widehat{\xi};\theta),
\end{eqnarray*}

\begin{eqnarray*}
\widehat{p_{t}}(a_{t}=j|x_{t}) & = & \frac{\sum_{s}w_{s}Pr_{t}(x_{t}|s)\cdot\frac{\exp\left(\overline{\pi}(x_{t},\Omega_{t}^{(data)},a_{t}=j,s,\theta)+\widehat{\xi_{jt}}(x_{t})+\beta E_{x}\left[\widehat{V_{t+1}}(x_{t+1},\Omega_{t+1}^{(data)},s)|x_{t},s,a_{t}=j\right]\right)}{\exp\left(\widehat{V_{t}}(x_{t},\Omega_{t}^{(data)},s)\right)}}{\sum_{s}w_{s}Pr_{t}(x_{t}|s)},\\
\widehat{p_{t}}(a_{t}=0|x_{t}) & = & \frac{\sum_{s}w_{s}Pr_{t}(x_{t}|s)\cdot\frac{\exp\left(\overline{\pi}(x_{t},\Omega_{t}^{(data)},a_{t}=0,s,\theta)+\beta E_{x}\left[\widehat{V_{t+1}}(x_{t+1},\Omega_{t+1}^{(data)},s)|x_{t},s,a_{t}=0\right]\right)}{\exp\left(\widehat{V_{t}}(x_{t},\Omega_{t}^{(data)},s)\right)}}{\sum_{s}w_{s}Pr_{t}(x_{t}|s)}\equiv s_{0t}(x_{t},\widehat{V}),
\end{eqnarray*}

\begin{eqnarray*}
\widehat{\xi_{jt}}(x_{t}) & = & \log\left(\widehat{p}_{t}(a_{t}=j|x_{t})\right)-\log\left(\sum_{s}w_{s}Pr_{t}(x_{t}|s)\cdot\frac{\exp\left(\overline{\pi}(x_{t},\Omega_{t}^{(data)},a_{t}=j,s,\theta)+\beta E_{x}\left[\widehat{V_{t+1}}(x_{t+1},\Omega_{t+1}^{(data)},s)|x_{t},s,a_{it}=j\right]\right)}{\exp\left(\widehat{V_{t}}(x_{t},\Omega_{t}^{(data)},s)\right)}\right)\\
 &  & +\log\left(\sum_{s}w_{s}Pr_{t}(x_{t}|s)\right)\\
 & \equiv & \iota_{\widehat{V}\rightarrow\widehat{\xi}}(\widehat{V};\theta).
\end{eqnarray*}
}{\footnotesize\par}

The equations above imply that given parameters $\theta$, we can solve for $\widehat{V}$ by iteratively applying the mapping $\Phi^{\widehat{V}\gamma}(\widehat{V};\theta)\equiv\Psi_{V\widehat{\xi}\rightarrow V,x_{t},s}(\widehat{V},\iota_{\widehat{V}\rightarrow\widehat{\xi}}(\widehat{V};\theta);\theta)$, and we can recover $\widehat{\xi}$ in the process.

Here, we assume $E[\xi|Z]=E[e|Z]=0$. Then, $E[\widehat{\xi}|Z]=0$ also holds, and if the appropriate identification conditions are satisfied and $\widehat{p_{t}}(a_{t}|x_{t})$ are observed in the data, we can estimate $\theta$ by Algorithm \ref{alg:DDC_algorithm}.

\begin{algorithm}[H]
\begin{enumerate}
\item Inner loop: Given parameter values $\theta$,
\begin{enumerate}
\item Set initial values of $\widehat{V}^{(0)}$. Iterate the following $(n=0,1,2,\cdots)$: 
\begin{enumerate}
\item Compute $\widehat{\xi}^{(n)}=\iota_{\widehat{V}\rightarrow\widehat{\xi}}^{\gamma}\left(\widehat{V}^{(n)};\theta\right)$
\item Update $\widehat{V}$ by $\widehat{V}^{(n+1)}=\Psi_{\widehat{V}\widehat{\xi}\rightarrow\widehat{V}}(\widehat{V}^{(n)},\widehat{\xi}^{(n)};\theta)$
\item Exit the iteration if $\left\Vert \widehat{V}^{(n+1)}-\widehat{V}^{(n)}\right\Vert <\epsilon_{\widehat{V}}$
\end{enumerate}
\item Compute GMM objective using $\widehat{\xi}$ based on $E[\widehat{\xi}|Z]=0$
\end{enumerate}
\item Outer loop: Search for $\theta$ minimizing the GMM objective
\end{enumerate}
\caption{Algorithm applicable to the DDC model with unobserved payoffs and unobserved heterogeneity\label{alg:DDC_algorithm}}
\end{algorithm}

Typically, the values of $Pr_{t}(x_{t}|s)$ are unknown. Hence, we need to impose assumptions on the form of $Pr_{t}(x_{t}|s)$ at the initial period and also solve for the variables (See the discussion in \citet{kasahara2009nonparametric} and \citet{arcidiacono2011conditional}). Because $Pr_{t}(x_{t}|s)$ can be represented as a function of $\widehat{\xi}$ and $\widehat{V}$, in principle, we can alternatively represent $Pr_{t}(x_{t}|s)$ as a function of $\widehat{V}$, and incorporate it in the algorithm.

%% file: BLP_algorithm_paper_DP_202504.bbl
\begin{thebibliography}{}

\bibitem[Anderson, 1965]{anderson1965iterative}
Anderson, D.~G. (1965).
\newblock Iterative procedures for nonlinear integral equations.
\newblock {\em Journal of the ACM (JACM)}, 12(4):547--560.

\bibitem[Arcidiacono and Miller, 2011]{arcidiacono2011conditional}
Arcidiacono, P. and Miller, R.~A. (2011).
\newblock Conditional choice probability estimation of dynamic discrete choice
  models with unobserved heterogeneity.
\newblock {\em Econometrica}, 79(6):1823--1867.

\bibitem[Berry et~al., 1995]{berry1995automobile}
Berry, S., Levinsohn, J., and Pakes, A. (1995).
\newblock Automobile prices in market equilibrium.
\newblock {\em Econometrica}, 63(4):841--890.

\bibitem[Berry et~al., 1999]{berry1999voluntary}
Berry, S., Levinsohn, J., and Pakes, A. (1999).
\newblock Voluntary export restraints on automobiles: Evaluating a trade
  policy.
\newblock {\em American Economic Review}, 89(3):400--431.

\bibitem[Berry and Pakes, 2007]{berry2007pure}
Berry, S. and Pakes, A. (2007).
\newblock The pure characteristics demand model.
\newblock {\em International Economic Review}, 48(4):1193--1225.

\bibitem[Berry, 1994]{berry1994estimating}
Berry, S.~T. (1994).
\newblock Estimating discrete-choice models of product differentiation.
\newblock {\em The RAND Journal of Economics}, 25(2):242--262.

\bibitem[Bonnet et~al., 2022]{bonnet2022yogurts}
Bonnet, O., Galichon, A., Hsieh, Y.-W., O'hara, K., and Shum, M. (2022).
\newblock Yogurts choose consumers? estimation of random-utility models via
  two-sided matching.
\newblock {\em The Review of Economic Studies}, 89(6):3085--3114.

\bibitem[Brunner et~al., 2017]{brunner2017reliable}
Brunner, D., Heiss, F., Romahn, A., and Weiser, C. (2017).
\newblock {\em Reliable estimation of random coefficient logit demand models}.
\newblock Number 267. DICE Discussion Paper.

\bibitem[Conlon and Gortmaker, 2020]{conlon2020best}
Conlon, C. and Gortmaker, J. (2020).
\newblock {Best practices for differentiated products demand estimation with
  PyBLP}.
\newblock {\em The RAND Journal of Economics}, 51(4):1108--1161.

\bibitem[Conlon, 2012]{conlon2012dynamic}
Conlon, C.~T. (2012).
\newblock {A dynamic model of prices and margins in the LCD TV industry}.
\newblock {\em mimeo, Columbia University}, 80:1433--1504.

\bibitem[Doi, 2022]{doi2022simple}
Doi, N. (2022).
\newblock A simple method to estimate discrete-type random coefficients logit
  models.
\newblock {\em International Journal of Industrial Organization}, 81:102825.

\bibitem[Dub{\'e} et~al., 2012]{dube2012improving}
Dub{\'e}, J.-P., Fox, J.~T., and Su, C.-L. (2012).
\newblock Improving the numerical performance of static and dynamic aggregate
  discrete choice random coefficients demand estimation.
\newblock {\em Econometrica}, 80(5):2231--2267.

\bibitem[Duch-Brown et~al., 2023]{duch2023evaluating}
Duch-Brown, N., Grzybowski, L., Romahn, A., and Verboven, F. (2023).
\newblock Evaluating the impact of online market integration - evidence from
  the eu portable pc market.
\newblock {\em American Economic Journal: Microeconomics}, 15(4):268--305.

\bibitem[Fang and Saad, 2009]{fang2009two}
Fang, H.-r. and Saad, Y. (2009).
\newblock Two classes of multisecant methods for nonlinear acceleration.
\newblock {\em Numerical linear algebra with applications}, 16(3):197--221.

\bibitem[Fukasawa, 2025]{fukasawa2024lightbulb}
Fukasawa, T. (2025).
\newblock {When do firms sell high durability products? The case of Light Bulb
  Industry}.
\newblock {\em arXiv preprint arXiv:2503.23792}.

\bibitem[Goeree, 2008]{goeree2008limited}
Goeree, M.~S. (2008).
\newblock Limited information and advertising in the us personal computer
  industry.
\newblock {\em Econometrica}, 76(5):1017--1074.

\bibitem[Gowrisankaran and Rysman, 2012]{gowrisankaran2012dynamics}
Gowrisankaran, G. and Rysman, M. (2012).
\newblock Dynamics of consumer demand for new durable goods.
\newblock {\em Journal of Political Economy}, 120(6):1173--1219.

\bibitem[Grigolon and Verboven, 2014]{grigolon2014nested}
Grigolon, L. and Verboven, F. (2014).
\newblock {Nested logit or random coefficients logit? A comparison of
  alternative discrete choice models of product differentiation}.
\newblock {\em Review of Economics and Statistics}, 96(5):916--935.

\bibitem[Hendel and Nevo, 2006]{hendel2006measuring}
Hendel, I. and Nevo, A. (2006).
\newblock Measuring the implications of sales and consumer inventory behavior.
\newblock {\em Econometrica}, 74(6):1637--1673.

\bibitem[Huang and Wan, 2017]{huang2017new}
Huang, S. and Wan, Z. (2017).
\newblock A new nonmonotone spectral residual method for nonsmooth nonlinear
  equations.
\newblock {\em Journal of Computational and Applied Mathematics}, 313:82--101.

\bibitem[Igami, 2017]{igami2017estimating}
Igami, M. (2017).
\newblock Estimating the innovator's dilemma: Structural analysis of creative
  destruction in the hard disk drive industry, 1981--1998.
\newblock {\em Journal of Political Economy}, 125(3):798--847.

\bibitem[Iizuka, 2007]{iizuka2007experts}
Iizuka, T. (2007).
\newblock Experts' agency problems: evidence from the prescription drug market
  in japan.
\newblock {\em The RAND journal of economics}, 38(3):844--862.

\bibitem[Jamshidian and Jennrich, 1997]{jamshidian1997acceleration}
Jamshidian, M. and Jennrich, R.~I. (1997).
\newblock {Acceleration of the EM algorithm by using quasi-Newton methods}.
\newblock {\em Journal of the Royal Statistical Society: Series B (Statistical
  Methodology)}, 59(3):569--587.

\bibitem[Judd, 1998]{judd1998numerical}
Judd, K.~L. (1998).
\newblock {\em Numerical methods in economics}.
\newblock MIT press.

\bibitem[Kalouptsidi, 2012]{kalouptsidi2012market}
Kalouptsidi, M. (2012).
\newblock From market shares to consumer types: Duality in differentiated
  product demand estimation.
\newblock {\em Journal of Applied Econometrics}, 27(2):333--342.

\bibitem[Kalouptsidi et~al., 2020]{kalouptsidi2020linear}
Kalouptsidi, M., Scott, P.~T., and Souza-Rodrigues, E. (2020).
\newblock Linear iv regression estimators for structural dynamic discrete
  choice models.
\newblock {\em Journal of Econometrics}.

\bibitem[Kasahara and Shimotsu, 2009]{kasahara2009nonparametric}
Kasahara, H. and Shimotsu, K. (2009).
\newblock Nonparametric identification of finite mixture models of dynamic
  discrete choices.
\newblock {\em Econometrica}, 77(1):135--175.

\bibitem[La~Cruz et~al., 2006]{la2006spectral}
La~Cruz, W., Mart{\'\i}nez, J., and Raydan, M. (2006).
\newblock Spectral residual method without gradient information for solving
  large-scale nonlinear systems of equations.
\newblock {\em Mathematics of computation}, 75(255):1429--1448.

\bibitem[Lee and Seo, 2015]{lee2015computationally}
Lee, J. and Seo, K. (2015).
\newblock A computationally fast estimator for random coefficients logit demand
  models using aggregate data.
\newblock {\em The RAND Journal of Economics}, 46(1):86--102.

\bibitem[Lee and Seo, 2016]{lee2016revisiting}
Lee, J. and Seo, K. (2016).
\newblock Revisiting the nested fixed-point algorithm in blp random
  coefficients demand estimation.
\newblock {\em Economics Letters}, 149:67--70.

\bibitem[Miranda and Fackler, 2004]{miranda2004applied}
Miranda, M.~J. and Fackler, P.~L. (2004).
\newblock {\em Applied computational economics and finance}.
\newblock MIT press.

\bibitem[Nevo, 2001]{nevo2001measuring}
Nevo, A. (2001).
\newblock Measuring market power in the ready-to-eat cereal industry.
\newblock {\em Econometrica}, 69(2):307--342.

\bibitem[Pakes and McGuire, 1994]{Pakes_McGuire1994}
Pakes, A. and McGuire, P. (1994).
\newblock {Computing Markov-Perfect Nash Equilibria: Numerical Implications of
  a Dynamic Differentiated Product Model}.
\newblock {\em The RAND Journal of Economics}, 25(4):555.

\bibitem[P{\'a}l and S{\'a}ndor, 2023]{pal2023comparing}
P{\'a}l, L. and S{\'a}ndor, Z. (2023).
\newblock Comparing procedures for estimating random coefficient logit demand
  models with a special focus on obtaining global optima.
\newblock {\em International Journal of Industrial Organization}, 88:102950.

\bibitem[Reynaerts et~al., 2012]{reynaerts2012enhencing}
Reynaerts, J., Varadha, R., and Nash, J.~C. (2012).
\newblock {Enhencing the convergence properties of the BLP (1995) contraction
  mapping}.

\bibitem[Salani{\'e} and Wolak, 2022]{salanie2022fast}
Salani{\'e}, B. and Wolak, F.~A. (2022).
\newblock Fast, detail-free, and approximately correct: Estimating mixed demand
  systems.
\newblock Technical report, Working paper.

\bibitem[Schiraldi, 2011]{schiraldi2011automobile}
Schiraldi, P. (2011).
\newblock Automobile replacement: a dynamic structural approach.
\newblock {\em The RAND journal of economics}, 42(2):266--291.

\bibitem[Shcherbakov, 2016]{shcherbakov2016measuring}
Shcherbakov, O. (2016).
\newblock Measuring consumer switching costs in the television industry.
\newblock {\em The RAND Journal of Economics}, 47(2):366--393.

\bibitem[Sun and Ishihara, 2019]{Sun2019}
Sun, Y. and Ishihara, M. (2019).
\newblock {A computationally efficient fixed point approach to dynamic
  structural demand estimation}.
\newblock {\em Journal of Econometrics}, 208(2):563--584.

\bibitem[Varadhan and Roland, 2008]{varadhan2008simple}
Varadhan, R. and Roland, C. (2008).
\newblock {Simple and globally convergent methods for accelerating the
  convergence of any EM algorithm}.
\newblock {\em Scandinavian Journal of Statistics}, 35(2):335--353.

\bibitem[Zhang et~al., 2020]{zhang2020globally}
Zhang, J., O'Donoghue, B., and Boyd, S. (2020).
\newblock Globally convergent type-i anderson acceleration for nonsmooth
  fixed-point iterations.
\newblock {\em SIAM Journal on Optimization}, 30(4):3170--3197.

\end{thebibliography}
